\documentclass[11pt,halfline,a4paper]{ouparticle}
\usepackage{amssymb}
\usepackage{amsthm, color}
\usepackage{amsmath}
\usepackage{multicol}
\usepackage{multirow}
\usepackage{algorithm} 
\usepackage{amsmath}
\usepackage{algpseudocode} 
\usepackage{subfigure}
\usepackage[authoryear,round]{natbib}
\begin{document}
\setstretch{1.15}
\title{
Anomaly detection in autoregressive networks}

\author{%
\name{Xianghe Zhu}
\address{Department of Statistics, London School of Economics and Political \\
Science, London, WC2A 2AE, United Kingdom
}
\email{X.Zhu58@lse.ac.uk}
}

\abstract{
We study anomaly detection in temporally dependent network sequences. Methods based only on adjacency matrices, which are widely used for static networks, can miss changes in the way edges evolve. We instead represent each pair of consecutive networks by separate formation and dissolution event matrices, and embed the resulting matrix sequences using unfolded adjacency spectral embedding. Comparing these embeddings against a stationary baseline yields vertex- and network-level statistics that detect an anomalous transition and identify whether it involves formation, dissolution, or both. For autoregressive random dot product graphs, we establish uniform rowwise consistency of the transition-event embeddings and derive high-probability guarantees for vertex- and network-level detection and exact recovery of the anomalous vertex set. The theory separates a vertex's own displacement from interference caused by other vertices and from autoregressive memory. We quantify the memory left by an earlier anomaly, show that it decays geometrically after the transition mechanism returns to baseline, and give conditions under which it is negligible. A degree-corrected stochastic block model extension gives exact community recovery and provides community-reassignment, centre-shift, split, and merge anomaly statistics with high-probability detection and recovery guarantees. Simulations and applications to an international trade dataset and a primary-school contact network illustrate the performance of the proposed method, revealing a dissolution-driven trade decline followed by formation-driven recovery during COVID-19 and temporary merge-type mixing between school classes.
}



\maketitle

\section{Introduction}

Sequences of networks occur in communication, trade, transport, and social-contact systems. Their evolution is often stable for long periods, but anomalies may occur. At a given time, an adjacency matrix records which edges are present. A dynamic network, however, may evolve through the formation of new edges and the dissolution of existing edges, and an anomaly may affect one or both mechanisms. A useful method should therefore not only detect when an anomaly occurs and identify the affected vertices, but also determine whether formation, dissolution, or both have changed and assess when the network returns to its usual evolution.

Network anomaly detection has been studied from several perspectives. Existing approaches include Bayesian count-process models \citep{10.1214/10-AOAS329}, scan statistics for locally unusual subgraphs \citep{priebe_enron_2005,neil_scan_2013,wang2013locality}, and dissimilarity methods for network-valued data \citep{josephs2023bayesian}. Spectral methods compare low-dimensional representations of adjacency matrices. In particular, \citet{chen_multiple_2025} compare consecutive embeddings to find anomalous networks and vertices. Their theoretical guarantee is network-level, based on a Frobenius-norm comparison of the embedding matrices, whereas their vertex-level procedure is algorithmic and has no corresponding vertex-level guarantee. A consecutive comparison can also produce a large statistic both when an anomaly begins and when the process returns to normal, and their population model does not describe temporal dependence. Embedding-based change-point methods allow temporal dependence \citep{madrid_padilla_2022,JMLR:v24:22-0845,athreya2025euclidean}. However, they typically consider a change in the network mechanism that continues over time, rather than a temporary departure from a  known baseline.

In this paper, we develop a transition-event spectral method for anomaly detection in autoregressive networks \citep{JMLR:v24:22-0845}. We treat each pair of consecutive networks as a single transition, represent it by separate formation and dissolution event matrices, and embed them using unfolded adjacency spectral embedding (UASE) \citep{jones_multilayer_2021}. This transition-level representation is able to indicate whether an anomaly occurs in formation, dissolution, or both. Comparing each target with a candidate stationary baseline allows us to determine whether the transition mechanism is anomalous at that time or has returned to baseline. The method also applies in the case where the formation and dissolution mechanisms change in such a way that the marginal connection probabilities remain unchanged. This property shows that the transition-event representation retains dynamic information that is absent from marginal connection probabilities.

For autoregressive random dot product graphs, we establish uniform rowwise consistency of the transition-event right embeddings. This yields vertex- and network-level anomaly detection guarantees and exact recovery of anomalous vertices.
The vertex result separates the vertex’s own latent displacement from cross-vertex interference and autoregressive memory. When the set of anomalous vertices is sufficiently sparse, the interference is asymptotically negligible.  We also show how an earlier anomaly can still affect later statistics after the mechanism returns to baseline, and give a burn-in time after which this effect is no larger than spectral estimation error. Empirical \(p\)-values are computed by comparing each target statistic with statistics from all pairs of baseline transitions.

We further extend the framework to autoregressive degree-corrected stochastic block models. A joint embedding formed by concatenating the separately row-normalised formation and dissolution right embeddings gives exact community recovery and supports exact recovery of anomalous vertex reassignments. We also construct centre-shift, split, and merge statistics with high-probability anomaly detection guarantees.

The new contributions of this paper include the following. First, the transition-event representation provides detailed information about network anomalies, including when an anomaly occurs, which vertices are affected, whether formation, dissolution, or both have changed, and when the transition mechanism returns to baseline. Second, the random dot product graph theory links uniform rowwise consistency to vertex- and network-level anomaly detection and exact recovery of anomalous vertices, while making explicit the roles of cross-vertex interference and autoregressive memory. Third, for the degree-corrected stochastic block model, we establish exact-recovery guarantees for community partitions and the community reassignments of anomalous vertices, together with detection guarantees for community-level anomalies including centre shifts, splits, and merges.
Simulations examine anomaly detection at the vertex, network, and community levels, as well as community recovery. The trade application identifies whether anomalies arise from formation or dissolution and which countries are flagged as anomalous, while the primary-school application detects temporary mixing between pairs of classes.

The rest of the paper is organised as follows. Section~2 introduces the model and transition-event embedding, establishing identifiability and uniform rowwise consistency. Section~3 develops the vertex- and
network-level anomaly statistics and analyses memory under the isolated anomaly process. Section~4
introduces the autoregressive degree-corrected stochastic block model, community recovery
and the four community anomaly statistics. Section~5 describes how empirical \(p\)-values are computed and how
baseline stationarity is checked. Sections 6 and 7 report the simulation studies and the two real-data analyses, respectively. The
Supplementary Material contains complete proofs, algorithms, additional
simulations, and robustness analyses.

	\section{Transition-event embedding}
	
	\subsection{Autoregressive random dot product graph and identifiability}
	
	Let $\mathbf X^{(t)}\in\{0,1\}^{p\times p}$ denote the adjacency matrix at time
	$t \in \{0,\ldots,n\}$ of an undirected network without self-loops, so that $X_{ij}^{(t)}=X_{ji}^{(t)}, X^{(t)}_{ii}=0$ for $i,j\in[p]$. We first specify the edgewise transition mechanism using the autoregressive network model
	\citep{JMLR:v24:22-0845}. 
	
	\begin{definition}[AR(1) Network]\label{def:ar1-network}
		The sequence $\{\mathbf X^{(t)}:t\in\{0,\ldots,n\}\}$ is an autoregressive network
	if, for every $i<j$ and $t\in[n]$, $$
	X_{ij}^{(t)}=X_{ij}^{(t-1)} I\left(\varepsilon_{ij}^{(t)}=0\right)+I\left(\varepsilon_{ij}^{(t)}=1\right), \quad t \geq 1,
	$$
	where the innovations
	$\{\varepsilon_{ij}^{(t)}:i<j\}$ are independent, and
	$$
	\mathbb P\left(\varepsilon_{ij}^{(t)}=1\right)=\alpha_{ij}^{(t)},\quad
	\mathbb P\left(\varepsilon_{ij}^{(t)}=-1\right)=\beta_{ij}^{(t)},\quad
	\mathbb P\left(\varepsilon_{ij}^{(t)}=0\right)=1-\alpha_{ij}^{(t)}-\beta_{ij}^{(t)}.
	$$
	Here, \(\alpha_{ij}^{(t)}\) and \(\beta_{ij}^{(t)}\) are the transition
probabilities governing edge formation and dissolution, respectively, and satisfy
	$\alpha_{ij}^{(t)},\beta_{ij}^{(t)}\geq0$ and
$\alpha_{ij}^{(t)}+\beta_{ij}^{(t)}\leq1$.  The initial edge states
	$\{X_{ij}^{(0)}:i<j\}$ are independent across dyads and
	independent of the innovation sequences.
	\end{definition}

	For all $i<j$, if
	$\alpha_{ij}^{(t)} \equiv \alpha_{ij}$ and
	$\beta_{ij}^{(t)} \equiv \beta_{ij}$ for all $t \geq 1$, and the initial connection probabilities satisfy 	$\pi_{ij}=\alpha_{ij}/(\alpha_{ij}+\beta_{ij})$  with
	\(\alpha_{ij}+\beta_{ij}>0\), then the process is strictly stationary with stationary connection probabilities $\pi_{ij}$. 
    
	 The transition probabilities can be parameterised by inner
	products of latent position vectors and by global sparsity factors.  This gives
	the following AR(1) random dot product graph model.  Background on random dot
	product graphs and adjacency spectral embedding may be found in
	\citet{young2007random}, \citet{sussman2012consistent}, and
	\citet{athreya2018statistical}.
	
	\begin{definition}[AR(1)-RDPG]\label{def:ar1-trdpg}
		Fix a dimension $d\ge 1$ and sparsity factors $\rho_Y,\rho_Z\in(0,1]$.
		Let $\mathcal L_Y^{(1)},\ldots,\mathcal L_Y^{(n)}\subset\mathbb R^d$ and
		$\mathcal L_Z^{(1)},\ldots,\mathcal L_Z^{(n)}\subset\mathbb R^d$ be bounded sets such that for some finite constant $\ell>0$, 
		$\|y\|_{\infty},\|z\|_{\infty} \leq \ell, 
		y^\top y'\in[0,1], z^\top z'\in[0,1]$ and $\rho_Yy^\top y'+\rho_Zz^\top z'\in[0,1]$ for any $y,y'\in\mathcal L_Y^{(t)}, z,z'\in\mathcal L_Z^{(t)}, t\in[n].$
		Fix latent position vectors $\boldsymbol\ell_{Y,i}^{(t)}\in\mathcal L_Y^{(t)}$ and
		$\boldsymbol\ell_{Z,i}^{(t)}\in\mathcal L_Z^{(t)}$. They are collected in
		the rows of
		$\mathbf L_Y^{(t)}
		= [\boldsymbol\ell_{Y,1}^{(t)}|\boldsymbol\ell_{Y,2}^{(t)}|\ldots|
		\boldsymbol\ell_{Y,p}^{(t)}]^{\top}\in \mathbb{R}^{p\times d}$ and
		$\mathbf L_Z^{(t)}
		= [\boldsymbol\ell_{Z,1}^{(t)}|\boldsymbol\ell_{Z,2}^{(t)}|\ldots|
		\boldsymbol\ell_{Z,p}^{(t)}]^{\top}\in \mathbb{R}^{p\times d}$.
		Then the transition probabilities in Definition~\ref{def:ar1-network} are parameterised by
		\begin{equation}\label{eq:alpha-beta-ip}
			\alpha_{ij}^{(t)}=\rho_Y(\boldsymbol\ell_{Y,i}^{(t)})^\top\boldsymbol\ell_{Y,j}^{(t)},
			\qquad
			\beta_{ij}^{(t)}=\rho_Z(\boldsymbol\ell_{Z,i}^{(t)})^\top\boldsymbol\ell_{Z,j}^{(t)}.
		\end{equation}
        \end{definition}
	Under the parameterisation in \eqref{eq:alpha-beta-ip}, if the latent positions remain unchanged over time, the formation and
dissolution transition probabilities are time-invariant. If the initial
network follows the corresponding stationary distribution, the
AR(1)-RDPG process is stationary. The converse may not hold, since
the same transition-probability matrices can have different
latent-position representations.
	
The adjacency matrix \(\mathbf X^{(t)}\) records the current edge
states but not the edge formations or dissolutions since \(t-1\).
We therefore define the formation and
dissolution event matrices
\(\mathbf Y^{(t)}=(Y_{ij}^{(t)})\) and
\(\mathbf Z^{(t)}=(Z_{ij}^{(t)})\), where
\[
Y_{ij}^{(t)}
=
(1-X_{ij}^{(t-1)})X_{ij}^{(t)},
\qquad
Z_{ij}^{(t)}
=
X_{ij}^{(t-1)}(1-X_{ij}^{(t)}).
\]
Thus, \(Y_{ij}^{(t)}=1\) records an edge formation, whereas
\(Z_{ij}^{(t)}=1\) records an edge dissolution. Conditional on
\(X_{ij}^{(t-1)}=0\), \(Y_{ij}^{(t)}\) is Bernoulli with parameter
\(\alpha_{ij}^{(t)}\); conditional on \(X_{ij}^{(t-1)}=1\),
\(Z_{ij}^{(t)}\) is Bernoulli with parameter
\(\beta_{ij}^{(t)}\). 
Thus, the \(Y\) and \(Z\) channels provide separate information about
the formation and dissolution probabilities and can distinguish
transition mechanisms with the same stationary marginal connection
probabilities.

    Under this stationary model, the unconditional formation and
dissolution probabilities are
\[
{\mathbb E}(Y_{ij}^{(t)})
=
\frac{\alpha_{ij}\beta_{ij}}
{\alpha_{ij}+\beta_{ij}},
\qquad
{\mathbb E}(Z_{ij}^{(t)})
=
\frac{\alpha_{ij}\beta_{ij}}
{\alpha_{ij}+\beta_{ij}}.
\]
The following example illustrates the
non-identifiability. Suppose the process is stationary up to time \(t-1\).
At time \(t\), replace
\((\alpha_{ij},\beta_{ij})\) by
\((c_{ij}\alpha_{ij},c_{ij}\beta_{ij})\), where
\(
c_{ij}>0,
c_{ij}\neq1,
c_{ij}(\alpha_{ij}+\beta_{ij})\leq1.
\)
This changes the transition mechanism but leaves the marginal
connection probability at time \(t\) equal to \(\pi_{ij}\). The
marginal adjacency distribution at time \(t\) is therefore unchanged
from the baseline, so the anomaly cannot be
identified from the marginal connection probabilities alone. By
contrast, both formation and dissolution event probabilities are
multiplied by \(c_{ij}\). The event populations can therefore
distinguish the two mechanisms.

	\subsection{Estimation of right embeddings}
	We work with a fixed number of transitions $n$ and let $p\to\infty$.  The next
	condition fixes the sparsity scaling used by the UASE consistency proof.
	
	\begin{condition}\label{ass:transition-sparsity}
		The sparsity factors either satisfy \(\inf_p\rho_Y,\inf_p\rho_Z>0\) in the dense regime or  converge to zero in the sparse
		regime.  In the sparse regime, we assume that as $p\to\infty$,
		$\sqrt p\,\rho_Y/\log p, \;\sqrt p\,\rho_Z/\log p\to\infty$.    
	\end{condition}
	The next condition fixes the initial constant-rank connection probability matrix 
    used by the UASE consistency proof.
	
	\begin{condition}\label{ass:initial-finite-rank}
		Assume that the initial connection probability matrix $\mathbf\Pi^{(0)}$ 
        has rank $r_0$, where $r_0$ is constant as $p\rightarrow\infty$. 
		Furthermore, it admits the decomposition
		$\mathbf\Pi^{(0)} = \sum_{r=1}^{r_{0}} \mathbf{c}_{r}^{(0)}
		\mathbf{d}_{r}^{(0) \top}$ with
		$\|\mathbf{c}_{r}^{(0)}\|_{\infty},
		\|\mathbf{d}_{r}^{(0)}\|_{\infty} \leq C_{0}$ for some finite constant $C_0$.
	\end{condition}
	Before stating the consistency result, we fix the population targets used for
	UASE. Stack the transition-event matrices as
$\mathbf Y=[\mathbf Y^{(1)}|\cdots|\mathbf Y^{(n)}]$ and
$\mathbf Z=[\mathbf Z^{(1)}|\cdots|\mathbf Z^{(n)}]$. For $i\ne j$, write
$\Pi_{ij}^{(t)}={\mathbb P}\{X_{ij}^{(t)}=1\}$ and let $\mathbf\Pi^{(t)}$ collect
these off-diagonal probabilities. Conditioning on the previous edge state gives
\begin{equation}
\begin{split}
\Pi_{ij}^{(t)}
=\alpha_{ij}^{(t)}+
 \{1-\alpha_{ij}^{(t)}-\beta_{ij}^{(t)}\}\Pi_{ij}^{(t-1)},
 \qquad i\ne j.
\end{split}
\label{eq:mean-recursion}
\end{equation}
For the population spectral analysis we complete the diagonal by the same
recursion,
\[
\Pi_{ii}^{(t)}=\rho_Y\|\boldsymbol\ell_{Y,i}^{(t)}\|_2^2+
\{1-\rho_Y\|\boldsymbol\ell_{Y,i}^{(t)}\|_2^2
-\rho_Z\|\boldsymbol\ell_{Z,i}^{(t)}\|_2^2\}\Pi_{ii}^{(t-1)}
\]
starting from \(\Pi^{(0)}_{ii}\) in
Condition~\ref{ass:initial-finite-rank}. 
	Define the transition-event population matrices
\[
\mathbf S_Y^{(t)}=\rho_Y\mathbf L_Y^{(t)}\mathbf L_Y^{(t)\top}
 \circ(\mathbf 1\mathbf 1^\top-\mathbf\Pi^{(t-1)}),\qquad
\mathbf S_Z^{(t)}=\rho_Z\mathbf L_Z^{(t)}\mathbf L_Z^{(t)\top}
 \circ\mathbf\Pi^{(t-1)},
\]
where $\circ$ denotes the Hadamard product. Their off-diagonal entries equal the unconditional means of
$\mathbf Y^{(t)}$ and $\mathbf Z^{(t)}$. Let \(\mathcal H=\{Y,Z\}\) index the formation and dissolution
channels, respectively, and write
$\mathbf S_h=[\mathbf S_h^{(1)}|\cdots|\mathbf S_h^{(n)}]$ for
$h\in\mathcal H$.
	
	The transition-event population matrices need not have rank $d$, as the previous-time connection
probabilities enter through the Hadamard products and can introduce additional
directions. Nevertheless, the rank remains bounded when the number of time points $n$ is fixed. 
	
	\begin{proposition}\label{prop:u1u2}
		In the AR(1)-RDPG model, under Condition~\ref{ass:initial-finite-rank}, there exist matrices
		$\overline{\mathbf U}_{h,1}\in\mathbb R^{p\times \overline K_h},\overline{\mathbf U}_{h,2}\in\mathbb R^{np\times \overline K_h}
		$
		such that
		$
		\mathbf S_h
		=
		\overline{\mathbf U}_{h,1}
		\overline{\mathbf U}_{h,2}^{\top}
		$
		for
$h\in\mathcal H$, where $\overline K_h$ has an explicit expression depending
		only on $n$, \(d\), and \(r_0\).  
	\end{proposition}
	
	Proposition~\ref{prop:u1u2} gives a worst-case 
bound on the ranks $K_h=\operatorname{rank}(\mathbf S_h)\leq \overline K_h$.
		The proof repeatedly expands \eqref{eq:mean-recursion} into sums of bounded
rank-one terms and uses that the Hadamard product of two finite-rank matrices
has rank at most the product of their ranks. 
The following condition states directly the singular-value scale  needed for the UASE consistency result.

	\begin{condition}\label{ass:3}
For
$h\in\mathcal H$, there exist constants
\(0<c_{h,S}\le C_{h,S}<\infty\) independent of \(p\), such that
\[
c_{h,S}p\rho_h
\le
\sigma_{K_h}(\mathbf S_h)
\le
\sigma_1(\mathbf S_h)
\le
C_{h,S}p\rho_h.
\]
\end{condition}

	Let the top-$K_Y$ singular value decompositions (SVDs) be
	$$
	\mathbf Y
	=
	\mathbf U_Y\boldsymbol\Sigma_Y\mathbf V_Y^\top
	+
	\mathbf U_{Y,\perp}\boldsymbol\Sigma_{Y,\perp}\mathbf V_{Y,\perp}^\top,
	\quad
	\mathbf S_Y
	=
	\mathbf U_{S_Y}\boldsymbol\Sigma_{S_Y}\mathbf V_{S_Y}^\top.
	$$
	The columns of
\(
\mathbf U_Y,\mathbf U_{S_Y}\in\mathbb R^{p\times K_Y}
\text{ and }
\mathbf V_Y,\mathbf V_{S_Y}\in\mathbb R^{np\times K_Y}
\)
are the left and right singular vectors corresponding to the \(K_Y\)
largest singular values of \(\mathbf Y\) and \(\mathbf S_Y\),
respectively. These singular values are contained in the diagonal
matrices
\(
\boldsymbol\Sigma_Y,\boldsymbol\Sigma_{S_Y}
\in\mathbb R^{K_Y\times K_Y}.
\)
The matrices
\(
\mathbf U_{Y,\perp}\in\mathbb R^{p\times(p-K_Y)}
\text{ and }
\mathbf V_{Y,\perp}\in\mathbb R^{np\times(p-K_Y)}
\)
contain the remaining left and right singular vectors of \(\mathbf Y\),
and their corresponding singular values are contained in
\(
\boldsymbol\Sigma_{Y,\perp}
\in\mathbb R^{(p-K_Y)\times(p-K_Y)}.
\)
	Similarly, let the top-$K_Z$ SVDs be
	$$
	\mathbf Z
	=
	\mathbf U_Z\boldsymbol\Sigma_Z\mathbf V_Z^\top
	+
	\mathbf U_{Z,\perp}\boldsymbol\Sigma_{Z,\perp}\mathbf V_{Z,\perp}^\top,
	\quad
	\mathbf S_Z
	=
	\mathbf U_{S_Z}\boldsymbol\Sigma_{S_Z}\mathbf V_{S_Z}^\top,
	$$
	where the corresponding singular vector matrices have $K_Z$ columns.

	For
$h\in\mathcal H$, define the right embeddings
	$$
	\mathbf R_h
	=
	\mathbf V_h\boldsymbol\Sigma_h^{1/2}
	=
	(\mathbf R_h^{(1)};\ldots;\mathbf R_h^{(n)}),
	\qquad
	\mathbf R_{S_h}
	=
	\mathbf V_{S_h}\boldsymbol\Sigma_{S_h}^{1/2}
	=
	(\mathbf R_{S_h}^{(1)};\ldots;\mathbf R_{S_h}^{(n)}),
	$$
	where each block belongs to $\mathbb R^{p\times K_h}$. 
	
	For a matrix $\mathbf A$, write
$\|\mathbf A\|_{2\to\infty}=\max_i\|\mathbf e_i^\top \mathbf A\|_2$. We then have the following consistency result.
	\begin{theorem}\label{thm1}
		Suppose Conditions
		\ref{ass:transition-sparsity}, \ref{ass:initial-finite-rank}, and
		\ref{ass:3} hold.  Then, for every fixed \(c_0>0\) and $h\in \mathcal{H}$, there are constants
		\(C_{h,c_0},C_{c_0}<\infty\) and orthogonal matrices
		$\mathbf W_{hS}\in\mathbb O(K_h)$ such that, with probability at least
		\(1-C_{c_0}p^{-c_0}\),
		\[
		\max_{1\le t\le n}
		\left\|
		\mathbf R_h^{(t)}
		-
		\mathbf R_{S_h}^{(t)}
		\mathbf W_{hS}
		\right\|_{2\to\infty}
		\le
		C_{h,c_0}\sqrt{\frac{\log p}{p}}.
		\]

		The alignment matrices may be taken from the singular value decompositions
		\[
		\mathbf{U}_{S_h}^{\top} \mathbf{U}_{h}
		+
		\mathbf{V}_{S_h}^{\top} \mathbf{V}_{h}
		=
		\mathbf{W}_{h,1} \boldsymbol{\Lambda}_h\mathbf{W}_{h,2}^{\top},
		\qquad
		\mathbf{W}_{hS}=\mathbf{W}_{h,1} \mathbf{W}_{h,2}^{\top}.
		\]
	\end{theorem}
	Since the network layers are temporally dependent, concentration
results for independent layers do not apply directly. The proof uses independence across dyads,
together with spectral perturbation and a leave-one-vertex-out
argument.  After a common orthogonal
alignment, differences between estimated rows approximate their
population counterparts with error
\(O(\sqrt{\frac{\log p}{p}})\) uniformly over times and vertices.
Complete
proofs are given in the Supplementary Material.

	\section{Theory for anomaly detection}
	 Consider a stationary baseline
	\(B=\{t_0,\ldots,t^*-1\}\), over which the network process is stationary
	and the formation and dissolution latent positions remain unchanged.  For the theoretical results, a target time \(t\ge t^\ast\) is compared with a single reference point in the baseline, chosen to be the endpoint \(t^\ast-1\). The one-step comparison is the
	special case \(t=t^*\). Let $\epsilon_p=\max_{h\in\mathcal H}C_{h,c_0}(\log p/p)^{1/2}$ denote the
rowwise error term in
Theorem~\ref{thm1}.

	\subsection{Vertex-level anomaly detection}

	For $h\in\mathcal H$, define the vertex-level channel-specific and joint test statistics
	\begin{equation}\label{eq:vertex-statistic}
	    T_{h,i}^{(t^*-1,t)}
		=
		\left\|\mathbf e_i^\top\bigl(\mathbf R_{h}^{(t)}-\mathbf R_h^{(t^*-1)}\bigr)\right\|_2,
	T_{i}^{(t^*-1,t)}=\max_hT_{h,i}^{(t^*-1,t)},
	\qquad i\in[p].
	\end{equation}
	For vertex \(i\), the joint statistic tests
\[
H_{0,i}^{(t^*-1,t)}:
\boldsymbol\ell_{Y,i}^{(t)}
=
\boldsymbol\ell_{Y,i}^{(t^*-1)}
\ \text{and}\
\boldsymbol\ell_{Z,i}^{(t)}
=
\boldsymbol\ell_{Z,i}^{(t^*-1)}
\qquad\text{against}\qquad
H_{1,i}^{(t^*-1,t)}:
\boldsymbol\ell_{Y,i}^{(t)}
\neq
\boldsymbol\ell_{Y,i}^{(t^*-1)}
\ \text{or}\
\boldsymbol\ell_{Z,i}^{(t)}
\neq
\boldsymbol\ell_{Z,i}^{(t^*-1)}.
\]
The statistics \(T_{h,i}^{(t^*-1,t)}\) in (\ref{eq:vertex-statistic}) test whether vertex \(i\) is anomalous in channel \(h\), that is, whether its latent position has changed.
	Define the vertex-level latent displacement  between time $t^*-1$ and $t$ as
	$
\delta_{h,i}^{(t^*-1,t)}
=
\left\|
\boldsymbol\ell_{h,i}^{(t)}
-
\boldsymbol\ell_{h,i}^{(t^*-1)}
\right\|_2.
$
	We further define its channel-specific anomalous-vertex set and
size by
\[
\mathcal I_h^{(t^*-1,t)}
=\{j\in[p]:\delta_{h,j}^{(t^*-1,t)}>0\},
\qquad
m_h^{(t^*-1,t)}=|\mathcal I_h^{(t^*-1,t)}|.
\]

	Besides the latent displacement \(\delta_{h,i}^{(t^*-1,t)}\) of vertex $i$, additional
	quantities affect performance for anomaly detection. First, each transition probability depends on
the latent positions of both vertices, giving the interference term
\begin{equation*}
\xi_{h,i}^{(t^*-1,t)}=\sqrt{\frac{\rho_h}{p}}
\|\{\mathbf L_h^{(t)}-\mathbf L_h^{(t^*-1)}\}_{-i}
  \boldsymbol\ell_{h,i}^{(t)}\|_2.
\end{equation*}
Thus a non-anomalous vertex can have a nonzero population row displacement when
many other vertices move. The following bound identifies a sparse regime in which this interference is smaller than the estimation error.

\begin{lemma}
\label{lem:sparse-xi-main}
Uniformly in $i$, we have 
\[
\xi_{h,i}^{(t^*-1,t)}\leq
2d\ell^2\left(\frac{\rho_hm_h^{(t^*-1,t)}}{p}\right)^{1/2}.
\]
In particular, if $\rho_hm_h^{(t^*-1,t)}=o(\log p)$, then
$\xi_{h,i}^{(t^*-1,t)}=o(\epsilon_p)$ for
$\epsilon_p\asymp(\log p/p)^{1/2}$.
\end{lemma}
        
	Second, an anomaly at vertex \(i\) may be difficult to detect
if it changes its inner products with the latent positions of other
vertices only slightly. Define
\(w_{Y,ij}^{(t^*)}
=
1-\Pi_{ij}^{(t^*-1)},
w_{Z,ij}^{(t^*)}
=
\Pi_{ij}^{(t^*-1)}.
\)

\begin{condition}
\label{ass:column-visibility}
For every \(h\in\mathcal H\), vertex \(i\), and target pair
\((t^*-1,t)\), there is a constant \(\mu_{h,i}>0\) such that
\[
\left[
\frac{1}{p}
\sum_{j\ne i}
\left\{
w_{h,ji}^{(t^*)}
(\boldsymbol\ell_{h,j}^{(t^*-1)})^\top
\left(
\boldsymbol\ell_{h,i}^{(t)}
-
\boldsymbol\ell_{h,i}^{(t^*-1)}
\right)
\right\}^{2}
\right]^{1/2}
\ge
\mu_{h,i}\delta_{h,i}^{(t^*-1,t)}.
\]
\end{condition}

The left-hand side measures the average weighted effect of the anomaly
on the inner products between vertex \(i\) and the other vertices. For
formation, the weight is the baseline probability that an edge is
absent and can form; for dissolution, it is the baseline probability
that an edge is present and can dissolve. The condition requires this
effect to be at least a fixed proportion of the anomaly magnitude
\(\delta_{h,i}^{(t^*-1,t)}\).
	
Finally, because transition events depend on the previous edge state,
an earlier anomaly may affect a later comparison even after the
transition probabilities have returned to their baseline values. For
\(h\in\mathcal H\), define
\begin{equation*}
M_{h,i}^{(t^*-1,t)}
=
\left[
\frac{1}{p}
\sum_{j\ne i}
\left\{
\frac{a_{h,ji}^{(t)}}{\sqrt{\rho_h}}
\left(
\Pi_{ji}^{(t-1)}-\Pi_{ji}^{(t^*-1)}
\right)
\right\}^2
\right]^{1/2},
\end{equation*}
where \(a_{Y,ji}^{(t)}=\alpha_{ji}^{(t)}\) and
\(a_{Z,ji}^{(t)}=\beta_{ji}^{(t)}\). This term measures the remaining effect of the earlier anomaly on the marginal adjacency distribution. It is zero for the one-step
comparison \(t=t^*\), and whenever
\(\Pi_{ji}^{(t-1)}=\Pi_{ji}^{(t^*-1)}\) for all \(j\ne i\), including
anomalies that preserve the marginal connection probabilities.
	
Using these quantities, the following theorem gives type I error and
power guarantees for the vertex-level test with threshold
\(\tau_p=\gamma\epsilon_p\), where \(\gamma>2\) is fixed.
	\begin{theorem}\label{thm:2v0}
		Suppose Conditions~\ref{ass:transition-sparsity}, \ref{ass:initial-finite-rank},~\ref{ass:3}, and~\ref{ass:column-visibility} hold.  Fix integers
		\(2\le t^*\le t\le n\), and define the test
		$\phi_{i}^{(t^*-1,t)}=\mathbf{1}\{T_{i}^{(t^*-1,t)}>\tau_p\}, i\in[p]$.
		For all sufficiently large \(p\) and every \(i\in[p]\):
		\begin{enumerate}
			\item 
			Under \(H_{0,i}^{(t^*-1,t)}\), suppose that the null condition
			\begin{align}
				\max_{h\in\mathcal H}c_{h,S}^{-1/2}
\left\{\xi_{h,i}^{(t^*-1,t)}
+M_{h,i}^{(t^*-1,t)}+(\rho_h/p)^{1/2}\right\}
<(\gamma-2)\epsilon_p,
				\label{eq:null_drift_control}
			\end{align}
			holds. Then
			\(
			\mathbb P_{H_{0,i}^{(t^*-1,t)}}
			\left(
			\phi_i^{(t^*-1,t)}=1
			\right)
			\le
			C_{c_0}p^{-c_0}.
			\)
			
			\item 
			Under $H_{1,i}^{(t^*-1,t)}$, suppose that the separation condition
			\begin{align}
				\max_{h\in\{Y,Z\}}C_{h,S}^{-1/2}
\left[\sqrt{\rho_h}
\mu_{h,i}\delta_{h,i}^{(t^*-1,t)}-\xi_{h,i}^{(t^*-1,t)}
-M_{h,i}^{(t^*-1,t)}\right]
>(\gamma+2)\epsilon_p,
				\label{eq:vertex_sep}
			\end{align}
			holds. Then
			$\mathbb{P}_{H_{1,i}^{(t^*-1,t)}}(\phi_{i}^{(t^*-1,t)}=1)\ge 1-C_{c_0}p^{-c_0}$.
		\end{enumerate}
	\end{theorem}
	
	The conditions simplify when
\(\xi_{h,i}^{(t^*-1,t)}=o(\epsilon_p)\), as established by
Lemma~\ref{lem:sparse-xi-main}, and
\(M_{h,i}^{(t^*-1,t)}=o(\epsilon_p)\) for both channels. The latter holds when the memory term is zero, as in the cases
described above, or when it becomes negligible after a sufficient
burn-in, as shown in the next subsection. In this case,
\eqref{eq:null_drift_control} holds for non-anomalous vertices for all
sufficiently large \(p\), while \eqref{eq:vertex_sep} requires
\(
\sqrt{\rho_h}\,
\mu_{h,i}\delta_{h,i}^{(t^*-1,t)}
\)
to exceed a constant multiple of \(\epsilon_p\) in at least one
channel \(h\in\mathcal H\).

The next corollary gives exact recovery of the anomalous vertex set.
	\begin{corollary}\label{cor:perfect-v}
		Fix integers \(2\le t^*\le t\le n\) and define the true anomalous vertex set
		$\mathcal{A}^{(t^*-1,t)} = \{ i \in [p] : H^{(t^*-1,t)}_{1,i}\text{ holds} \}$ and
		the detected set
		$\widehat{\mathcal{A}}^{(t^*-1,t)}=\{i\in[p]:\phi_{i}^{(t^*-1,t)}=1\}$.
		Assume Conditions~\ref{ass:transition-sparsity}, \ref{ass:initial-finite-rank},
		\ref{ass:3}, and
		\ref{ass:column-visibility} hold.  If \eqref{eq:null_drift_control} holds uniformly over
$i\notin\mathcal A^{(t^*-1,t)}$ and \eqref{eq:vertex_sep} holds uniformly over
$i\in\mathcal A^{(t^*-1,t)}$, then, for all sufficiently large $p$,
		$$\mathbb P\{\widehat{\mathcal{A}}^{(t^*-1,t)} = \mathcal{A}^{(t^*-1,t)}\}
		\ge 1 - C_{c_0}p^{-c_0}.$$
        The same statement applied separately to $T_{Y,i}$ and $T_{Z,i}$ gives exact
recovery of the formation- and dissolution-anomaly sets and hence exact
channel attribution.
	\end{corollary}
	
	\subsection{Memory in the isolated anomaly process}
	\label{sec:vertex-memory-main-insert}
	
	Consider the isolated anomaly process where an isolated anomaly at $t^*$ is followed by an immediate return to the
baseline formation and dissolution latent positions. Although the transition mechanism has returned to baseline, the connection probabilities at time \(t-1\) may still differ from their stationary values because of the anomaly at \(t^*\), giving rise to the memory term.  This subsection quantifies
	this AR(1) memory in the vertex statistic and traces its decay into the
	right-embedding coordinates.
	
Using the anomalous-vertex sets defined in
Section 3.1, write
\[
\mathcal I^{(t^*)}
=\mathcal I_Y^{(t^*-1,t^*)}\cup\mathcal I_Z^{(t^*-1,t^*)},
\qquad
m^{(t^*)}=|\mathcal I^{(t^*)}|
\leq m_Y^{(t^*-1,t^*)}+m_Z^{(t^*-1,t^*)}.
\]
Thus $m^{(t^*)}$ counts vertices anomalous in either channel at the
anomaly time. Let
	\[
		\bar r
		=
		\sup_{i<j}|1-\alpha_{ij}-\beta_{ij}|
		\le1, \qquad \Delta_i^{\rm mean}
			=
			\left[
			\frac1p\sum_{j\ne i}
			\left(
			\Pi_{ij}^{(t^*)}-\Pi_{ij}^{(t^*-1)}
			\right)^2
			\right]^{1/2}.
	\]
	
	The following proposition shows that when \(\bar r<1\), the effect of the anomaly-time edge
probabilities on later transition events decreases geometrically.
	
	\begin{proposition}
		\label{prop:memory-geometric-decay}
		Under the isolated anomaly process above, for \(t>t^*\) and $h\in\mathcal H$,
		\begin{equation}\label{eq:memory-decay}
		    M^{(t^*-1,t)}_{h,i}\le \sqrt{\rho_h}\,\bar r^{\,t-t^*-1}\Delta_i^{\rm mean}.
		\end{equation}
		Moreover, for every $i\notin\mathcal I^{(t^*)}$,
\[
M_{h,i}^{(t^*-1,t)}\leq
\left(\frac{\rho_h m^{(t^*)}}{p}\right)^{1/2}
\bar r^{\,t-t^*-1}.
\]
Hence, for a fixed channel \(h\),
$\rho_hm^{(t^*)}=o(\log p)$ places the memory of every non-anomalous
vertex at $o(\epsilon_p)$ at all later times. The conclusion holds in both
channels if $\max_{h\in\mathcal H}\rho_hm^{(t^*)}=o(\log p)$.
	\end{proposition}

	We next quantify how the decaying population memory appears directly in the observed vertex statistic.
\begin{theorem}
	\label{thm:memory-echo-main}
	Suppose Conditions~\ref{ass:transition-sparsity}, \ref{ass:initial-finite-rank}, and~\ref{ass:3}
	hold for the isolated anomaly process. For every fixed \(c_0>0\), there exist
	constants \(0<c<C<\infty\) such that, for all sufficiently large \(p\),
	with probability at least \(1-C_{c_0}p^{-c_0}\), simultaneously for
	\(h\in\mathcal H\), \(i\in[p]\), and \(t>t^*\),
	\[
	\max\left\{cM_{h,i}^{(t^*-1,t)}-2\epsilon_p,0\right\}
	\le
	T_{h,i}^{(t^*-1,t)}
	\le
	CM_{h,i}^{(t^*-1,t)}
	+C\sqrt{\frac{\rho_h}{p}}+2\epsilon_p.
	\]
	Consequently, if the right-hand side of \eqref{eq:memory-decay}
	is \(O(\epsilon_p)\), then
	\(T_{h,i}^{(t^*-1,t)}=O_{\mathbb P}(\epsilon_p)\).
\end{theorem}
Thus memory larger than the rowwise estimation error remains visible
in the observed statistic, whereas memory of order \(\epsilon_p\) is absorbed by
the rowwise estimation error. Combining the geometric envelope with this sample
bound gives a burn-in time.
	
		\begin{corollary}
		\label{cor:vertex-memory-burnin-main}
		Assume the conditions of Theorem~\ref{thm:memory-echo-main} hold, and let
		\(0<\bar r<1\).  For \(i\in[p]\), if \(t>t^*\) satisfies
		\[
			t-t^*-1
			\ge
			\left\lceil
			\frac{
			\log\left\{
			\max\left\{
			1,
			\sqrt{\max\{\rho_Y,\rho_Z\}}\,
			\Delta_i^{\rm mean}/\epsilon_p
			\right\}
			\right\}
			}
			{\log(1/\bar r)}
			\right\rceil,
		\]
		then
$\max_{h\in\mathcal H}M_{h,i}^{(t^*-1,t)}\leq\epsilon_p$ and
$\max_{h\in\mathcal H}T_{h,i}^{(t^*-1,t)}=O_{\mathbb P}(\epsilon_p)$.
	\end{corollary}

	The burn-in is logarithmic in the initial mean displacement and inversely
proportional to $\log(1/\bar r)$. When $\bar r$ is small, dyads quickly forget
an isolated anomaly; when $\bar r$ is close to one, its effect on the marginal adjacency distribution may remain visible for several transitions. The bounds show when this
effect is no larger than the embedding error.
	\subsection{Network-level anomaly detection}
	For network-level detection and $h\in\mathcal H$, let the latent displacement be
	\(
	\Delta_h^{(t^*-1,t)}
	=
	p^{-1/2}\|\mathbf L_h^{(t)}-\mathbf L_h^{(t^*-1)}\|_F.
	\)
	Then define the channel-specific and joint network-level test statistics
	\[
    T_{h,G}^{(t^*-1,t)}
		=
		p^{-1/2}\|\mathbf R_h^{(t)}-\mathbf R_h^{(t^*-1)}\|_F, \quad
	T_{G}^{(t^*-1,t)}=\max_hT_{h,G}^{(t^*-1,t)}.
	\]
	The joint statistic tests
\[
H_0^{(t^*-1,t)}:
\Delta_Y^{(t^*-1,t)}
=
\Delta_Z^{(t^*-1,t)}
=
0
\qquad\text{against}\qquad
H_1^{(t^*-1,t)}:
\Delta_Y^{(t^*-1,t)}>0
\ \text{or}\
\Delta_Z^{(t^*-1,t)}>0.
\]
For each \(h\in\mathcal H\), \(T_{h,G}^{(t^*-1,t)}\) tests whether the
network is anomalous in channel \(h\).
We now state the network-level visibility condition. 
\begin{condition}\label{ass:graph-visibility}
Let
$\mathbf W_Y=\mathbf 1\mathbf 1^\top-\mathbf\Pi^{(t^*-1)}$ and
$\mathbf W_Z=\mathbf\Pi^{(t^*-1)}$. For every \(h\in\mathcal H\)  and target pair
\((t^*-1,t)\), there is a constant \(\mu_{h}>0\) such that
$$p^{-1}\left\|
\{\mathbf L_h^{(t)}\mathbf L_h^{(t)\top}
-\mathbf L_h^{(t^*-1)}\mathbf L_h^{(t^*-1)\top}\}\circ\mathbf W_h
\right\|_F\geq\mu_h\Delta_h^{(t^*-1,t)},
\qquad h\in\mathcal H.$$
\end{condition}
The left-hand side of Condition 5 admits the same two-term decomposition as the vertex-level analysis, corresponding to Condition 4 and the interference term. At the network level, however, both terms contribute to network displacement and are therefore combined in a single Frobenius norm. 
The network-level memory terms can be written as
\[
M_{h,G}^{(t^*-1,t)}=
\left[p^{-2}\sum_{i\ne j}
\left\{\frac{a_{h,ij}^{(t)}}{\sqrt{\rho_h}}
(\Pi_{ij}^{(t-1)}-\Pi_{ij}^{(t^*-1)})\right\}^2\right]^{1/2},
\qquad h\in\mathcal H.
\]

	The network-level anomaly detection theorem is then as follows.
	
	\begin{theorem}\label{thm:2g}
		Suppose Conditions~\ref{ass:transition-sparsity}, \ref{ass:initial-finite-rank},~\ref{ass:3}, and
		\ref{ass:graph-visibility} hold.  Fix integers
		\(2\le t^*\le t\le n\), and define the network-level test
		$\phi^{(t^*-1,t)}=\mathbf 1\{T_G^{(t^*-1,t)}>\tau_p\}$.
		Then for all sufficiently large $p$:
		\begin{enumerate}
			\item 
			Under $H_0^{(t^*-1,t)}$, if
			\begin{align}
				\max_{h\in\mathcal H}
				\frac{1}{\sqrt{c_{h,S}}}
				\left(
				M_{h,G}^{(t^*-1,t)}
				+
				\sqrt{\frac{\rho_h}{p}}
				\right)
				<
				(\gamma-2)\epsilon_p
				\label{eq:graph_null_drift}
			\end{align}
			holds, we have
			$\mathbb P_{H_0^{(t^*-1,t)}}(\phi^{(t^*-1,t)}=1)\le C_{c_0}p^{-c_0}$.
			
			\item 
			Under $H_1^{(t^*-1,t)}$, if the separation condition
			\begin{align}
				\max_{h\in\mathcal H}
				\frac{1}{\sqrt{C_{h,S}}}
				\left[
				\sqrt{\rho_h}\mu_h\Delta_h^{(t^*-1,t)}
				-
				M_{h,G}^{(t^*-1,t)}
				-
				\sqrt{\frac{\rho_h}{p}}
				\right]
				>
				(\gamma+2)\epsilon_p
				\label{eq:graph_sep}
			\end{align}
			holds, then
			$\mathbb P_{H_1^{(t^*-1,t)}}(\phi^{(t^*-1,t)}=1)\ge1-C_{c_0}p^{-c_0}$.
		\end{enumerate}
        Applying the result separately to
	\(T_{Y,G}^{(t^*-1,t)}\) and \(T_{Z,G}^{(t^*-1,t)}\), under the corresponding
	channelwise versions of the null-drift condition~\eqref{eq:graph_null_drift}
and the separation condition~\eqref{eq:graph_sep}, identifies whether the anomaly
	affects formation, dissolution, or both with probability at least
	\(1-C_{c_0}p^{-c_0}\).
	\end{theorem}

In the one-step case,
suppose the $m_h^{(t^*-1,t^*)}$ vertices in
$\mathcal I_h^{(t^*-1,t^*)}$ have the same channel-$h$ displacement $\delta$
and the others are unchanged. Then
$\Delta_h^{(t^*-1,t^*)}=\delta(m_h^{(t^*-1,t^*)}/p)^{1/2}$ and the leading
separation requirement is
\begin{equation*}
m_h^{(t^*-1,t^*)}\delta^2
\gtrsim\frac{\log p}{\rho_h\mu_h^2}.
\end{equation*}
Thus, within the network-level test, smaller per-vertex displacements are detectable when they occur across more vertices.

		\section{Anomaly detection under the degree-corrected stochastic block model}
	
This section considers anomaly detection under a degree-corrected stochastic
	block model (DCSBM).  In contrast to the standard stochastic
	block model, vertices in the same
	community need not have identical latent positions, because vertex-specific degree
	heterogeneity is allowed.  	
	We then specialise Definition~\ref{def:ar1-trdpg} to a degree-corrected
	stochastic block model.

\begin{definition}[AR(1) DCSBM]
\label{def:ar1_dcsbm}
An AR(1)-RDPG is an AR(1) degree-corrected stochastic block model if its channel-specific latent positions have the
block form
\begin{equation}
\boldsymbol\ell_{h,i}^{(t)}=\psi_{h,i}^{(t)}
\mathbf\theta_{h,\nu^{(t)}(i)}^{(t)},\qquad h\in\mathcal H,
\label{eq:dcsbm}
\end{equation}
where $\nu^{(t)}:[p]\to[q_t]$ is the community assignment,
$\psi_{h,i}^{(t)}>0$ is a degree parameter and
$\mathbf\theta_{h,a}^{(t)}\in\mathbb R^d$ is a unit-norm community direction.
\end{definition}
	
 Because the degree scales in \eqref{eq:dcsbm} may differ between channels, we first
normalise each channel separately. For \(h\in\mathcal H\), define
\[
\widetilde{\mathbf R}_{h,i}^{(t)}
=
\frac{\mathbf e_i^\top\mathbf R_h^{(t)}}
{\|\mathbf e_i^\top\mathbf R_h^{(t)}\|_2},
\qquad
\widetilde{\mathbf R}_{S,h,i}^{(t)}
=
\frac{\mathbf e_i^\top\mathbf R_{S_h}^{(t)}\mathbf W_{hS}}
{\|\mathbf e_i^\top\mathbf R_{S_h}^{(t)}\|_2},
\]
where \(\mathbf W_{hS}\) is the alignment matrix from Section~2. The joint anomaly statistic in Section~3 takes the maximum over
the \(Y\) and \(Z\) channels because an anomaly in either channel is
sufficient for rejection. Community membership, however, is determined
jointly by the two channels. We therefore concatenate their embeddings
for clustering. Define
the sample and population joint embeddings by
\[
\widetilde{\mathbf J}_i^{(t)}
=
\frac{1}{\sqrt2}
\left[
\widetilde{\mathbf R}_{Y,i}^{(t)}
\mid
\widetilde{\mathbf R}_{Z,i}^{(t)}
\right],
\qquad
\widetilde{\mathbf J}_{S,i}^{(t)}
=
\frac{1}{\sqrt2}
\left[
\widetilde{\mathbf R}_{S,Y,i}^{(t)}
\mid
\widetilde{\mathbf R}_{S,Z,i}^{(t)}
\right].
\]
The factor \(1/\sqrt2\) gives each joint row unit norm when both channel
rows are nonzero. 

The normalised population rows of vertices in the same community need
not be identical because the transition-event populations also depend on the
connection probabilities at time \(t-1\). We therefore allow these
rows to vary around a common community centre and require the centres
of different communities to be sufficiently separated.
	To formalise this, let $C_1^{(t)},\cdots,C_{q_t}^{(t)}$ be the community partition at time $t$, and define the population centre of community $a$ and the population within-community radius by
	\[
	\widetilde{\boldsymbol{c}}_{S,a}^{(t)}
	=
	|C_a^{(t)}|^{-1}
	\sum_{i\in C_a^{(t)}}\widetilde{\mathbf J}_{S,i}^{(t)}
	\in\mathbb R^{K_Y+K_Z},\qquad
	\widetilde{\eta}_{S,t}
	=
	\max_{a\in[q_t]}
	\max_{i\in C_a^{(t)}}
	\left\|
	\widetilde{\mathbf J}_{S,i}^{(t)}
	-
	\widetilde{\boldsymbol c}_{S,a}^{(t)}
	\right\|_2.
	\]
	
	The following proposition bounds this radius using within-community
differences in previous-time connection probabilities.
\begin{proposition}\label{prop:eta}
		In the AR(1)-DCSBM model, let
$m_{h,S,t}=\min_i\|\mathbf e_i^\top\mathbf R_{S_h}^{(t)}\|_2$ and
$m_{\min,S,t}=\min_{h\in\mathcal H}m_{h,S,t}>0$. Suppose Condition 3 holds. Define
		\[
		\omega_{\Pi,S,u}^{(t)}
		=
		\max_a
		\max_{i,j\in C_a^{(t)}}
		\left[
		\frac{1}{p}
		\sum_{k\notin\{i,j\}}
		\left(
		\Pi_{ki}^{(u)}-\Pi_{kj}^{(u)}
		\right)^2
		\right]^{1/2}
		\]
		for any time \(u\). Then there exists a constant \(C>0\) such that
		\[
		\widetilde{\eta}_{S,t}
		\le
		\frac{C}{m_{\min,S,t}}
		\left(
		\omega_{\Pi,S,t-1}^{(t)}
		+
		p^{-1/2}
		\right).
		\]
		
	\end{proposition}
	
		Proposition~\ref{prop:eta} shows that the within-community radius is controlled by differences in the previous-time connection probabilities; the \(p^{-1/2}\) term accounts for the two completed diagonal coordinates. The normalisation removes the vertex-specific degree scales before the two embeddings are concatenated.

	We consider two related tasks. The first is to recover the community
	partition at every time \(t\in[n]\).  The second  is to detect community anomalies relative to a baseline partition. We first detect vertex community reassignments relative to the baseline partition among the vertices flagged by the
	vertex-level procedure in Section~3.1. Then we consider test statistics for community-level anomaly detection, including 
	centre shifts, splits, and merges. The corresponding statistics
	describe departures in the embedding geometry but do not identify a unique mechanism that generated them.
	
	\subsection{Community detection}

 Assume the number of communities \(q_t\) is known. The community partition
	at time $t$ is estimated by applying a $q_t$-means algorithm to the rows
	of $\widetilde{\mathbf{J}}^{(t)}$.  Specifically, let
	\[
	(\widehat\nu^{(t)},\widehat{\boldsymbol m}_1^{(t)},\ldots,
		\widehat{\boldsymbol m}_{q_t}^{(t)})
	\in
	\arg\min_{\nu:[p]\to[q_t],\,
		\boldsymbol m_1,\ldots,\boldsymbol m_{q_t}\in\mathbb R^{K_Y+K_Z}}
	\sum_{i=1}^p
	\left\|
	\mathbf e_i^\top\widetilde{\mathbf J}^{(t)}
	-
		\boldsymbol m_{\nu(i)}
	\right\|_2^2 .
	\]
	We then define the minimum distance
	between distinct population community centres by
	\(
	\widetilde{\Delta}_{S,t}
	=
	\min_{a\ne b}
	\left\|
	\widetilde{\boldsymbol c}_{S,a}^{(t)}
	-
	\widetilde{\boldsymbol c}_{S,b}^{(t)}
	\right\|_2,
	\)
	and define the minimum community size by
	\(
	s_t^{\min}
	=
	\min_{a\in[q_t]}|C_a^{(t)}|.
	\)
	The following theorem gives a sufficient condition for exact community recovery.

\begin{theorem}\label{thm:dcsbm_community_exact_recovery}
			Suppose the AR(1)-DCSBM satisfies
			Conditions~\ref{ass:transition-sparsity},
			\ref{ass:initial-finite-rank}, and~\ref{ass:3}.  For every
			\(t\in[n]\), 	suppose that \(q_t\ge2\) and
		\(\min_{h\in\mathcal H}m_{h,S,t}>\epsilon_p\).
		For each \(t\in[n]\), define
		\(
		\widetilde{\epsilon}_t
		=
		\sqrt{2}\,\epsilon_p
		\left(
		m_{Y,S,t}^{-2}+m_{Z,S,t}^{-2}
		\right)^{1/2}.
		\)
			If, for every \(t\in[n]\),
			\begin{equation}
				\label{eq:dcsbm_community_sep_condition}
				\widetilde{\Delta}_{S,t}
				>
				6\bigl(
				\widetilde{\epsilon}_t+
				\widetilde{\eta}_{S,t}
				\bigr)
				\sqrt{\frac{p}{s_t^{\min}}},
			\end{equation}
			then $q_t$-means applied to the rows of $\widetilde{\mathbf J}^{(t)}$ recovers all
community partitions, up to label permutations, with probability at least
$1-C_{c_0}np^{-c_0}$.
	\end{theorem}
	In the balanced stationary SBM, we have
\(s_t^{\min}=\Theta(p/q_t)\) and
\(\widetilde\eta_{S,t}=0\). If
\(m_{h,S,t}\gtrsim\sqrt{\rho_h}\) for \(h\in\mathcal H\), a sufficient
separation condition is
\[
\widetilde\Delta_{S,t}
\gtrsim
\left[
\frac{q_t\log p}{p}
\left(
\frac{1}{\rho_Y}+\frac{1}{\rho_Z}
\right)
\right]^{1/2}.
\]
The right-hand side tends to zero as $p\rightarrow\infty$ under
Condition~\ref{ass:transition-sparsity}.

		Suppose that, over an interval \(\mathcal T\), the process is
	stationary and the community assignments and latent positions remain
	unchanged. Let \(q_{\mathcal T}\) denote the common number of
	communities over this interval. Then define
	\[
	\overline{\widetilde{\mathbf J}}_{\mathcal T}
	=
	\frac{1}{|\mathcal T|}
	\sum_{t\in\mathcal T}\widetilde{\mathbf J}^{(t)},
	\qquad
	\overline{\widetilde{\mathbf J}}_{S,\mathcal T}
	=
	\frac{1}{|\mathcal T|}
	\sum_{t\in\mathcal T}\widetilde{\mathbf J}_S^{(t)}.
	\]
	We apply \(q_{\mathcal T}\)-means to the rows of
	\(\overline{\widetilde{\mathbf J}}_{\mathcal T}\). Define
	\(\widetilde\eta_{S,\mathcal T}\),
	\(\widetilde\Delta_{S,\mathcal T}\), and \(s_{\mathcal T}^{\min}\)
	from the interval-averaged population embedding as above.
	
	\begin{corollary}
		\label{cor:pooled_community_exact_recovery}
		Suppose Conditions~\ref{ass:transition-sparsity},
		\ref{ass:initial-finite-rank}, and~\ref{ass:3} hold over
		\(\mathcal T\). Suppose \(q_{\mathcal T}\ge2\) and, for every
		\(t\in\mathcal T\),
		\(\min_{h\in\mathcal H}m_{h,S,t}>\epsilon_p\). Let
		\(
		\widetilde\epsilon_{\mathcal T}
		=
		|\mathcal T|^{-1}
		\sum_{t\in\mathcal T}\widetilde\epsilon_t
		\).
		If
		\[
		\widetilde\Delta_{S,\mathcal T}
		>
		6\left(
		\widetilde\epsilon_{\mathcal T}
		+
		\widetilde\eta_{S,\mathcal T}
		\right)
		\sqrt{\frac{p}{s_{\mathcal T}^{\min}}},
		\]
		then for every fixed \(c_0>0\) and all sufficiently large \(p\), \(q_{\mathcal T}\)-means recovers the common
		community partition exactly, up to a permutation of labels, with
		probability at least \(1-C_{c_0}p^{-c_0}\).
	\end{corollary}
\subsection{Anomaly detection for communities}
	
	Let a stationary baseline $B$ end at transition $t^*-1$. Throughout \(B\), the community memberships and latent positions remain
	unchanged. Set
$\mathcal T=B$ in Corollary~\ref{cor:pooled_community_exact_recovery}, and
denote the population and estimated baseline assignments by $\nu_{S,B}$ and
$\widehat\nu_B$, with community sets $C_{S,B,a}$ and $\widehat C_{B,a}$.
After relabelling, we hold this partition fixed and define, for
$u\in\{t^*-1,t\}$,
\begin{equation*}
\begin{aligned}
\widetilde{\mathbf c}_{B,a}^{(u)}
=|\widehat C_{B,a}|^{-1}
\sum_{i\in\widehat C_{B,a}}\widetilde{\mathbf J}_i^{(u)}, \qquad
\widetilde{\mathbf c}_{S,B,a}^{(u)}
=|C_{S,B,a}|^{-1}
\sum_{i\in C_{S,B,a}}\widetilde{\mathbf J}_{S,i}^{(u)}.
\end{aligned}
\end{equation*}
Set $\overline\epsilon_{t^*-1,t}=\widetilde\epsilon_{t^*-1}+
\widetilde\epsilon_t$.
Community assignments are estimated once from the baseline and then
held fixed. This allows the same communities to be compared over time without
label permutations.

First, we determine which vertices identified as anomalous by the
vertex-level procedure in Section~3.1 have also been reassigned to a
different baseline community. An anomalous vertex need not be
reassigned: its target row may move while its original baseline centre
remains the nearest centre. A reassignment occurs when another baseline
centre becomes nearest. We therefore estimate its target community
label by
\[
\widehat\tau_B^{(t^*-1,t)}(i)
=
\arg\min_{b\in[q_B]}
\left\|
\widetilde{\mathbf J}_i^{(t)}
-
\widetilde{\mathbf c}_{B,b}^{(t^*-1)}
\right\|_2.
\]
Define the population reassignment set and its estimator by
\begin{equation*}
\mathcal J_{S,B}^{(t^*-1,t)}
=
\left\{
i\in\mathcal A^{(t^*-1,t)}:
\nu^{(t)}(i)\ne\nu_{S,B}(i)
\right\},\quad
\widehat{\mathcal J}_B^{(t^*-1,t)}
=
\left\{
i\in\widehat{\mathcal A}^{(t^*-1,t)}:
\widehat\tau_B^{(t^*-1,t)}(i)\ne\widehat\nu_B(i)
\right\}.
\end{equation*}

\begin{corollary}
\label{thm:dcsbm_jump_set_exact_recovery}
Suppose the conditions of
Corollary~\ref{cor:pooled_community_exact_recovery} hold with
$\mathcal T=B$, and the conditions of Corollary~\ref{cor:perfect-v} hold. Suppose also
that $q_t=q_B$, that $\nu^{(t)}$ and $\nu_{S,B}$ use a common labelling, and
that $m_{h,S,t}>\epsilon_p$ for $h\in\mathcal H$. If, for some
$\widetilde\zeta_{S,B,t}\geq0$,
\begin{equation}
\max_{a\in[q_B]}
\|\widetilde{\mathbf c}_{S,a}^{(t)}-
\widetilde{\mathbf c}_{S,B,a}^{(t^*-1)}\|_2
\leq\widetilde\zeta_{S,B,t},
\label{eq:dcsbm_centre_drift}
\end{equation}
and
\begin{equation}
\min_{a\ne b}
\|\widetilde{\mathbf c}_{S,B,a}^{(t^*-1)}-
\widetilde{\mathbf c}_{S,B,b}^{(t^*-1)}\|_2
>2\{\widetilde\eta_{S,t}+\widetilde\zeta_{S,B,t}+
\overline\epsilon_{t^*-1,t}\},
\label{eq:dcsbm_jump_set_margin}
\end{equation}
then
\[
{\mathbb P}\{\widehat{\mathcal J}_B^{(t^*-1,t)}
=\mathcal J_{S,B}^{(t^*-1,t)}\}
\geq1-C_{c_0}(|B|+1)p^{-c_0}.
\]
\end{corollary}

Condition~\eqref{eq:dcsbm_centre_drift} allows each population
community centre to move by at most
\(\widetilde\zeta_{S,B,t}\) from its corresponding baseline centre;
exact centre stability is the special case
\(\widetilde\zeta_{S,B,t}=0\). Condition~\eqref{eq:dcsbm_jump_set_margin} ensures that the nearest
baseline centre identifies each target community correctly.
	
We now define test statistics for the three types of community-level anomaly. The sample and population \emph{centre-shift} statistics of
community $a$ are measured by
\[
\widetilde T_{B,a}^{\mathrm{cent},(t^*-1,t)}
=
\left\|
\widetilde{\boldsymbol c}_{B,a}^{(t)}
-
\widetilde{\boldsymbol c}_{B,a}^{(t^\ast-1)}
\right\|_2,
\qquad
\widetilde T_{S,B,a}^{\mathrm{cent},(t^*-1,t)}
=
\left\|
\widetilde{\boldsymbol c}_{S,B,a}^{(t)}
-
\widetilde{\boldsymbol c}_{S,B,a}^{(t^\ast-1)}
\right\|_2.
\]
These statistics measure how far the centre of community \(a\) moves
	from its position at \(t^\ast-1\) in the channelwise-normalised joint
	embedding.

To detect a \emph{split} of community $a$, set
\[
\widetilde{\boldsymbol d}_{B,i}^{(t^*-1,t)}
:=
\widetilde{\mathbf J}_i^{(t)}
-
\widetilde{\mathbf J}_i^{(t^\ast-1)},
\qquad
\widetilde{\boldsymbol d}_{S,B,i}^{(t^*-1,t)}
:=
\widetilde{\mathbf J}_{S,i}^{(t)}
-
\widetilde{\mathbf J}_{S,i}^{(t^\ast-1)}.
\]

Define the sample and population collections of candidate subsets by
\[
\widehat{\mathcal P}_{B,a}
:=
\left\{
\widehat A:
\widehat A\neq\varnothing,\ 
\widehat A\subsetneq\widehat C_{B,a}
\right\},
\qquad
\mathcal P_{S,B,a}
:=
\left\{
A:
A\neq\varnothing,\ 
A\subsetneq C_{S,B,a}
\right\}.
\]
Define their mean displacements by
\[
\widetilde{\boldsymbol d}_{B,a,\widehat A}^{(t^*-1,t)}
:=
\frac{1}{|\widehat A|}
\sum_{i\in\widehat A}\widetilde{\boldsymbol d}_{B,i}^{(t^*-1,t)},
\qquad
\widetilde{\boldsymbol d}_{S,B,a,A}^{(t^*-1,t)}
:=
\frac{1}{|A|}
\sum_{i\in A}\widetilde{\boldsymbol d}_{S,B,i}^{(t^*-1,t)},
\qquad
\widehat A\in\widehat{\mathcal P}_{B,a},
A\in\mathcal P_{S,B,a}.
\]
Then the sample and population split statistics are
\begin{align}
	\widetilde T_{B,a}^{\mathrm{split},(t^*-1,t)}
	&:=
	\max_{\widehat A\in\widehat{\mathcal P}_{B,a}}
	\left\{
	\frac{|\widehat A|\,|\widehat C_{B,a}\setminus\widehat A|}
	{|\widehat C_{B,a}|^2}
	\right\}^{1/2}
	\left\|
	\widetilde{\boldsymbol d}_{B,a,\widehat A}^{(t^*-1,t)}
	-
	\widetilde{\boldsymbol d}_{B,a,\widehat C_{B,a}\setminus\widehat A}^{(t^*-1,t)}
	\right\|_2,
	\label{eq:balanced-split-statistic}\\
	\widetilde T_{S,B,a}^{\mathrm{split},(t^*-1,t)}
	&:=
	\max_{A\in\mathcal P_{S,B,a}}
	\left\{
	\frac{|A|\,|C_{S,B,a}\setminus A|}
	{|C_{S,B,a}|^2}
	\right\}^{1/2}
	\left\|
	\widetilde{\boldsymbol d}_{S,B,a,A}^{(t^*-1,t)}
	-
	\widetilde{\boldsymbol d}_{S,B,a,C_{S,B,a}\setminus A}^{(t^*-1,t)}
	\right\|_2.
	\label{eq:balanced-split-population}
\end{align}
The weight in \eqref{eq:balanced-split-statistic} and \eqref{eq:balanced-split-population} reduces the contribution of highly unbalanced divisions
without excluding them. The theoretical statistic considers every
division of the community into two nonempty groups. In computation, we
apply \(k\)-means with \(k=2\) to the vertex displacement vectors and
evaluate the statistic using the resulting two groups.

Finally, a \emph{merge} is a reduction in the separation of a prespecified
group $G\subseteq[q_B]$, $|G|\ge2$. For \(u\in\{t^\ast-1,t\}\), define the sample and population
mean centres by
\[
\overline{\widetilde{\boldsymbol c}}_{B,G}^{(u)}
=
\frac{1}{|G|}
\sum_{a\in G}
\widetilde{\boldsymbol c}_{B,a}^{(u)},
\qquad
\overline{\widetilde{\boldsymbol c}}_{S,B,G}^{(u)}
=
\frac{1}{|G|}
\sum_{a\in G}
\widetilde{\boldsymbol c}_{S,B,a}^{(u)},
\]
and their dispersions by
\[
\widetilde D_{B,G}^{(u)}
=
\left[
\frac{1}{|G|}
\sum_{a\in G}
\left\|
\widetilde{\boldsymbol c}_{B,a}^{(u)}
-
\overline{\widetilde{\boldsymbol c}}_{B,G}^{(u)}
\right\|_2^2
\right]^{1/2},
\widetilde D_{S,B,G}^{(u)}
=
\left[
\frac{1}{|G|}
\sum_{a\in G}
\left\|
\widetilde{\boldsymbol c}_{S,B,a}^{(u)}
-
\overline{\widetilde{\boldsymbol c}}_{S,B,G}^{(u)}
\right\|_2^2
\right]^{1/2}.
\]
The sample and population merge
statistics are
\[
\widetilde T_{B,G}^{\mathrm{merge},(t^*-1,t)}
=
\widetilde D_{B,G}^{(t^\ast-1)}
-
\widetilde D_{B,G}^{(t)},
\qquad
\widetilde T_{S,B,G}^{\mathrm{merge},(t^*-1,t)}
=
\widetilde D_{S,B,G}^{(t^\ast-1)}
-
\widetilde D_{S,B,G}^{(t)}.
\]
A positive value means that the selected centres have moved closer together.

The following result gives a common detection guarantee for the three
tests. For any one of the three tests above, let
\(
T_C^{(t^*-1,t)}
\in
\left\{
\widetilde T_{B,a}^{\mathrm{cent},(t^*-1,t)},
\widetilde T_{B,a}^{\mathrm{split},(t^*-1,t)},
\widetilde T_{B,G}^{\mathrm{merge},(t^*-1,t)}
\right\}\), \(
T_{S,C}^{(t^*-1,t)}
\in
\left\{
\widetilde T_{S,B,a}^{\mathrm{cent},(t^*-1,t)},
\widetilde T_{S,B,a}^{\mathrm{split},(t^*-1,t)},
\widetilde T_{S,B,G}^{\mathrm{merge},(t^*-1,t)}
\right\}
\)
denote the corresponding sample and population statistics,
respectively. The same notation is used for any pair \((u,v)\). Both
\(T_C^{(u,v)}\) and \(T_{S,C}^{(u,v)}\) are obtained by replacing
the baseline endpoint \(t^\ast-1\) and target \(t\) by \(u\) and \(v\),
respectively.
For the centre-shift and split tests, the population null
hypothesis is \(T_{S,C}^{(t^*-1,t)}=0\), whereas for the merge test,
it is \(T_{S,C}^{(t^*-1,t)}\leq0\). In all three cases, the corresponding
alternative is \(T_{S,C}^{(t^*-1,t)}>0\). Let \(\tau>0\) denote the rejection threshold for the selected test.
\begin{theorem}
	\label{thm:community_score_detection}
	Suppose the conditions of
	Corollary~\ref{cor:pooled_community_exact_recovery} hold with
	\(\mathcal T=B\), and
	Conditions~\ref{ass:transition-sparsity}--\ref{ass:3} hold at
	time \(t\). Assume that \(m_{Y,S,t}>\epsilon_p\) and
		\(m_{Z,S,t}>\epsilon_p\). Fix a community \(a\) for the centre-shift or split
	test, or a prespecified group
	\(G\subseteq[q_B]\), \(|G|\ge2\), for the merge test. 
    
    Consider the
	corresponding test that rejects when
	\(
	T_C^{(t^*-1,t)}>\tau.
	\)
	For every fixed \(c_0>0\) and all sufficiently large \(p\), under
	the corresponding null hypothesis, if
	\(
	\overline\epsilon_{t^\ast-1,t}\le\tau,
	\)
	the test does not reject with probability at least
	\(
	1-C_{c_0}(|B|+1)p^{-c_0}.
	\)
	Under the corresponding alternative hypothesis, the test rejects
	with the same probability if
	\(
	T_{S,C}^{(t^*-1,t)}
	>
	\tau+\overline\epsilon_{t^\ast-1,t}.
	\)
\end{theorem}

\section{Empirical \(p\)-values and baseline stationarity}
\label{sec:calibration}

The theoretical thresholds depend on unknown constants. In practice,
we use all pairs of transitions from a candidate stationary baseline
\(B\) to form an empirical reference distribution. The theory uses the
endpoint of \(B\) as its reference, while the empirical procedure
averages over all baseline transitions.

For a target \(t\notin B\), define the network-level statistic and
\(p\)-value by
\begin{equation}
\overline T_G^{(t;B)}
=
\frac{1}{|B|}\sum_{b\in B}T_G^{(b,t)},
\qquad
p_G(t)
=
\frac{
1+\displaystyle\sum_{\substack{a<b,a,b\in B}}
\mathbf 1\!\left\{
T_G^{(a,b)}\geq\overline T_G^{(t;B)}
\right\}
}{
1+\binom{|B|}{2}
}.
\label{eq:network-empirical-p}
\end{equation}

For vertex \(i\), define
\begin{equation}
\overline T_i^{(t;B)}
=
\frac{1}{|B|}\sum_{b\in B}T_i^{(b,t)},
\qquad
p_i(t)
=
\frac{
1+\displaystyle\sum_{\substack{a<b,a,b\in B}}\sum_{j=1}^p
\mathbf 1\!\left\{
T_j^{(a,b)}\geq\overline T_i^{(t;B)}
\right\}
}{
1+p\binom{|B|}{2}
}.
\label{eq:vertex-empirical-p}
\end{equation}
The baseline statistics are combined across vertices to give a larger
reference set when \(|B|\) is small.
 Under a stationary baseline with
fixed latent positions, each vertex's population right-embedding row is unchanged over \(B\), so all baseline-pair vertex statistics measure
sampling variation on the same scale. 

For any prespecified centre-shift, split, or merge statistic, define
\begin{equation}
\overline T_C^{(t;B)}
=
\frac{1}{|B|}\sum_{b\in B}T_C^{(b,t)},
\qquad
p_C(t)
=
\frac{
1+\displaystyle\sum_{\substack{a<b,a,b\in B}}
\mathbf 1\!\left\{
T_C^{(a,b)}\geq\overline T_C^{(t;B)}
\right\}
}{
1+\binom{|B|}{2}
}.
\label{eq:community-empirical-p}
\end{equation}

Before testing targets outside \(B\), we check baseline stationarity by treating each \(r\in B\) in turn as a target and using the remaining baseline indices \(B_{-r}=B\setminus\{r\}\) as its reference set. For every relevant statistic
\(T_\bullet\in\{T_G,T_i,T_C\}\), define
\(
\overline T_\bullet^{(r;B_{-r})}
=
\frac{1}{|B|-1}
\sum_{b\in B_{-r}}T_\bullet^{(b,r)}.
\)
The corresponding stationarity \(p\)-values are computed from
\eqref{eq:network-empirical-p}--\eqref{eq:community-empirical-p} with
\(t=r\) and \(B\) replaced by \(B_{-r}\). The baseline stationarity
check combines multiple \(p\)-values to make an overall decision about
\(B\). We therefore apply the Benjamini--Hochberg (BH) adjustment
\citep{benjamini1995controlling} separately to the network-level
family, the vertex-level family, and each selected community-level
family. The baseline passes if no adjusted \(p\)-value is below the
significance level. A baseline-only implementation would first fit UASE over \(B\) and
then refit it over the full analysis window after \(B\) passes the
stationarity check. For simplicity, the numerical analyses below
fit UASE once over the full window \(W\) and apply the stationarity
check to the fitted blocks indexed by \(B\). 

The complete workflow is as follows.
\begin{enumerate}
\item Construct the formation and dissolution matrices over \(W\),
    choose the working dimensions, and fit the two unfolded adjacency spectral embeddings
    over \(W\).

\item Apply the baseline stationarity check to the fitted blocks
    indexed by \(B\). For each \(r\in B\), treat \(r\) as a target and
    use pairs from \(B\setminus\{r\}\) as its reference. Apply the
    Benjamini--Hochberg adjustment separately to the network-level
    family, the vertex-level family, and each selected community-level
    family. Proceed only if the baseline passes.

    \item Using the full-time fitted embeddings, compute
    \eqref{eq:network-empirical-p}--\eqref{eq:community-empirical-p}
    for each external target. Compute channel-specific \(p\)-values
    using \(T_Y\) and \(T_Z\) separately.

    \item Report the target \(p\)-values pointwise. Each vertex
    \(p\)-value concerns only that vertex, and no simultaneous
    error-control claim across vertices is made.
\end{enumerate}

Complete pseudocode is provided in the Supplementary Material.

		\section{Simulation studies}
	\label{sec:simulations}	
	In this section, we compare UASE with multiple adjacency spectral embedding (MASE) \citep{chen_multiple_2025}, omnibus embedding (OMNI)
\citep{levin2017central}, and the scan statistic of
\citet{priebe_enron_2005}. Every reported detection rate is based on $100$ Monte
Carlo replications at level 0.05. The latent dimension used to generate a
network is denoted by $d$, whereas $D$ is the embedding dimension. 

	\subsection{Vertex-level transition anomalies}
	\label{subsec:sim-vertex}
		\begin{figure}[t]
		\centering
		\includegraphics[width=.96\textwidth]{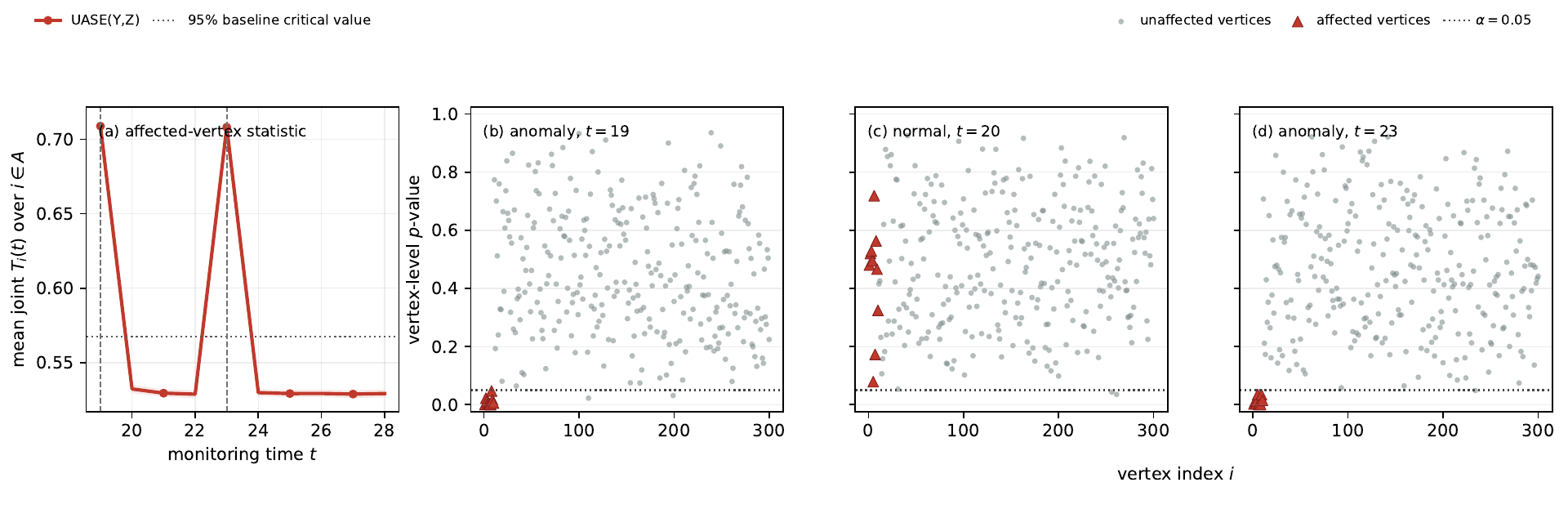}
		\vspace{0.1em}
		
		\includegraphics[width=.94\textwidth]{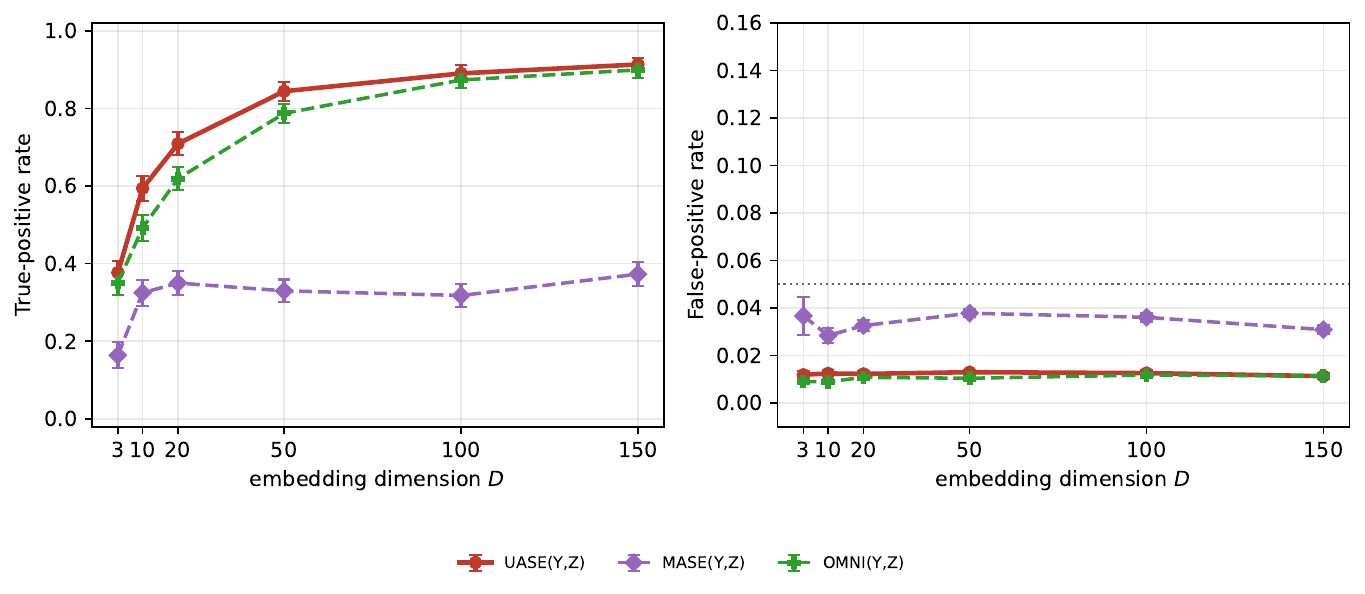}
		\caption{Setting 1: vertex-level anomaly experiments. The upper row shows the transition-event
statistic over time and pointwise vertex $p$-values at two anomaly times and an
intervening normal time. The lower row shows the true-positive rate (TPR) and false-positive rate (FPR)
against the embedding dimension $D$ when the transition is anomalous but the marginal connection
probabilities do not change. Error bars are 95\% Monte Carlo confidence intervals.}
		\label{fig:sim-vertex-diagnostics-d}
	\end{figure}

	Setting~1 considers vertex-level transition anomalies in an
AR(1)-RDPG with \(n=28\) and \(d=3\). Unless stated otherwise, we use
\(p=300\) vertices, of which \(m=10\) vertices are anomalous. The coordinates of
\(\boldsymbol\ell_{Y,i}^{0}\) and
\(\boldsymbol\ell_{Z,i}^{0}\) are generated independently from
\(\operatorname{Unif}(0.25,0.30)\) and
\(\operatorname{Unif}(0.20,0.30)\), respectively, with
\(\rho_Y=\rho_Z=0.4\). We use \(t=1,\ldots,18\) as the stationary
baseline. At \(t=19\) and \(t=23\), both latent positions of each
anomalous vertex are doubled:
\(
\boldsymbol\ell_{Y,i}^{(t)}
=
2\boldsymbol\ell_{Y,i}^{0},
\boldsymbol\ell_{Z,i}^{(t)}
=
2\boldsymbol\ell_{Z,i}^{0},
 i\in\mathcal A.
\)
At all other times, the latent positions equal their baseline values.
With \(D=50\), the first row of Figure~\ref{fig:sim-vertex-diagnostics-d} shows the
\(\mathrm{UASE}(Y,Z)\) statistics for the affected vertices over time
and representative vertex-level \(p\)-values at \(t=19,20,23\). Both
distinguish the anomaly times \(t=19,23\) from the normal time \(t=20\).

		\begin{figure}[t]
		\centering
		\includegraphics[width=.82\textwidth]{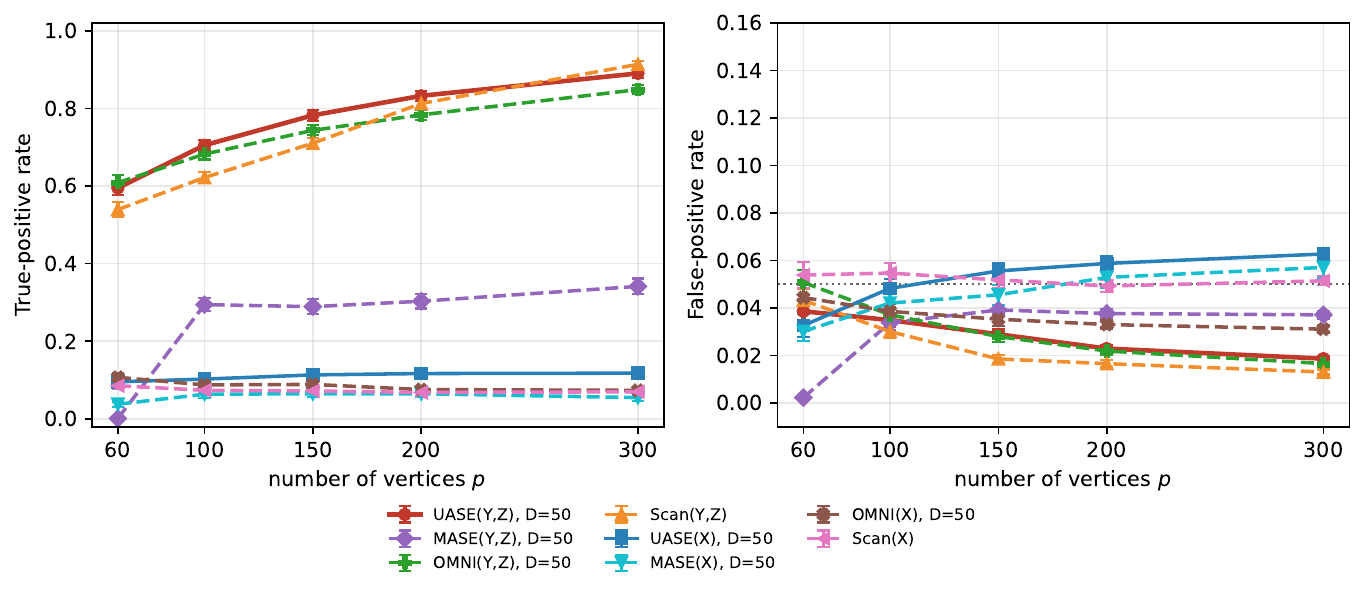}
		\vspace{0.2em}
		
		\includegraphics[width=.82\textwidth]{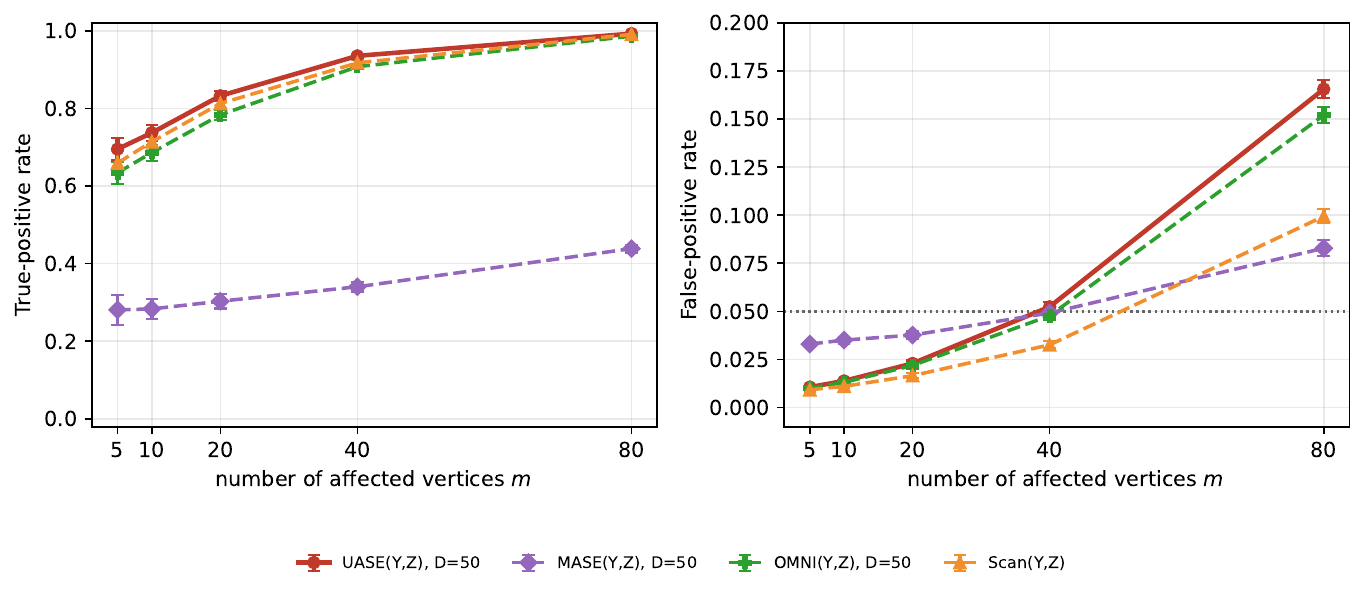}
		\caption{Scaling behaviour in Setting~1. Top: TPR and FPR as the number
			of vertices \(p\) increases, with \(m=20\) and \(D=50\); methods
			based on \(X\) and on \((Y,Z)\) are both included. Bottom: TPR and
			FPR as the number of anomalous vertices \(m\) increases, with
			\(p=200\) and \(D=50\); only the \((Y,Z)\)-based methods are shown.
			Error bars show \(95\%\) Monte Carlo confidence intervals over
			\(100\) replications.}
		\label{fig:sim-vertex-scaling}
	\end{figure}
    
	We next examine the effect of the embedding dimension \(D\), as shown
	in the second row of Figure~\ref{fig:sim-vertex-diagnostics-d}. The
	detection performance of UASE and OMNI improves steadily as more
	embedding dimensions are retained, whereas MASE exhibits a weaker
	improvement. Because the anomalies affect only a small subset of
	vertices, the anomaly signal may not align with the leading spectral directions
of the baseline. A larger embedding dimension retains additional spectral directions
that may contain this localised anomaly signal, thereby improving
vertex localisation. In
	contrast, the false-positive rates of the spectral methods remain low and
	comparatively stable across \(D\).

	We then fix \(m=20\) and \(D=50\) and vary the number of vertices, as
	shown in the first row of Figure~\ref{fig:sim-vertex-scaling}. We apply
	UASE, MASE, OMNI, and the scan statistic to both \(X\) and the transition
	decomposition \(Y,Z\). The anomaly multiplies
\(\alpha_{ij}\) and \(\beta_{ij}\) by the same factor, so
\(\pi_{ij}\) is unchanged. The methods based on \(X\) consequently
show little separation, whereas those based on \(Y,Z\) improve with
\(p\). Among them, \(\mathrm{UASE}(Y,Z)\) has the highest TPR for
\(p\leq200\); Scan is slightly higher at \(p=300\); OMNI is close
behind; and MASE has lower power.
			
	Finally, we fix \(p=200\) and \(D=50\) and vary the number of anomalous
	vertices, as shown in the second row of
	Figure~\ref{fig:sim-vertex-scaling}. The TPR
increases with \(m\) because the number of anomalous--anomalous dyads
increases, strengthening the anomaly signal. The FPR at the anomaly times also increases with \(m\),
consistent with the cross-vertex interference in Section~3.1 becoming
stronger as more vertices are anomalous.

	\subsection{Network-level anomalies and channel attribution}
	\label{subsec:sim-graph}
	\begin{figure}[t]
		\centering
		\includegraphics[width=\textwidth]
		{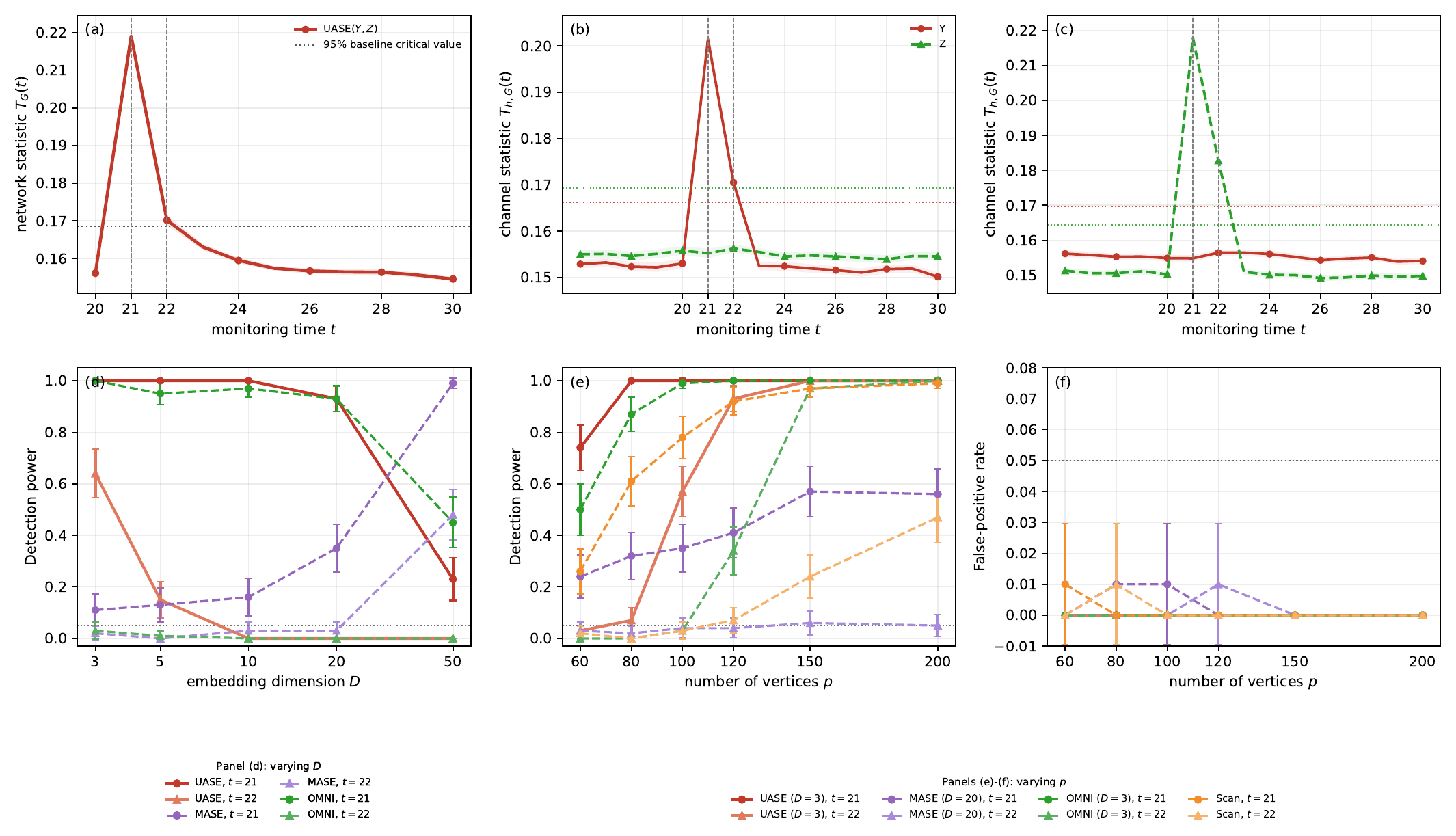}
		\caption{Setting~2: Network anomaly experiments. Panels (a)--(c) report the joint and
channel-specific statistics with anomalies involving both channels, formation only, and dissolution only. Panel (d) varies the fitted dimension; panels (e)
and (f) report power and false-positive rate against network size. Error bars
are 95\% Monte Carlo confidence intervals.}
		\label{fig:sim-graph}
	\end{figure}
	Setting~2 considers network-level transition anomalies in an
AR(1)-RDPG with \(n=30\), \(d=3\), and, unless stated otherwise,
\(p=100\). The coordinates of \(\boldsymbol\ell^0_{Y,i}\) and
\(\boldsymbol\ell^0_{Z,i}\) are generated independently from
\(\operatorname{Unif}(0.25,0.30)\) and
\(\operatorname{Unif}(0.20,0.30)\), respectively, with
\(\rho_Y=\rho_Z=0.6\). We use \(t=1,\ldots,19\) as the stationary
baseline.
For anomalies occurring in both channels, we randomly divide the vertices into two
equal groups indexed by \(g_i\in\{-1,1\}\). At \(t=21,22\), set
\[
\begin{aligned}
g_i=1:\quad&
\boldsymbol\ell^{(t)}_{Y,i}
=\boldsymbol\ell^0_{Y,i}+0.20\boldsymbol 1_d,
&\quad
\boldsymbol\ell^{(t)}_{Z,i}
=\boldsymbol\ell^0_{Z,i}-0.18\boldsymbol 1_d,\\
g_i=-1:\quad&
\boldsymbol\ell^{(t)}_{Y,i}
=\boldsymbol\ell^0_{Y,i}-0.18\boldsymbol 1_d,
&\quad
\boldsymbol\ell^{(t)}_{Z,i}
=\boldsymbol\ell^0_{Z,i}+0.20\boldsymbol 1_d .
\end{aligned}
\]
The latent positions return to baseline after \(t=22\). For channel
attribution, we repeat the experiment with only the \(Y\)-channel or
only the \(Z\)-channel shift above, keeping the other channel at
baseline.

	The first row of Figure~\ref{fig:sim-graph} shows the overall
\(\mathrm{UASE}(Y,Z)\) statistic for the two-channel anomaly and the
channel-specific statistics for the one-channel anomalies. Both
two-channel anomaly times are detected, although the signal at \(t=22\)
is weaker, consistent with the memory effect from the preceding
anomaly. In the one-channel experiments, the statistics correctly
identify the \(Y\) and \(Z\) channels,
	respectively. 
	The lower-left panel compares power across embedding dimensions at
\(t=21,22\). UASE performs best at small and moderate \(D\), while the
performance of UASE and OMNI decreases at larger \(D\), where additional
dimensions mainly add noise. MASE improves with \(D\) because at small $D$ its
low-dimensional common basis may omit the temporary anomaly direction.
The lower-middle and lower-right panels vary \(p\), using \(D=3\) for
UASE and OMNI and \(D=20\) for MASE. UASE has the highest power and
approaches one as \(p\) increases, followed by OMNI and Scan, while
MASE generally has the lowest power. All four methods maintain low
empirical false-positive rates.

	\subsection{Community-level structure and anomaly geometry}
	\label{subsec:sim-community}

		\begin{figure}[t]
		\centering
		\includegraphics[width=\textwidth]
{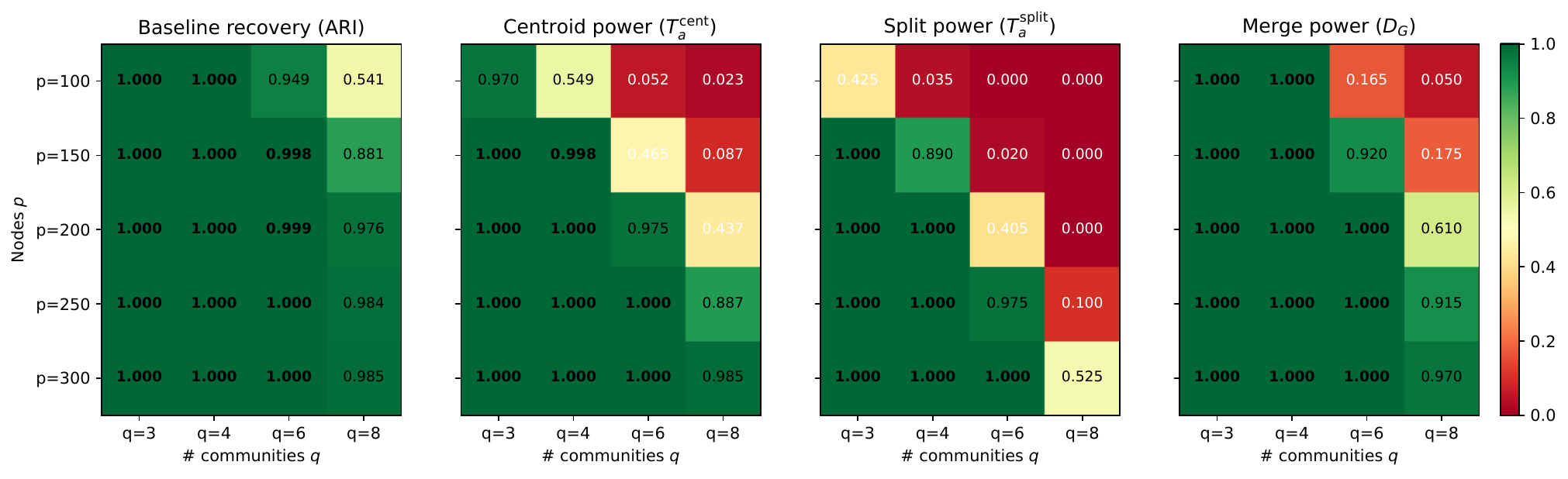}
		\caption{Setting 3: Community recovery and anomaly detection over the $(p,q)$ grid.
The first panel reports the adjusted Rand index (ARI) of the baseline
partition. The remaining panels report power for centre-shift, split and
merge alternatives. Each cell averages 100 paired Monte Carlo
replications.}
		\label{fig:sim-community-pq}
	\end{figure}
	
	Setting~3 studies three types of community anomaly with
	\(p\in\{100,150,200,250,300\}\) vertices and
	\(q\in\{3,4,6,8\}\) balanced communities.  For each cell we use
	\(d=q\), \(D=2q+3\), \(n=45\), and
	\(\rho_Y=\rho_Z=1\). The degree parameters are drawn from
	\(\operatorname{Unif}(0.22,0.60)\). The baseline ends at \(t=30\), and
	anomalies are introduced at \(t=35\) and \(t=40\).
For each \((p,q)\) cell, we use 100 paired Monte Carlo
	replications. Baseline
	memberships are estimated by \(k\)-means with \(k=q\), applied to the
	channelwise normalised rows averaged over \(t=1,\ldots,30\). The first panel
	of Figure~\ref{fig:sim-community-pq} reports the mean ARI, obtained by averaging over the 100 replications. Recovery 	improves with \(p\) and becomes more difficult as \(q\) increases, consistent
	with Theorem~\ref{thm:dcsbm_community_exact_recovery}.
	The two transition channels share the same latent geometry. The baseline
	direction of community \(k\in[q]\) is
	\(
	\boldsymbol\theta_k
	=
	\frac{\boldsymbol r_k}{\|\boldsymbol r_k\|_2},
	\boldsymbol r_k
	=
	\mathbf e_k
	+
	0.16\sum_{\ell\ne k}\mathbf e_\ell,
	\)
	where \(\mathbf e_k\) denotes the \(k\)th canonical basis vector. Vertex \(i\)
	has baseline latent position \(\psi_i\boldsymbol\theta_{\nu(i)}\). Each
	alternative changes this geometry only at the anomaly times. Under the split
	alternative, community \(1\) is divided into two equal groups whose directions
	are proportional to
	\(
	(0.05,10.0,0.05,\ldots,0.05)
	\text{ and }
	(0.05,0.05,10.0,0.05,\ldots,0.05).
	\)
	Under the merge alternative for \(G=\{1,2,3\}\), each community’s new direction is the normalised sum of \(0.4\) times its original direction and \(0.6\) times the common direction \(\boldsymbol r_1+\boldsymbol r_2+\boldsymbol r_3\).
	Under the global-shift alternative, the same vector
	\(0.22\mathbf 1_q/\sqrt q\) is added to every latent position.  The three
	community tests exhibit the same pattern as recovery: detection
	power generally increases with \(p\) and decreases with \(q\).

	\section{Real-data analyses}
	
	\subsection{Trade data}
	\label{subsec:trade}
	
\begin{figure}[t]
		\centering
		\includegraphics[width=0.93\textwidth]{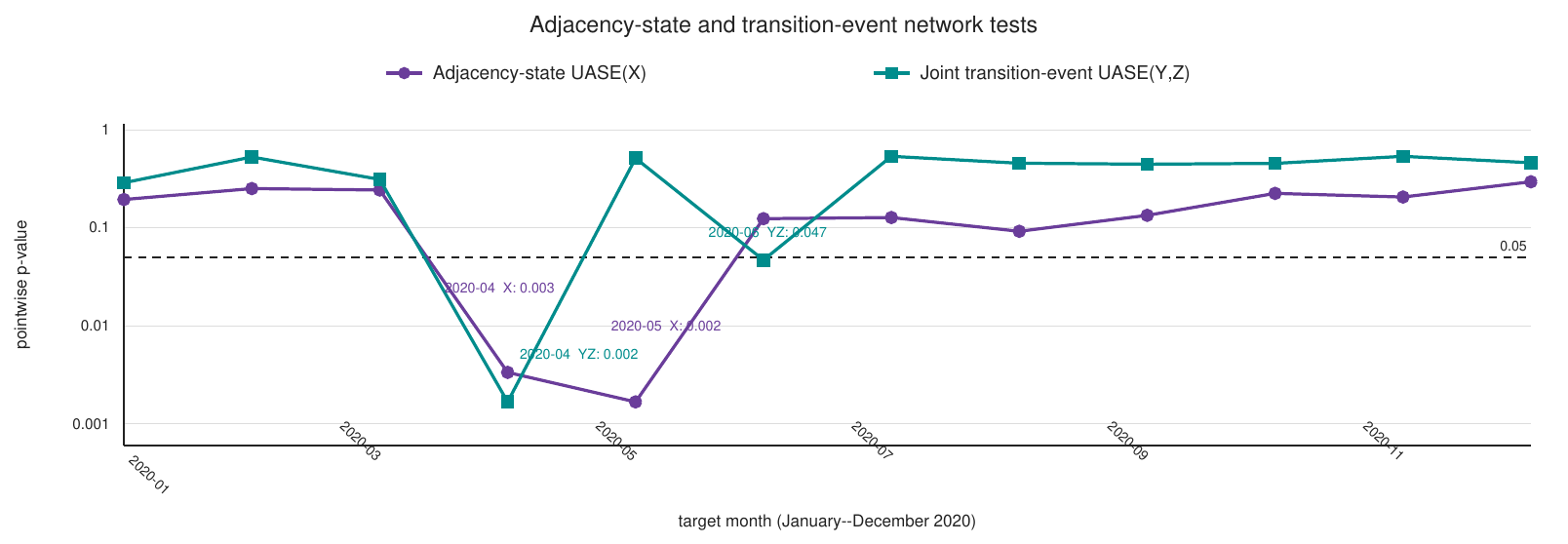}
		\vspace{0.8ex}
		\includegraphics[width=0.93\textwidth]{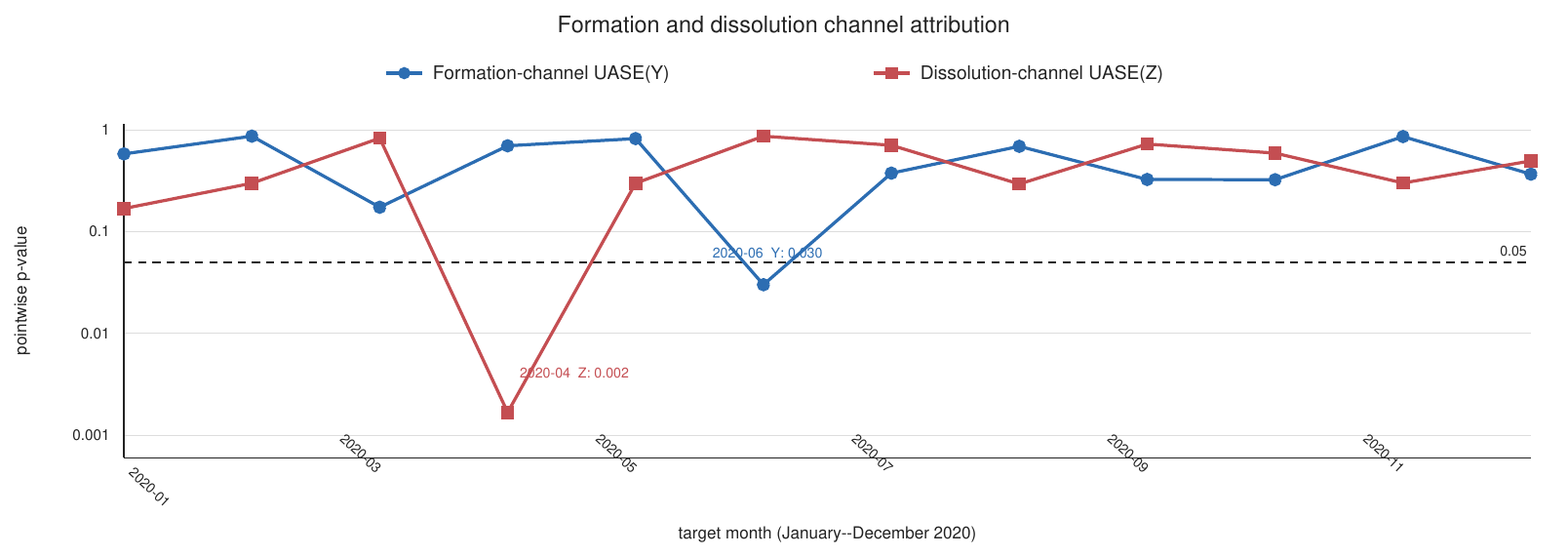}
		\caption{Four network-level trade tests.  The upper panel overlays adjacency-based UASE\((X)\)
			and joint transition-event UASE\((Y,Z)\); the lower panel overlays the
			formation UASE\((Y)\) and dissolution UASE\((Z)\) channels. }
		\label{fig:trade-uase-overlaid}
	\end{figure}
We analyse monthly bilateral merchandise trade from the United Nations
Comtrade Database \citep{uncomtrade2}, covering January 2017 to December 2020. From the 152 reporters observed at least once, we exclude the two aggregate reporters, the European Union and the Association of Southeast Asian Nations (ASEAN), together with any economy for which at least one month contains neither an export nor an import record. This leaves $p=104$ economies listed in Appendix S.6.1.  Let $E_{ij,m}$ be exports from economy $i$ to
$j$ in month $m$ reported by $i$, and let $I_{ji,m}$ be the corresponding
imports reported by $j$. Define \(D_{ij,m}\) as the average of the positive values among
\(E_{ij,m}\) and \(I_{ji,m}\). If only one is positive, we use that
value; if neither is positive, we set \(D_{ij,m}=0\). We then define the undirected bilateral
trade value $W_{ij,m}=D_{ij,m}+D_{ji,m}$ and place an edge when $W_{ij,m}$ is at
least US\$15 million. We use embedding dimension $D=8$. Using January 2017 only as the initial adjacency state, UASE jointly embeds the 47 monthly transition layers from February 2017 to December 2020, and the baseline contains all $|B|=35$ pre-COVID transitions from February 2017 to December 2019. We treat the twelve months of 2020 as the target months. No baseline month is
flagged by the network- or vertex-level stationarity checks, so the baseline passes
both stationarity checks.

The upper panel of Figure~\ref{fig:trade-uase-overlaid} compares the adjacency-based and joint transition-event tests. The adjacency-based test rejects in April and May, whereas the transition-event test rejects in April and June. The lower panel resolves this difference: dissolution rejects only in April and formation rejects only in June. In April, the dissolution channel is anomalous while the formation channel remains normal. The large number of  dissolutions moves the adjacency matrix \(X\) away from its baseline, so both the dissolution and adjacency-based tests reject. In May, both transition channels are normal, but \(X\) still retains the effect of the April anomaly; hence only the adjacency-based test rejects. In June, the formation channel is anomalous while the dissolution channel remains normal. The new formations move \(X\) back towards its baseline, so the formation test rejects but the adjacency-based test does not.

The channel-specific vertex tests identify the economies associated with the April and June network-level anomalies. Table~\ref{tab:trade-countries} reports all economies with a pointwise vertex \(p\)-value below \(0.05\). Eleven economies are flagged in the dissolution channel in April and six in the formation channel in June; Canada and France are flagged at both times.

	\begin{table}[H]
		\centering
		\footnotesize
		\setlength{\tabcolsep}{4pt}
		\begin{tabular}{@{}ll@{\hspace{0.8em}}p{0.66\textwidth}@{}}
			\hline
			Anomaly month & Channel (network \(p\)) & Countries flagged by the pointwise vertex test (vertex \(p\))\\
			\hline
			Apr. 2020 & \(Z\) (0.002) & CAN (0.014), COL (0.013), DEU (0.017), FRA (0.041), IDN (0.031), IND (\(<0.001\)), PHL (0.008), PRT (0.009), ROU (0.041), SVK (0.020), VNM (0.011)\\
			Jun. 2020 & \(Y\) (0.030) & CAN (0.012), CHL (0.027), FRA (0.021), ITA (0.049), MYS (0.027), THA (0.033)\\
			\hline
		\end{tabular}
		\caption{Pointwise country tests for the channel-specific
			trade-network anomaly months. }
		\label{tab:trade-countries}
	\end{table}

Sensitivity analyses for the edge threshold and embedding
dimension are reported in the Supplementary Material.

	\subsection{Primary-school contacts}
	The primary-school data can be downloaded from the SocioPatterns repository
\citep{sociopatterns_primary2}. They record face-to-face contacts among \(232\)
students in \(10\) classes over two school days at \(20\)-second resolution
\citep{stehle_primary_2011}. We form \(10\)-minute binary networks, placing an edge
between two students if they interact within the window. Transition
events are constructed only from consecutive windows within the same
day, excluding the overnight boundary.

	\begin{figure}[t]
		\centering
		\includegraphics[width=0.82\textwidth]{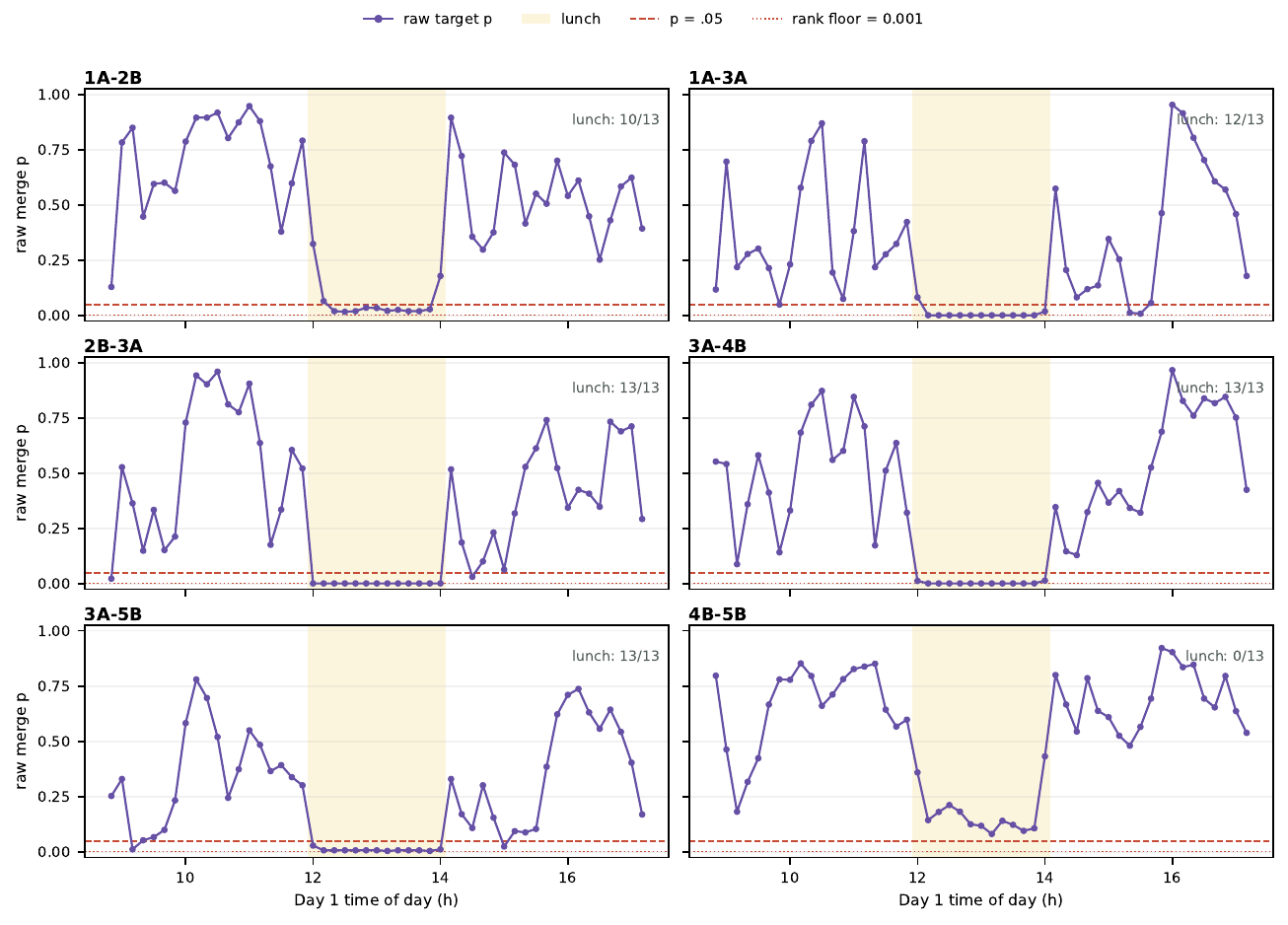}
		\caption{Primary-school pairwise community-merge tests on the first day. Each
panel reports raw pointwise target $p$-values for one class pair. Gold shading
marks lunch, the dashed line is 0.05, the dotted line is the minimum attainable empirical \(p\)-value based on the baseline pairs, and the annotation gives the number of rejected lunch transitions. }
		\label{fig:emp-primary-community-events}
	\end{figure}

Figure~10 of \citet{stehle_primary_2011} provides first-day
spatiotemporal trajectories for classes 1A, 2B, 3A, 4B, and 5B.
Together with Video~S1 of \citet{stehle_primary_2011}, these trajectories describe first-day contacts
and movements between classrooms, the courtyard, and the canteen,
providing external context for interpreting our detected anomalies. No
corresponding detailed description is available for the second day.
Because the two days involve the same students and school setting, and
the study reports similar contact summaries and class-level
distributions across them, we use the second day as the baseline and
the documented first day as the target.

Lunch runs from \(12{:}00\) to \(14{:}00\) in two consecutive turns,
although the study does not report the classes in each turn. The
contact matrix combining both days shows a broad separation between
grades 1--3 and grades 4--5, which the authors suggest may reflect the
lunch schedule. We therefore examine pairwise mixing among the five
classes documented in Figure~10.

	We use all \(38\) consecutive second-day transitions outside the lunch period
	to estimate the baseline and analyse all \(51\) consecutive first-day
	transitions as targets.  With \(q=10\) and embedding dimension \(D=10\) for each
	channel, the resulting partition has an ARI of \(0.960\)
	relative to the true labels.  The four misclassified students, whose true classes are 2A, 3B, 5A, and 5B, are all assigned
	to the estimated 4B community. 

The BH-adjusted baseline stationarity check passes at level \(0.05\). Figure~\ref{fig:emp-primary-community-events} displays six of the ten pairwise tests among these five classes. During lunch, merge signals are detected for 1A--2B, 1A--3A, and 2B--3A in \(10\), \(12\), and \(13\) of the \(13\) transitions, respectively.  Class 3A also shows merge signals with the upper-grade classes. The
merge nulls for 3A--4B and 3A--5B are rejected in all \(13\) lunch
transitions. By contrast, the merge nulls for 4B--5B and the four
omitted pairs (1A--4B, 1A--5B, 2B--4B, and 2B--5B) are not rejected
in any of the \(13\) lunch transitions.

These lunch-period merge patterns agree with the class-mixing structure reported in Figure~10 of \citet{stehle_primary_2011}. More specifically, Figure~10 shows 1A and 2B mainly in the cafeteria during 12:00--13:00 and in the playground during 13:00--14:00, with 4B and 5B following the reverse pattern; 3A is present in both locations during both lunch turns. The absence of a 4B–5B merge signal has a different interpretation. Although the two classes may share the same lunch space and interact directly, their class-level contact patterns do not become more similar than under the baseline.

\section{Discussion}

Transition events provide information not captured by individual adjacency matrices and therefore support anomaly tests for changes in network dynamics. Our vertex-, network- and community-level statistics test for departures from stationary baseline behaviour, while the separate formation and dissolution embeddings identify which component of the dynamics has changed. The rowwise theory provides high-probability guarantees for detecting and exactly localising anomalous vertices, explicitly separating interference from changes at other vertices and autoregressive memory from earlier anomalies. It gives conditions under which each contribution is negligible relative to the rowwise estimation error.

The current anomaly-detection theory keeps the number of observation times fixed as the network size grows. Deriving anomaly-detection guarantees when both dimensions diverge is an important extension. The transition-event representation could be generalised to detect anomalies in directed, weighted, and higher-order networks, including dynamic hypergraphs. The community-level tests could also be extended to estimate the number of communities rather than treating it as known.

\section*{Acknowledgement}
We gratefully acknowledge the support from EPSRC DASS Programme grant EP/Z531327/1. 

\section*{Supplementary material}

The supplementary material contains complete proofs of the theoretical results,
detailed algorithms, additional simulation studies, and, for the international
trade application, details of country coverage and robustness and sensitivity
analyses. 

\bibliographystyle{plainnat}
\bibliography{main}
\newpage
\section*{Supplementary material for “Anomaly detection in autoregressive networks” by Xianghe Zhu}
	
	\subsection*{S.0 Notation and preliminary matrix results}
	Throughout the paper, matrices and vectors are denoted by \emph{boldface} symbols
	(e.g., $\mathbf{M}$, $\mathbf{m}$), while scalars are written in nonbold font.
	For any matrix $\mathbf{M}$, its $(i,j)$-entry is denoted by $(\mathbf{M})_{ij}$.

	We write $\circ$ for the Hadamard (elementwise) product between matrices of the same dimension.
	For vectors $\mathbf{a},\mathbf{b}$ of the same length, $\mathbf{a}\odot \mathbf{b}$ denotes the elementwise product.
	Finally, $\oplus$ denotes the (block) direct sum: for matrices $\mathbf{A}_1,\ldots,\mathbf{A}_n$,
	$\mathbf{A}_1\oplus\cdots\oplus \mathbf{A}_n
	:=\operatorname{diag}(\mathbf{A}_1,\ldots,\mathbf{A}_n),$
	i.e., the block-diagonal matrix with diagonal blocks $\mathbf{A}_1,\ldots,\mathbf{A}_n$.
	
	We first record a simple algebraic lemma used repeatedly below.
	
	\begin{lemma}\label{lem:diag-hadamard}
		(i) Let $\mathbf{a}, \mathbf{c} \in \mathbb{R}^m$ and $\mathbf{b}, \mathbf{d} \in \mathbb{R}^n$ be column vectors. Then
		
		$$
		\left(\mathbf{a} \mathbf{b}^{\top}\right) \circ\left(\mathbf{c} \mathbf{d}^{\top}\right)=(\mathbf{a} \odot \mathbf{c})(\mathbf{b} \odot \mathbf{d})^{\top}.
		$$
		
		(ii) Let $\mathbf{X}\in\mathbb{R}^{m\times n}$ and let $\mathbf{a}\in\mathbb{R}^m$, $\mathbf{b}\in\mathbb{R}^n$ be column vectors.
		Then
		\begin{align}\label{eq:diag-hadamard}
			\operatorname{diag}(\mathbf{a})\,\mathbf{X}\,\operatorname{diag}(\mathbf{b})
			\;=\;
			\mathbf{X}\circ (\mathbf{a}\mathbf{b}^\top).
		\end{align}
		In particular, for any $\mathbf{L}\in\mathbb{R}^{p\times d}$,
		\begin{align*}
			(\mathbf{LL}^\top)\circ (\mathbf{a}\mathbf{b}^\top)
			\;=\;
			(\operatorname{diag}(\mathbf{a})\mathbf{L})\,(\operatorname{diag}(\mathbf{b})\mathbf{L})^\top.
		\end{align*}
	\end{lemma}
	
	\begin{proof}
		(i) Consider the $(i, j)$-entry on both sides. For the left-hand side,
		
		$$
		\left[\left(\mathbf{a} \mathbf{b}^{\top}\right) \circ\left(\mathbf{c} \mathbf{d}^{\top}\right)\right]_{i j}=\left(a_i b_j\right)\left(c_i d_j\right) .
		$$

		For the right-hand side,
		
		$$
		\left[(\mathbf{a} \odot \mathbf{c})(\mathbf{b} \odot \mathbf{d})^{\top}\right]_{i j}=\left(a_i c_i\right)\left(b_j d_j\right).
		$$
		
		Since $\left(a_i b_j\right)\left(c_i d_j\right)=\left(a_i c_i\right)\left(b_j d_j\right)$ for all $i=1, \ldots, m, j=1, \ldots, n$, the two matrices have identical entries. Hence they are equal.
		
		(ii) Compare the $(i,j)$-entries of both sides.
		On the left-hand side,
		$$
		\bigl[\operatorname{diag}(\mathbf{a})\,\mathbf{X}\,\operatorname{diag}(\mathbf{b})\bigr]_{ij}
		=
		a_i\,X_{ij}\,b_j.
		$$
		On the right-hand side,
		$$
		\bigl[\mathbf{X}\circ(\mathbf{a}\mathbf{b}^\top)\bigr]_{ij}
		=
		X_{ij}\,(\mathbf{a}\mathbf{b}^\top)_{ij}
		=
		X_{ij}\,a_i\,b_j.
		$$
		Hence the entries coincide for all $i,j$, which proves \eqref{eq:diag-hadamard}.
		
		Apply $\mathbf{X}=\mathbf{LL}^\top$:
		$$
		(\mathbf{LL}^\top)\circ(\mathbf{a}\mathbf{b}^\top)
		=
		\operatorname{diag}(\mathbf{a})\,(\mathbf{LL}^\top)\,\operatorname{diag}(\mathbf{b}).
		$$
		Finally, use associativity of matrix multiplication:
		$$
		\operatorname{diag}(\mathbf{a})\,(\mathbf{LL}^\top)\,\operatorname{diag}(\mathbf{b})
		=
		(\operatorname{diag}(\mathbf{a})\mathbf{L})\,(\mathbf{L}^\top\operatorname{diag}(\mathbf{b}))
		=
		(\operatorname{diag}(\mathbf{a})\mathbf{L})\,(\operatorname{diag}(\mathbf{b})\mathbf{L})^\top,
		$$
		since $\mathbf{L}^\top\operatorname{diag}(\mathbf{b})=(\operatorname{diag}(\mathbf{b})\mathbf{L})^\top$.
		
	\end{proof}

	\begin{lemma}[Matrix Bernstein]\label{lem:matrixBernstein}
		Let $\mathbf{M}_1, \ldots, \mathbf{M}_n$ be independent random matrices with common dimensions $m_1 \times m_2$, satisfying $\mathbb{E}\left[\mathbf{M}_k\right]=0$ and $\left\|\mathbf{M}_k\right\| \leq L$ for each $1 \leq k \leq n$, for some fixed value $L$.
		
		Let $\mathbf{M}=\sum \mathbf{M}_k$ and let $v(\mathbf{M})=\max \left\{\left\|\mathbb{E}\left[\mathbf{M} \mathbf{M}^{\top}\right]\right\|,\left\|\mathbb{E}\left[\mathbf{M}^{\top} \mathbf{M}\right]\right\|\right\}$ denote the matrix variance statistic of $\mathbf{M}$.  This is the rectangular form obtained by applying the self-adjoint dilation to the summands.  Then for all $\gamma \geq 0$, we have
		\[
		\mathbb{P}(\|\mathbf{M}\| \geq \gamma) \leq\left(m_1+m_2\right) \exp \left(\frac{-\gamma^2 / 2}{v(\mathbf{M})+L \gamma / 3}\right).
		\]
	\end{lemma}
	The next lemma is a variant of the Davis--Kahan theorem (see Theorem~4 in \citet{yu2015useful} and Theorem~12 in \citet{jones_multilayer_2021}).
	\begin{lemma}\label{lem:davis}
		Let $\mathbf{M}_1, \mathbf{M}_2 \in \mathbb{R}^{m \times n}$ have singular value decompositions
		
		$$
		\mathbf{M}_i=\mathbf{U}_i \boldsymbol{\Sigma}_i \mathbf{V}_i^{\top}+\mathbf{U}_{i, \perp} \boldsymbol{\Sigma}_{i, \perp} \mathbf{V}_{i, \perp}^{\top},
		$$
		
		where $\mathbf{U}_i \in \mathbb{O}(m \times d)$ has orthonormal columns corresponding to the $d$ largest singular values of $\mathbf{M}_i$, for some $1 \leq d \leq \min(m,n)$. Then, if $\sigma_d(\mathbf M_1)>\sigma_{d+1}(\mathbf M_1)$, we have
		
		$$
		\left\|\sin \Theta\left(\mathbf{U}_2, \mathbf{U}_1\right)\right\|_F \leq \frac{2 \sqrt{d}\left(2 \sigma_1\left(\mathbf{M}_1\right)+\left\|\mathbf{M}_2-\mathbf{M}_1\right\|\right)\left\|\mathbf{M}_2-\mathbf{M}_1\right\|}{\sigma_d\left(\mathbf{M}_1\right)^2-\sigma_{d+1}\left(\mathbf{M}_1\right)^2} .
		$$
		
		Here singular values beyond the rank are taken to be \(0\).
		
		Moreover, there exist orthogonal matrices $\mathbf{O}_\mathbf{U},\mathbf{ O}_\mathbf{V} \in \mathbb{R}^{d \times d}$ such that
		
		$$
		\|\mathbf{U}_2 \mathbf{O}_\mathbf{U}-\mathbf{U}_1\|_{\mathrm{F}} \leq \frac{2^{3/2} \sqrt{d}\left(2 \sigma_1\left(\mathbf{M}_1\right)+\left\|\mathbf{M}_2-\mathbf{M}_1\right\|\right)\left\|\mathbf{M}_2-\mathbf{M}_1\right\|}{\sigma_d\left(\mathbf{M}_1\right)^2-\sigma_{d+1}\left(\mathbf{M}_1\right)^2} ,
		$$
		$$
		\|\mathbf{V}_2 \mathbf{O}_\mathbf{V}-\mathbf{V}_1\|_{\mathrm{F}} \leq \frac{2^{3/2} \sqrt{d}\left(2 \sigma_1\left(\mathbf{M}_1\right)+\left\|\mathbf{M}_2-\mathbf{M}_1\right\|\right)\left\|\mathbf{M}_2-\mathbf{M}_1\right\|}{\sigma_d\left(\mathbf{M}_1\right)^2-\sigma_{d+1}\left(\mathbf{M}_1\right)^2} .
		$$
	\end{lemma}
	
	\subsection*{S.1 Supporting results and proof of Theorem~\ref{thm1}}

	\subsubsection*{S.1.1 Proof of Proposition \ref{prop:u1u2}}

	\begin{proof}
		For \(i\ne j\), the marginal probabilities satisfy
		\[
		\Pi_{ij}^{(t)}
		=
		\alpha_{ij}^{(t)}
		+
		\bigl(
		1-\alpha_{ij}^{(t)}-\beta_{ij}^{(t)}
		\bigr)
		\Pi_{ij}^{(t-1)}.
		\]
		By the definition of the completed diagonal, the same recursion
		holds for \(i=j\). Hence,
		$$
		\mathbf\Pi^{(t)}=\mathbf P_Y^{(t)}+\bigl(\mathbf 1\mathbf 1^\top-\mathbf P_Y^{(t)}-\mathbf P_Z^{(t)}\bigr)\circ \mathbf\Pi^{(t-1)}.
		$$
		By Theorem~5.1.7 of \citet{horn1994topics},
		$\operatorname{rank}(\mathbf A\circ \mathbf B)
		\le \operatorname{rank}(\mathbf A)\operatorname{rank}(\mathbf B)$.
		Moreover,
		$$
		\operatorname{rank}(\mathbf P_Y^{(t)})\le d,\qquad
		\operatorname{rank}(\mathbf P_Z^{(t)})\le d,\qquad
		\operatorname{rank}(\mathbf 1\mathbf 1^\top-\mathbf P_Y^{(t)}-\mathbf P_Z^{(t)})\le 1+2d.
		$$
		Set
		$$
		\bar r_0=r_0,
		\qquad
		\bar r_t=d+(1+2d)\bar r_{t-1},
		\quad t\ge1.
		$$
		Then $\operatorname{rank}(\mathbf\Pi^{(t)})\le\bar r_t$ for each fixed $t$.
		
		We further show by induction that $\mathbf\Pi^{(t)}$ admits a finite-rank decomposition
		$$
		\mathbf\Pi^{(t)}=\sum_{r=1}^{\bar r_t}\mathbf c_r^{(t)}\mathbf d_r^{(t)\top},
		\qquad
		\|\mathbf c_r^{(t)}\|_\infty,\|\mathbf d_r^{(t)}\|_\infty\le C_t
		$$
		for some finite $C_t$.
		The case $t=0$ holds by Condition~\ref{ass:initial-finite-rank}.
		Assume the claim holds for $t-1$, i.e. $\mathbf\Pi^{(t-1)}=\sum_{r=1}^{\bar r_{t-1}}\mathbf c_r^{(t-1)}\mathbf d_r^{(t-1)\top}$ with
		$\|\mathbf c_r^{(t-1)}\|_\infty,\|\mathbf d_r^{(t-1)}\|_\infty\le C_{t-1}$.
		
		Write
		$$
		\mathbf P_Y^{(t)}=\rho_Y\mathbf L_Y^{(t)}\mathbf L_Y^{(t)\top}
		=\rho_Y\sum_{k=1}^d \mathbf u_k^{(t)}\mathbf u_k^{(t)\top},
		\qquad
		\mathbf P_Z^{(t)}=\rho_Z\mathbf L_Z^{(t)}\mathbf L_Z^{(t)\top}
		=\rho_Z\sum_{k=1}^d \mathbf w_k^{(t)}\mathbf w_k^{(t)\top},
		$$
		where $\mathbf u_k^{(t)}$ and $\mathbf w_k^{(t)}$ are the $k$th columns of $\mathbf L_Y^{(t)}$ and $\mathbf L_Z^{(t)}$.
		Using Lemma~\ref{lem:diag-hadamard} (Hadamard product of outer products),
		\begin{align*}
			\mathbf\Pi^{(t)}
			=&\mathbf P_Y^{(t)}+\mathbf\Pi^{(t-1)}-\mathbf P_Y^{(t)}\circ\mathbf\Pi^{(t-1)}-\mathbf P_Z^{(t)}\circ\mathbf\Pi^{(t-1)}\\
			=&\sum_{k=1}^d (\sqrt{\rho_Y}\mathbf u_k^{(t)})(\sqrt{\rho_Y}\mathbf u_k^{(t)})^\top
			+\sum_{r=1}^{\bar r_{t-1}}\mathbf c_r^{(t-1)}\mathbf d_r^{(t-1)\top}\\
			&-\sum_{k=1}^d\sum_{r=1}^{\bar r_{t-1}}
			(\sqrt{\rho_Y}\mathbf u_k^{(t)}\odot \mathbf c_r^{(t-1)})
			(\sqrt{\rho_Y}\mathbf u_k^{(t)}\odot \mathbf d_r^{(t-1)})^\top\\
			&-\sum_{k=1}^d\sum_{r=1}^{\bar r_{t-1}}
			(\sqrt{\rho_Z}\mathbf w_k^{(t)}\odot \mathbf c_r^{(t-1)})
			(\sqrt{\rho_Z}\mathbf w_k^{(t)}\odot \mathbf d_r^{(t-1)})^\top.
		\end{align*}
		Under $\|\boldsymbol\ell_{Y,i}^{(t)}\|_\infty,\|\boldsymbol\ell_{Z,i}^{(t)}\|_\infty\le \ell$ and $\rho_Y,\rho_Z\le 1$, we have
		$\|\sqrt{\rho_Y}\mathbf u_k^{(t)}\|_\infty,\|\sqrt{\rho_Z}\mathbf w_k^{(t)}\|_\infty\le \ell$ and
		$\|\sqrt{\rho_Y}\mathbf u_k^{(t)}\odot \mathbf c_r^{(t-1)}\|_\infty\le \ell C_{t-1}$ (similarly for $\mathbf w_k^{(t)}$).
		Hence $\mathbf\Pi^{(t)}$ admits the desired finite-rank decomposition with some finite
		$C_t\le \max\{\ell,C_{t-1},\ell C_{t-1}\}$.
		
		\textbf{(i) For $\mathbf Y$.}
		Since $\mathbf Q_{\mathbf Y}^{(t)}=\mathbf 1\mathbf 1^\top-\mathbf\Pi^{(t-1)}$, write
		$$
		\mathbf Q_{\mathbf Y}^{(t)}=\sum_{r=1}^{m_{Y,t}}\mathbf a_{Y,r}^{(t)}\mathbf b_{Y,r}^{(t)\top},
		\qquad
		m_{Y,t}=1+\bar r_{t-1},
		$$
		where $\mathbf a_{Y,r}^{(t)}, \mathbf b_{Y,r}^{(t)} \in \mathbb{R}^p$ and $\|\mathbf a_{Y,r}^{(t)}\|_{\infty},\|\mathbf b_{Y,r}^{(t)}\|_{\infty}<C_Q$ with $C_Q = \max\{\max_t\{C_t\}, 1\}$.
		Lemma~\ref{lem:diag-hadamard} yields
		$$
		\mathbf S_{\mathbf Y}^{(t)}
		=\big(\rho_Y\mathbf L_Y^{(t)}\mathbf L_Y^{(t)\top}\big)\circ \mathbf Q_{\mathbf Y}^{(t)}
		=\rho_Y\sum_{r=1}^{m_{Y,t}}
		\big(\operatorname{diag}(\mathbf a_{Y,r}^{(t)})\mathbf L_Y^{(t)}\big)
		\big(\operatorname{diag}(\mathbf b_{Y,r}^{(t)})\mathbf L_Y^{(t)}\big)^\top.
		$$
		Define
		$$
		\mathbf A_{Y,r}^{(t)}:=\sqrt{\rho_Y}\operatorname{diag}(\mathbf a_{Y,r}^{(t)})\mathbf L_Y^{(t)}\in \mathbb{R}^{p \times d},\qquad
		\mathbf B_{Y,r}^{(t)}:=\sqrt{\rho_Y}\operatorname{diag}(\mathbf b_{Y,r}^{(t)})\mathbf L_Y^{(t)}\in \mathbb{R}^{p \times d},
		$$
		and concatenate
		$$
		\mathbf U_{Y,1}^{(t)}:=[\mathbf A_{Y,1}^{(t)},\ldots,\mathbf A_{Y,m_{Y,t}}^{(t)}]\in\mathbb R^{p\times dm_{Y,t}},\qquad
		\mathbf U_{Y,2}^{(t)}:=[\mathbf B_{Y,1}^{(t)},\ldots,\mathbf B_{Y,m_{Y,t}}^{(t)}]\in\mathbb R^{p\times dm_{Y,t}},
		$$
		so that $\mathbf S_{\mathbf Y}^{(t)}=\mathbf U_{Y,1}^{(t)}\mathbf U_{Y,2}^{(t)\top}$.
		
		Define $\overline{K}_Y =d\sum_{t=1}^n(1+\bar r_{t-1})$ and
		$\overline{\mathbf{U}}_{Y,1}:=\bigl[\mathbf{U_{Y,1}}^{(1)},\ldots,\mathbf{U_{Y,1}}^{(n)}\bigr]\in\mathbb R^{p\times \overline K_Y},
		\overline{\mathbf{U}}_{Y,2}:=\mathbf{U_{Y,2}}^{(1)}\oplus\cdots\oplus \mathbf{U_{Y,2}}^{(n)}\in\mathbb R^{np\times \overline K_Y},$
		so that
		$$\mathbf S_Y=[\mathbf S_Y^{(1)}|\cdots|\mathbf S_Y^{(n)}]=\overline{\mathbf{U}}_{Y,1}\overline{\mathbf{U}}_{Y,2}^\top.$$

		\textbf{(ii) For $\mathbf Z$.}
		Here $\mathbf Q_{\mathbf Z}^{(t)}=\mathbf\Pi^{(t-1)}=\sum_{r=1}^{\bar r_{t-1}}\mathbf c_r^{(t-1)}\mathbf d_r^{(t-1)\top}$.
		Then
		$$
		\mathbf S_{\mathbf Z}^{(t)}=\big(\rho_Z\mathbf L_Z^{(t)}\mathbf L_Z^{(t)\top}\big)\circ \mathbf Q_{\mathbf Z}^{(t)}
		=\sum_{r=1}^{\bar r_{t-1}}\mathbf C_{Z,r}^{(t)}\mathbf D_{Z,r}^{(t)\top},
		$$
		where we define
		$$
		\mathbf C_{Z,r}^{(t)}:=\sqrt{\rho_Z}\operatorname{diag}(\mathbf c_r^{(t-1)})\mathbf L_Z^{(t)}\in \mathbb{R}^{p \times d},\qquad
		\mathbf D_{Z,r}^{(t)}:=\sqrt{\rho_Z}\operatorname{diag}(\mathbf d_r^{(t-1)})\mathbf L_Z^{(t)}\in \mathbb{R}^{p \times d}.
		$$
		Concatenate
		$$
		\mathbf U_{Z,1}^{(t)}:=[\mathbf C_{Z,1}^{(t)},\ldots,\mathbf C_{Z,\bar r_{t-1}}^{(t)}]\in\mathbb R^{p\times d\bar r_{t-1}},\qquad
		\mathbf U_{Z,2}^{(t)}:=[\mathbf D_{Z,1}^{(t)},\ldots,\mathbf D_{Z,\bar r_{t-1}}^{(t)}]\in\mathbb R^{p\times d\bar r_{t-1}},
		$$
		so that $\mathbf S_{\mathbf Z}^{(t)}=\mathbf U_{Z,1}^{(t)}\mathbf U_{Z,2}^{(t)\top}$.
		
		Define $\overline K_Z =d\sum_{t=1}^n\bar r_{t-1}$ and
		$\overline{\mathbf{U}}_{Z,1}:=\bigl[\mathbf{U_{Z,1}}^{(1)},\ldots,\mathbf{U_{Z,1}}^{(n)}\bigr]\in\mathbb R^{p\times \overline K_Z},
		\overline{\mathbf{U}}_{Z,2}:=\mathbf{U_{Z,2}}^{(1)}\oplus\cdots\oplus \mathbf{U_{Z,2}}^{(n)}\in\mathbb R^{np\times \overline K_Z},$
		so that
		$$\mathbf{S_Z}=[\mathbf{S_Z}^{(1)}|\cdots|\mathbf{S_Z}^{(n)}]=\overline{\mathbf{U}}_{Z,1}\overline{\mathbf{U}}_{Z,2}^\top.$$
		
	\end{proof}
	\subsubsection*{S.1.2 Proof of Theorem \ref{thm1}}
	The following propositions and proofs are stated for $\{\mathbf Y^{(t)}\}_{t=1}^n$.
	The corresponding results for $\{\mathbf Z^{(t)}\}_{t=1}^n$ follow by the same arguments and are therefore omitted.
	 The Bernstein and Hoeffding bounds below can then be taken with constants depending on any fixed \(c_0>0\), and a finite union bound over ranks, blocks and leave-one-row indices gives failure probability \(O(p^{-c_0})\). The \(O_{\mathbb P}\) displays are shorthand for these high-probability bounds.

	\begin{lemma}\label{lem:fixed-n-dyad-block-concentration}
		Let $\mathbf S_{Y,0}$ be the off-diagonal part of $\mathbf S_Y$ in
		the unfolded $Y$ channel.  For every fixed \(c_0>0\), there is a constant
		\(C_{c_0}<\infty\) such that,
		\[
		\mathbb P\!\left(
		\|\mathbf Y-\mathbf S_{Y,0}\|
		>
		C_{c_0}(p\rho_Y\log p)^{1/2} \right)
		\le
		C_{c_0}p^{-c_0}.
		\]
		The same statement holds for the $Z$ channel with $\rho_Y$ replaced by
		$\rho_Z$.
	\end{lemma}
	\begin{proof}
		Write
		\(
		\mathbf M:=\mathbf Y-\mathbf S_{Y,0}=\sum_{1\le i<j\le p}\mathbf M_{ij},
		\)
		where $\mathbf M_{ij}\in\mathbb R^{p\times np}$ is the dyad block whose
		$t$th $p\times p$ block has nonzero entries only at $(i,j)$ and $(j,i)$, equal
		to $(\mathbf Y^{(t)})_{ij}-(\mathbf S_Y^{(t)})_{ij}$.  The dyad blocks are independent across unordered pairs, while
		within a dyad the coordinates over $t=1,\ldots,n$ may be dependent.  The proof
		therefore treats each unordered dyad as one matrix-valued summand. 
		
		Each block satisfies $\|\mathbf M_{ij}\|\le \|\mathbf M_{ij}\|_F\le C\sqrt n$.
		For the row variance term, only rows $i$ and $j$ can be nonzero, and
		$\operatorname{Var}\big((\mathbf Y^{(t)})_{ij}\big)\le C\rho_Y$.  Hence
		\(
		\left\|\sum_{i<j}\mathbb E(\mathbf M_{ij}\mathbf M_{ij}^{\top})\right\|
		\le Cnp\rho_Y.
		\)
		For the column variance term, $\mathbb E(\mathbf M_{ij}^{\top}\mathbf M_{ij})$
		can contain covariance terms across transition times for the same dyad.  These
		terms form fixed-size $n\times n$ covariance blocks supported on the two endpoint
		coordinates.  Since $n$ is fixed and the transition-event variances are bounded
		by $C\rho_Y$, their operator norms are bounded by $Cn\rho_Y$, and therefore
		\(
		\left\|\sum_{i<j}\mathbb E(\mathbf M_{ij}^{\top}\mathbf M_{ij})\right\|
		\le Cnp\rho_Y.
		\)
		Thus the matrix Bernstein variance is $v(\mathbf M)\le Cnp\rho_Y$ and the
		uniform summand bound is $L\le C\sqrt n$.  Matrix Bernstein gives
		\[
		\mathbb P\{\|\mathbf M\|\ge x\}
		\le (n+1)p
		\exp\left[-\frac{x^2/2}{Cnp\rho_Y+C\sqrt n\,x}\right].
		\]
		Taking $x=C_1(np\rho_Y\log p)^{1/2}$ with \(C_1\) large enough for the prescribed
		\(c_0\), and using $p\rho_Y/\log p\to\infty$, gives the displayed
polynomial tail, with constants depending on fixed $n$.  The same dyad-block
		construction applied to $\mathbf Z-\mathbf S_{Z,0}$ gives the $Z$ statement
		because $Z$ is the dissolution event matrix and its transition-event variances
		are bounded by $C\rho_Z$.
	\end{proof}
	\begin{proposition}\label{prop:boundy-p}
		For every fixed \(c_0>0\), there is a constant \(C_{c_0}<\infty\) such that,
		\[
		\mathbb P\!\left(
		\|\mathbf{Y}-\mathbf S_Y\|
		>
		C_{c_0}(\rho_Yp\log p)^{1/2} \right)
		\le
		C_{c_0}p^{-c_0}.
		\]
	\end{proposition}
	\begin{proof}
		By construction, $\mathbf{Y}^{(t)}$ has zero diagonal while $\mathbf S_Y^{(t)}$ may have a nonzero diagonal.  Let
		$\mathbf{D_Y}^{(t)}:=\operatorname{diag}(\mathbf S_Y^{(t)})$,
		$\mathbf{D_Y}:=[\mathbf{D_Y}^{(1)}|\cdots|\mathbf{D_Y}^{(n)}]$, and
		$\mathbf S_{Y,0}:=\mathbf S_Y-\mathbf D_Y$.  Then
		\(
		\mathbf Y-\mathbf S_Y=(\mathbf Y-\mathbf S_{Y,0})-\mathbf D_Y.
		\)
		By boundedness of the latent positions and $\rho_Y\in(0,1]$,
		$\|\mathbf D_Y\|=O(\rho_Y)$.
		Lemma~\ref{lem:fixed-n-dyad-block-concentration} gives the displayed
tail bound for $\|\mathbf Y-\mathbf S_{Y,0}\|$.
		Since $p\rho_Y/\log p\to\infty$, the diagonal term is dominated by this rate,
		which proves the proposition.
	\end{proof}
	\begin{proposition}\label{prop:sigmapy}
		The $K_Y$ leading singular values $\sigma_i(\mathbf{Y})$ satisfy $\sigma_i\left(\mathbf{Y}\right)=O_{\mathbb P}\left(p \rho_Y\right)$ and $\sigma_i\left(\mathbf{Y}\right)=\Omega_{\mathbb P}\left(p \rho_Y\right).$
	\end{proposition}
	\begin{proof}
		By Corollary 7.3.5 of \citet{horn2012matrix}, the ordered singular values of
		two matrices satisfy
		$|\sigma_i(\mathbf A)-\sigma_i(\mathbf B)|\le \|\mathbf A-\mathbf B\|_2$.
		Hence, for each $i\le K_Y$,
		\[
		\left|\sigma_i(\mathbf Y)-\sigma_i(\mathbf S_Y)\right|
		\le
		\|\mathbf Y-\mathbf S_Y\|.
		\]
		The conclusion follows from Condition~\ref{ass:3}
		and Proposition~\ref{prop:boundy-p}.
	\end{proof}

	\begin{proposition}\label{prop:uy-pv}
		$\left\|\mathbf{U}_{S_Y}^{\top}(\mathbf{Y}-\mathbf S_Y) \mathbf{V}_{S_Y}\right\|_F=O_{\mathbb P}\left(\log ^{1 / 2}(p)\right)$.
	\end{proposition}
	\begin{proof}
		Write the right singular matrix in blocks
		$\mathbf V_{S_Y}^{(1)},\ldots,\mathbf V_{S_Y}^{(n)}$.  For fixed
		$(i,j)$, we can write
		\(
		\left(\mathbf{U}_{S_Y}^{\top}(\mathbf{Y}-\mathbf S_Y) \mathbf{V}_{S_Y}\right)_{i,j}
		=E_{ij}-\sum_{t=1}^nF_{ij}^t,
		\)
		where
		\(
		E_{ij}
		=
		\sum_{a=1}^p\sum_{b\ne a}E^{ij}_{ab},
		E^{ij}_{ab}
		=
		\sum_{t=1}^n
		(\mathbf{U}_{S_Y})_{ai}
		(\mathbf{Y}_{ab}^t-(\mathbf S_Y^t)_{ab})
		(\mathbf{V}_{S_Y}^{(t)})_{bj},
		\)
		and
		\(
		F_{ij}^t
		=
		\sum_{a=1}^p
		(\mathbf{U}_{S_Y})_{ai}
		(\mathbf S_Y^t)_{aa}
		(\mathbf{V}_{S_Y}^{(t)})_{aj}.
		\)
		By the Cauchy--Schwarz inequality,
\begin{align*}
			|\sum_{t=1}^nF_{ij}^t|&\leq|\sum_{t=1}^n\sum_{a=1}^p(\mathbf{U}_{S_Y})_{ai}(\mathbf S_Y^t)_{aa}(\mathbf{V}_{S_Y}^{(t)})_{aj}|\\
			&\leq \left(\sum_{t=1}^n\sum_{a=1}^p(\mathbf{U}_{S_Y})^2_{ai}(\mathbf{S_Y^t})^2_{aa}\right)^{1/2}\left(\sum_{t=1}^n\sum_{a=1}^p(\mathbf{V}_{S_Y}^{(t)})^2_{aj}\right)^{1/2}\\
			&\leq\max_{t,a}\left( \mathbf{S_Y^t}\right)_{aa} \left(\sum_{t=1}^n\sum_{a=1}^p(\mathbf{U}_{S_Y})^2_{ai}\right)^{1/2}\left(\sum_{t=1}^n\sum_{a=1}^p(\mathbf{V}_{S_Y}^{(t)})^2_{aj}\right)^{1/2}\\
			& =O_{\mathbb P}(\rho_Y).
		\end{align*}
		The sum $\sum_{a=1}^p\sum_{b=a+1}^pE^{ij}_{ab}$ is a sum of independent
		mean-zero random variables, each bounded in absolute value by
		$\sum_{t=1}^n|(\mathbf{U}_{S_Y})_{ai}(\mathbf{V}_{S_Y}^{(t)})_{bj}|$.
		By Hoeffding's inequality,
		\[
		\mathrm{P}\left(\left|S_n-\mathrm{E}\left[S_n\right]\right| \geq t\right)
		\leq
		2 \exp \left(-\frac{2 t^2}{\sum_{i=1}^n\left(b_i-a_i\right)^2}\right).
		\]
		Hence, for some constant $C_2>0$,
		\begin{align*}
			\mathbb{P}\left(\left|\sum_{a=1}^p\sum_{b=a+1}^pE^{ij}_{ab}\right| \geq C_2n\sqrt{\log(p)}\right) \leq & 2 \exp \left(-\frac{2 C_2^2n^2\log(p)}{4\sum_{a=1}^p\sum_{b=a+1}^p(\sum_{t=1}^n|(\mathbf{U}_{S_Y})_{ai}(\mathbf{V}_{S_Y}^{(t)})_{bj}|)^2}\right) \\
			\leq & 2 \exp \left(-\frac{2 C_2^2n^2\log(p)}{4n\sum_{a=1}^p|(\mathbf{U}_{S_Y})_{ai}|^2\sum_{t=1}^n\sum_{b=1}^p|(\mathbf{V}_{S_Y}^{(t)})_{bj}|^2}\right) \\
			\leq & 2 \exp \left(-\frac{2 C_2^2n^2\log(p)}{4n^2}\right).
		\end{align*}
		Thus
		$|\sum_{a=1}^p\sum_{b=a+1}^pE^{ij}_{ab}|
		=O_{\mathbb P}(\sqrt{\log p})$.  The same argument gives
		$|\sum_{b=1}^p\sum_{a=b+1}^pE^{ij}_{ab}|
		=O_{\mathbb P}(\sqrt{\log p})$.  Since the diagonal correction is
		$O_{\mathbb P}(\rho_Y)$, it follows that
		\(
		\left(\mathbf{U}_{S_Y}^{\top}(\mathbf{Y}-\mathbf S_Y) \mathbf{V}_{S_Y}\right)_{i,j}
		=
		O_{\mathbb P}(\sqrt{\log p}).
		\)
		Therefore
		\begin{align*}
			\|\mathbf{U}_{S_Y}^{\top}(\mathbf{Y}-\mathbf S_Y) \mathbf{V}_{S_Y}\|_F = \sqrt{\sum_{i,j}|\left(\mathbf{U}_{S_Y}^{\top}(\mathbf{Y}-\mathbf S_Y) \mathbf{V}_{S_Y}\right)_{i,j}|^2} = O_{\mathbb P}(\sqrt{\log p})
		\end{align*}
		and the high-probability version follows by taking the Hoeffding
		constant large enough for the prescribed \(c_0\) and applying the finite union
		bound over \(i,j\le K_Y\).
	\end{proof}

	\begin{proposition}\label{prop:bounduv}
		The following bounds hold:
		
		(i) $\left\|\mathbf{U}_{\mathbf{Y}} \mathbf{U}_{\mathbf{Y}}^{\top}-\mathbf{U}_{S_Y} \mathbf{U}_{S_Y}^{\top}\right\|,\left\|\mathbf{V}_{\mathbf{Y}} \mathbf{V}_{\mathbf{Y}}^{\top}-\mathbf{V}_{S_Y} \mathbf{V}_{S_Y}^{\top}\right\|=O_{\mathbb P}\left(\frac{ \log ^{1 / 2}(p)}{\rho_Y^{1 / 2} p^{1 / 2}}\right)$.
		
		(ii) $\left\|\mathbf{U}_{\mathbf{Y}}-\mathbf{U}_{S_Y} \mathbf{U}_{S_Y}^{\top} \mathbf{U}_{\mathbf{Y}}\right\|_F,\left\|\mathbf{V}_{\mathbf{Y}}-\mathbf{V}_{S_Y} \mathbf{V}_{S_Y}^{\top} \mathbf{V}_{\mathbf{Y}}\right\|_F=O_{\mathbb P}\left(\frac{ \log ^{1 / 2}(p)}{\rho_Y^{1 / 2}p^{1 / 2}}\right)$.
		
		(iii) $\left\|\mathbf{U}_{S_Y}^{\top} \mathbf{U}_{\mathbf{Y}} \boldsymbol{\Sigma}_{\mathbf{Y}}-\boldsymbol{\Sigma}_{S_Y} \mathbf{V}_{S_Y}^{\top} \mathbf{V}_{\mathbf{Y}}\right\|_F,\left\|\boldsymbol{\Sigma}_{S_Y} \mathbf{U}_{S_Y}^{\top} \mathbf{U}_{\mathbf{Y}}-\mathbf{V}_{S_Y}^{\top} \mathbf{V}_{\mathbf{Y}} \boldsymbol{\Sigma}_{\mathbf{Y}}\right\|_F=O_{\mathbb P}\left( \log (p)\right)$.
		
		(iv) $\left\|\mathbf{U}_{S_Y}^{\top} \mathbf{U}_{\mathbf{Y}}-\mathbf{V}_{S_Y}^{\top} \mathbf{V}_{\mathbf{Y}}\right\|_F=O_{\mathbb P}\left(\frac{ \log (p)}{\rho_Y p }\right)$
	\end{proposition}
	\begin{proof}
		(i) Let
		\(
		P_Y=\mathbf{U}_{\mathbf{Y}} \mathbf{U}_{\mathbf{Y}}^{\top},
		P_{S_Y}=\mathbf{U}_{S_Y} \mathbf{U}_{S_Y}^{\top}
		\)
		be the two projection matrices.  Write the singular value decomposition
		\[
		\mathbf{U}_{\mathbf{Y}}^{\top} \mathbf{U}_{S_Y}
		=
		\mathbf{W_{Y,1}} \mathbf{\Sigma}_\mathbf{U} \mathbf{W_{Y,2}}^{\top},
		\]
		where the diagonal entries of \(\mathbf\Sigma_{\mathbf U}\) are
		\((\sigma_\mathbf{U})_1, \ldots, (\sigma_\mathbf{U})_{K_Y}\).  Since the
		Frobenius norm is unitarily invariant,
		\begin{align*}
			\|\mathbf{U}_{\mathbf{Y}} \mathbf{U}_{\mathbf{Y}}^{\top}-\mathbf{U}_{S_Y} \mathbf{U}_{S_Y}^{\top}\|_F^2 
			& =\operatorname{tr}\left((\mathbf{U}_{\mathbf{Y}} \mathbf{U}_{\mathbf{Y}}^{\top}-\mathbf{U}_{S_Y} \mathbf{U}_{S_Y}^{\top})(\mathbf{U}_{\mathbf{Y}} \mathbf{U}_{\mathbf{Y}}^{\top}-\mathbf{U}_{S_Y} \mathbf{U}_{S_Y}^{\top})^{\top}\right) \nonumber\\\nonumber
			& =2 K_Y-2 \operatorname{tr}\left(\mathbf{U}_{\mathbf{Y}} \mathbf{U}_{\mathbf{Y}}^{\top} \mathbf{U}_{S_Y} \mathbf{U}_{S_Y}^{\top}\right) \\\nonumber
			& =2 K_Y-2\left\|\mathbf{U}_{\mathbf{Y}}^{\top} \mathbf{U}_{S_Y}\right\|_F^2
			=2 K_Y-2\|\mathbf{\Sigma}_\mathbf{U}\|_F^2 \\\nonumber
			& =2 K_Y-2\left\|\cos \left(\Theta_{\mathbf{U}_{\mathbf{Y}}, \mathbf{U}_{S_Y}}\right)\right\|_F^2  =2 K_Y-2\sum_{i=1}^{K_Y} \cos ^2\left(\theta^i_{\mathbf{U}_{\mathbf{Y}}, \mathbf{U}_{S_Y}}\right)  \\&
			= 2\left\|\sin \left(\Theta_{\mathbf{U}_{\mathbf{Y}}, \mathbf{U}_{S_Y}}\right)\right\|_F^2 
		\end{align*}
		where $$
		\Theta_{\mathbf{U}_{\mathbf{Y}},\mathbf{U}_{S_Y}}:=\left(\Theta_{\mathbf{U}_{\mathbf{Y}}, \mathbf{U}_{S_Y}}^1, \ldots, \Theta_{\mathbf{U}_{\mathbf{Y}}, \mathbf{U}_{S_Y}}^{K_Y}\right)^{\top}=\left(\arccos \left((\sigma_\mathbf{U})_1\right), \ldots, \arccos \left((\sigma_\mathbf{U})_{K_Y}\right)\right)^{\top} 
		$$ is the column vector of the principal angles.
		
		Furthermore, by Lemma~\ref{lem:davis} and Propositions~\ref{prop:boundy-p} and~\ref{prop:sigmapy},
		\begin{align*}
			\|\mathbf{U}_{\mathbf{Y}} \mathbf{U}_{\mathbf{Y}}^{\top}-\mathbf{U}_{S_Y} \mathbf{U}_{S_Y}^{\top}\|_F & = \sqrt{2}\left\|\sin \left(\Theta_{\mathbf{U}_{\mathbf{Y}}, \mathbf{U}_{S_Y}}\right)\right\|_F \\
			& \leq \frac{2 \sqrt{2}\sqrt{K_Y}\left(2 \sigma_1(\mathbf S_Y)+\|\mathbf{Y}-\mathbf S_Y\|\right)\|\mathbf{Y}-\mathbf S_Y\|}{\sigma_{K_Y}(\mathbf S_Y)^2}\\
			& = O_{\mathbb P}\left(\frac{ \log ^{1 / 2}(p)}{\rho_Y^{1 / 2} p^{1 / 2}}\right).
		\end{align*}
		
		(ii)  Since
		\(
		\left\|\mathbf{U}_{\mathbf{Y}}\right\|_F
		=
		\sqrt{\operatorname{trace}\left(\mathbf{U}_{\mathbf{Y}}^{\top} \mathbf{U}_{\mathbf{Y}}\right)}
=
\sqrt{K_Y}
=
O(1),
		\)
		the bound from part (i) gives
		$$
		\left\|\mathbf{U}_{\mathbf{Y}}-\mathbf{U}_{S_Y} \mathbf{U}_{S_Y}^{\top} \mathbf{U}_{\mathbf{Y}}\right\|_F=\left\|\left(\mathbf{U}_{\mathbf{Y}} \mathbf{U}_{\mathbf{Y}}^{\top}-\mathbf{U}_{S_Y} \mathbf{U}_{S_Y}^{\top}\right) \mathbf{U}_{\mathbf{Y}}\right\|_F=O_{\mathbb P}\left(\frac{ \log ^{1 / 2}(p)}{\rho_Y^{1 / 2} p^{1 / 2}}\right)
		$$
		
		(iii) Part (ii), Proposition~\ref{prop:boundy-p}, and
		Proposition~\ref{prop:uy-pv} give
		\begin{align*}
			&\left\|\mathbf{U}_{S_Y}^{\top} \mathbf{U}_{\mathbf{Y}} \boldsymbol{\Sigma}_{\mathbf{Y}}-\boldsymbol{\Sigma}_{S_Y} \mathbf{V}_{S_Y}^{\top} \mathbf{V}_{\mathbf{Y}}\right\|_F= \left\|\mathbf{U}_{S_Y}^{\top}(\mathbf{Y}-\mathbf S_Y) \mathbf{V}_{\mathbf{Y}} \right\|_F\\=& \left\|\mathbf{U}_{S_Y}^{\top}(\mathbf{Y}-\mathbf S_Y)\left(\mathbf{V}_{\mathbf{Y}}-\mathbf{V}_{S_Y} \mathbf{V}_{S_Y}^{\top} \mathbf{V}_{\mathbf{Y}}\right)  +\mathbf{U}_{S_Y}^{\top}(\mathbf{Y}-\mathbf S_Y) \mathbf{V}_{S_Y} \mathbf{V}_{S_Y}^{\top} \mathbf{V}_{\mathbf{Y}}\right\|_F \\
			\leq&\left\|\mathbf{Y}-\mathbf S_Y\right\|\left\|\left(\mathbf{V}_{\mathbf{Y}}-\mathbf{V}_{S_Y} \mathbf{V}_{S_Y}^{\top} \mathbf{V}_{\mathbf{Y}}\right)\right\|_F  +\left\|\mathbf{U}_{S_Y}^{\top}(\mathbf{Y}-\mathbf S_Y) \mathbf{V}_{S_Y}\right\|_F \left\|\mathbf{V}_{S_Y}^{\top} \mathbf{V}_{\mathbf{Y}}\right\|_F\\
			=&O_{\mathbb P}\left( \log (p)\right).
		\end{align*}
		
		(iv) Use the identity
		$$
		\begin{aligned}
			\mathbf{U}_{S_Y}^{\top} \mathbf{U}_{\mathbf{Y}}-\mathbf{V}_{S_Y}^{\top} \mathbf{V}_{\mathbf{Y}}= & \left(\left(\mathbf{U}_{S_Y}^{\top} \mathbf{U}_{\mathbf{Y}} \boldsymbol{\Sigma}_{\mathbf{Y}}-\boldsymbol{\Sigma}_{S_Y} \mathbf{V}_{S_Y}^{\top} \mathbf{V}_{\mathbf{Y}}\right)+\left(\boldsymbol{\Sigma}_{S_Y} \mathbf{U}_{S_Y}^{\top} \mathbf{U}_{\mathbf{Y}}\right.\right. \\
			& \left.\left.-\mathbf{V}_{S_Y}^{\top} \mathbf{V}_{\mathbf{Y}} \boldsymbol{\Sigma}_{\mathbf{Y}}\right)\right) \boldsymbol{\Sigma}_{\mathbf{Y}}^{-1}-\boldsymbol{\Sigma}_{S_Y}\left(\mathbf{U}_{S_Y}^{\top} \mathbf{U}_{\mathbf{Y}}-\mathbf{V}_{S_Y}^{\top} \mathbf{V}_{\mathbf{Y}}\right) \boldsymbol{\Sigma}_{\mathbf{Y}}^{-1} .
		\end{aligned}
		$$
		Using $\left|X_{i j}\right| \leq\|\mathbf{X}\|_F$,
		$$
		\begin{aligned}
			\left|\left(\mathbf{U}_{S_Y}^{\top} \mathbf{U}_{\mathbf{Y}}-\mathbf{V}_{S_Y}^{\top} \mathbf{V}_{\mathbf{Y}}\right)_{i j}\right|\left(1+\frac{\sigma_i(\mathbf S_Y)}{\sigma_j(\mathbf{Y})}\right) \leq & \left(\left\|\mathbf{U}_{S_Y}^{\top} \mathbf{U}_{\mathbf{Y}} \boldsymbol{\Sigma}_{\mathbf{Y}}-\boldsymbol{\Sigma}_{S_Y} \mathbf{V}_{S_Y}^{\top} \mathbf{V}_{\mathbf{Y}}\right\|_F\right. \\
			& \left.+\left\|\boldsymbol{\Sigma}_{S_Y} \mathbf{U}_{S_Y}^{\top} \mathbf{U}_{\mathbf{Y}}-\mathbf{V}_{S_Y}^{\top} \mathbf{V}_{\mathbf{Y}} \boldsymbol{\Sigma}_{\mathbf{Y}}\right\|_F\right)\left\|\boldsymbol{\Sigma}_{\mathbf{Y}}^{-1}\right\|_F.
		\end{aligned}
		$$
		
		Since $\left(1+\frac{\sigma_i(\mathbf S_Y)}{\sigma_j(\mathbf{Y})}\right) \geq 1$,
		part (iii) and Proposition~\ref{prop:sigmapy} give
		$$
		\left|\left(\mathbf{U}_{S_Y}^{\top} \mathbf{U}_{\mathbf{Y}}-\mathbf{V}_{S_Y}^{\top} \mathbf{V}_{\mathbf{Y}}\right)_{i j}\right|=O_{\mathbb P}\left(\frac{\log (p)}{\rho_Y p}\right)
		$$
		and the asserted Frobenius-norm bound follows. 
	\end{proof}
	
	The next result identifies the orthogonal matrix $\mathbf{W_{YS}}$ used for the
simultaneous Procrustes alignment of $\mathbf{U}_{\mathbf{Y}}$ with
$\mathbf{U}_{S_Y}$ and of $\mathbf{V}_{\mathbf{Y}}$ with
$\mathbf{V}_{S_Y}$.
	
	\begin{proposition}\label{prop:uvw}
		Let $\mathbf{U}_{S_Y}^{\top} \mathbf{U}_{\mathbf{Y}}+\mathbf{V}_{S_Y}^{\top} \mathbf{V}_{\mathbf{Y}}$ admit the singular value decomposition
		
		$$
		\mathbf{U}_{S_Y}^{\top} \mathbf{U}_{\mathbf{Y}}+\mathbf{V}_{S_Y}^{\top} \mathbf{V}_{\mathbf{Y}}=\mathbf{W_{Y,1}} \boldsymbol{\Lambda}_Y \mathbf{W_{Y,2}}^{\top}
		$$
		
		and let $\mathbf{W_{YS}}=\mathbf{W_{Y,1}} \mathbf{W_{Y,2}}^{\top}$. Then
		
		$$
		\max \left\{\left\|\mathbf{U}_{S_Y}^{\top} \mathbf{U}_{\mathbf{Y}}-\mathbf{W_{YS}}\right\|_F,\left\|\mathbf{V}_{S_Y}^{\top} \mathbf{V}_{\mathbf{Y}}-\mathbf{W_{YS}}\right\|_F\right\}=O_{\mathbb P}\left(\frac{\log p}{\rho_Y p}\right).
		$$
	\end{proposition}
	\begin{proof}
		Let $\mathbf{U}_{S_Y}^{\top} \mathbf{U}_{\mathbf{Y}}=\mathbf{W}_{\mathbf{U}, 1} \boldsymbol{\Sigma}_{\mathbf{U}} \mathbf{W}_{\mathbf{U}, 2}^{\top}$ be the singular value decomposition of $\mathbf{U}_{S_Y}^{\top} \mathbf{U}_{\mathbf{Y}}$, and define $\mathbf{W}_{\mathbf{U}} \in \mathbb{O}(K_Y)$ by $\mathbf{W}_{\mathbf{U}}=\mathbf{W}_{\mathbf{U}, 1} \mathbf{W}_{\mathbf{U}, 2}^{\top}$. Proposition~\ref{prop:bounduv} gives
		$$
		\begin{aligned}
			\left\|\mathbf{U}_{S_Y}^{\top} \mathbf{U}_{\mathbf{Y}}-\mathbf{W}_{\mathbf{U}}\right\|_F  =&\|\boldsymbol{\Sigma}_{\mathbf{U}}-\mathbf{I}\|_F=\left(\sum_{i=1}^{K_Y}\left(1-(\sigma_{\mathbf{U}})_i\right)^2\right)^{1 / 2} \leq \sum_{i=1}^{K_Y}\left(1-(\sigma_{\mathbf{U}})_i\right) \\
			\leq &\sum_{i=1}^{K_Y}\left(1-(\sigma_{\mathbf{U}})_i^2\right) \leq \frac{1}{2} \left\|\mathbf{U}_{\mathbf{Y}} \mathbf{U}_{\mathbf{Y}}^{\top}-\mathbf{U}_{S_Y} \mathbf{U}_{S_Y}^{\top}\right\|_F^2=O_{\mathbb P}\left(\frac{\log (p)}{\rho_Y p}\right) .
		\end{aligned}
		$$
		Proposition~\ref{prop:bounduv}(iv) then implies
		\[
		\left\|\mathbf{V}_{S_Y}^{\top} \mathbf{V}_{\mathbf{Y}}-\mathbf{W}_{\mathbf{U}}\right\|_F
		\leq
		\left\|\mathbf{V}_{S_Y}^{\top} \mathbf{V}_{\mathbf{Y}}-\mathbf{U}_{S_Y}^{\top} \mathbf{U}_{\mathbf{Y}}\right\|_F
		+
		\left\|\mathbf{U}_{S_Y}^{\top} \mathbf{U}_{\mathbf{Y}}-\mathbf{W}_{\mathbf{U}}\right\|_F
		=
		O_{\mathbb P}\left(\frac{ \log (p)}{\rho_Y p}\right).
		\]
		Since $\mathbf{W_{YS}}$ minimises
		\[
		\left\|\mathbf{U}_{S_Y}^{\top} \mathbf{U}_{\mathbf{Y}}-\mathbf{Q}\right\|_F^2
		+
		\left\|\mathbf{V}_{S_Y}^{\top} \mathbf{V}_{\mathbf{Y}}-\mathbf{Q}\right\|_F^2
		\]
		over $\mathbf{Q} \in \mathbb{O}(K_Y)$, its objective value is no larger than
		the value at $\mathbf W_{\mathbf U}$.  Hence
		\[
		\max \left\{
		\left\|\mathbf{U}_{S_Y}^{\top} \mathbf{U}_{\mathbf{Y}}-\mathbf{W_{YS}}\right\|_F,
		\left\|\mathbf{V}_{S_Y}^{\top} \mathbf{V}_{\mathbf{Y}}-\mathbf{W_{YS}}\right\|_F
		\right\}
		=
		O_{\mathbb P}\left(\frac{\log (p)}{\rho_Y p}\right).
		\]
	\end{proof}
	
	\begin{proposition}\label{prop:sigmaw}
		The following bounds hold:
		
		(i) $\left\|\mathbf{W_{YS}} \boldsymbol{\Sigma}_{\mathbf{Y}}-\boldsymbol{\Sigma}_{S_Y} \mathbf{W_{YS}}\right\|_F=O_{\mathbb P}\left(\log (p)\right)$.
		
		(ii) $\left\|\mathbf{W_{YS}} \boldsymbol{\Sigma}_{\mathbf{Y}}^{1 / 2}-\boldsymbol{\Sigma}_{S_Y}^{1 / 2} \mathbf{W_{YS}}\right\|_F=O_{\mathbb P}\left(\frac{\log (p)}{\rho_Y^{1 / 2} p^{1 / 2}}\right)$.
		
		(iii) $\left\|\mathbf{W_{YS}} \boldsymbol{\Sigma}_{\mathbf{Y}}^{-1 / 2}-\boldsymbol{\Sigma}_{S_Y}^{-1 / 2} \mathbf{W_{YS}}\right\|_F=O_{\mathbb P}\left(\frac{\log (p)}{\rho_Y^{3 / 2} p^{3 / 2}}\right)$.
	\end{proposition}
	\begin{proof}
		(i) Observe that
		\begin{align*}
			\mathbf{W_{YS}} \boldsymbol{\Sigma}_{\mathbf{Y}}-\boldsymbol{\Sigma}_{S_Y} \mathbf{W_{YS}}&=\left(\mathbf{W_{YS}}-\mathbf{U}_{S_Y}^{\top} \mathbf{U}_{\mathbf{Y}}\right) \boldsymbol{\Sigma}_{\mathbf{Y}}+\mathbf{U}_{S_Y}^{\top} \mathbf{U}_{\mathbf{Y}} \boldsymbol{\Sigma}_{\mathbf{Y}}-\boldsymbol{\Sigma}_{S_Y} \mathbf{W_{YS}} \\
			& =\left(\mathbf{W_{YS}}-\mathbf{U}_{S_Y}^{\top} \mathbf{U}_{\mathbf{Y}}\right) \boldsymbol{\Sigma}_{\mathbf{Y}}+\left(\mathbf{U}_{S_Y}^{\top} \mathbf{U}_{\mathbf{Y}} \boldsymbol{\Sigma}_{\mathbf{Y}}-\boldsymbol{\Sigma}_{S_Y} \mathbf{V}_{S_Y}^{\top} \mathbf{V}_{\mathbf{Y}}\right)+\boldsymbol{\Sigma}_{S_Y}\left(\mathbf{V}_{S_Y}^{\top} \mathbf{V}_{\mathbf{Y}}-\mathbf{W_{YS}}\right) .
		\end{align*}
		
		The first and third terms in the last display are
		$O_{\mathbb P}\{\log(p)\}$ by Propositions~\ref{prop:uvw} and
		\ref{prop:sigmapy}.  The middle term is
		$O_{\mathbb P}\{\log(p)\}$ by Proposition~\ref{prop:bounduv}(iii).  Hence
		\(\left\|\mathbf{W_{YS}} \boldsymbol{\Sigma}_{\mathbf{Y}}
		-\boldsymbol{\Sigma}_{S_Y} \mathbf{W_{YS}}\right\|_F
		=O_{\mathbb P}\{\log(p)\}\).
		
		(ii) Entrywise,
		\[
		\begin{aligned}
			&\left(\mathbf{W_{YS}} \boldsymbol{\Sigma}_{\mathbf{Y}}^{1 / 2}-\boldsymbol{\Sigma}_{S_Y}^{1 / 2} \mathbf{W_{YS}}\right)_{i j}  =(\mathbf{W_{YS}})_{i j}\left(\sigma_j(\mathbf{Y})^{1 / 2}-\sigma_i(\mathbf S_Y)^{1 / 2}\right) \\
			&=
			\frac{(\mathbf{W_{YS}})_{i j}\left(\sigma_j(\mathbf{Y})-\sigma_i(\mathbf S_Y)\right)}{\sigma_j(\mathbf{Y})^{1 / 2}+\sigma_i(\mathbf S_Y)^{1 / 2}}
			=
			\frac{\left(\mathbf{W_{YS}} \boldsymbol{\Sigma}_{\mathbf{Y}}-\boldsymbol{\Sigma}_{S_Y} \mathbf{W_{YS}}\right)_{i j}}{\sigma_j(\mathbf{Y})^{1 / 2}+\sigma_i(\mathbf S_Y)^{1 / 2}} .
		\end{aligned}
		\]
		
		Summing over all $i,j\in [{K_Y}]$ and using Proposition~\ref{prop:sigmapy} and
		part (i) gives the stated bound.
		
		(iii) 
		$$
		\begin{aligned}
			\left(\mathbf{W_{YS}} \boldsymbol{\Sigma}_{\mathbf{Y}}^{-1 / 2}-\boldsymbol{\Sigma}_{S_Y}^{-1 / 2} \mathbf{W_{YS}}\right)_{i j} & =\frac{(\mathbf{W_{YS}})_{i j}\left(\sigma_i(\mathbf S_Y)^{1 / 2}-\sigma_j(\mathbf{Y})^{1 / 2}\right)}{\sigma_i(\mathbf S_Y)^{1 / 2} \sigma_j(\mathbf{Y})^{1 / 2}} \\
			& =\frac{\left(\mathbf{W_{YS}} \boldsymbol{\Sigma}_{\mathbf{Y}}^{1 / 2}-\boldsymbol{\Sigma}_{S_Y}^{1 / 2} \mathbf{W_{YS}}\right)_{i j}}{\sigma_i(\mathbf S_Y)^{1 / 2} \sigma_j(\mathbf{Y})^{1 / 2}}
		\end{aligned}
		$$
		Summing over all $i,j\in [{K_Y}]$ and using Proposition~\ref{prop:sigmapy} and
		part (ii) gives the asserted inverse-square-root bound.
	\end{proof}
	\begin{proposition}
		\label{prop:uv2}
		We have 
		\[
		\left\|\mathbf V_{\mathbf S_Y}\right\|_{2\to\infty},
\left\|\mathbf U_{\mathbf S_Y}\right\|_{2\to\infty}
=
O\left(p^{-1/2}\right), 
		\left\|\mathbf V_{\mathbf Y}\right\|_{2\to\infty},
		\left\|\mathbf U_{\mathbf Y}\right\|_{2\to\infty}
		=
		O_{\mathbb P}
		\left(
		\sqrt{\frac{\log p}{\rho_Y p}}
		\right).
		\]
		
	\end{proposition}
	\begin{proof}
By definition, we have 
\(0\le(\mathbf P_Y^{(t)})_{ij}\le\rho_Y\), 
\(0\le(\mathbf Q_Y^{(t)})_{ij}\le1\) and then 
\(
  \|\mathbf S_Y\|_{\mathrm F}^{2}
  =
  \sum_{i=1}^{p}\sum_{j=1}^{np}(\mathbf S_Y)_{ij}^{2}
  \le np^{2}\rho_Y^{2}.
\)
This gives
\(\sigma_1(\mathbf S_Y)\le \|\mathbf S_Y\|_{\mathrm F}
\le \sqrt n\,p\rho_Y\).
For every \(i\in[p]\), the singular value decomposition gives
\(
  \mathbf e_i^{\top}\mathbf U_{\mathbf S_Y}
  =
  \mathbf e_i^{\top}\mathbf S_Y
  \mathbf V_{\mathbf S_Y}
  \boldsymbol\Sigma_{\mathbf S_Y}^{-1}.
\)
Therefore, by Condition~\ref{ass:3},
\[
\begin{aligned}
  \|\mathbf e_i^{\top}\mathbf U_{\mathbf S_Y}\|_2
  &\le
  \frac{\|\mathbf e_i^{\top}\mathbf S_Y\|_2}
       {\sigma_{K_Y}(\mathbf S_Y)}
   \le
  \frac{\sqrt{np}\,\rho_Y}{c_{Y,S}p\rho_Y}
   =
  \frac{\sqrt n}{c_{Y,S}\sqrt p}.
\end{aligned}
\]
Taking the maximum over \(i\) yields the asserted bound for
\(\|\mathbf U_{\mathbf S_Y}\|_{2\to\infty}\).

Similarly, for every \(j\in[np]\),
\(
  \mathbf e_j^{\top}\mathbf V_{\mathbf S_Y}
  =
  \mathbf e_j^{\top}\mathbf S_Y^{\top}
  \mathbf U_{\mathbf S_Y}
  \boldsymbol\Sigma_{\mathbf S_Y}^{-1},
\)
and hence by Condition~\ref{ass:3},
\[
\begin{aligned}
  \|\mathbf e_j^{\top}\mathbf V_{\mathbf S_Y}\|_2
  &\le
  \frac{\|\mathbf S_Y\mathbf e_j\|_2}
       {\sigma_{K_Y}(\mathbf S_Y)}
   \le
  \frac{\sqrt p\,\rho_Y}{c_{Y,S}p\rho_Y}
   =
  \frac{1}{c_{Y,S}\sqrt p}.
\end{aligned}
\]
Taking the maximum over \(j\) proves the right singular-vector bound.
		
		It remains to bound the empirical singular vectors.  Since
		$\mathbf U_{\mathbf Y}$ has orthonormal columns,
		\[
		\left\|\mathbf U_{\mathbf Y}\mathbf U_{\mathbf Y}^{\top}\mathbf e_i\right\|_2^2
		=
		\mathbf e_i^{\top}
		\mathbf U_{\mathbf Y}\mathbf U_{\mathbf Y}^{\top}
		\mathbf U_{\mathbf Y}\mathbf U_{\mathbf Y}^{\top}
		\mathbf e_i
		=
		\mathbf e_i^{\top}
		\mathbf U_{\mathbf Y}\mathbf U_{\mathbf Y}^{\top}
		\mathbf e_i
		=
		\left\|\mathbf e_i^{\top}\mathbf U_{\mathbf Y}\right\|_2^2.
		\]
		Therefore, for any $i\in[p]$,
		\[
		\begin{aligned}
			\left\|\mathbf e_i^{\top}\mathbf U_{\mathbf Y}\right\|_2
			&=
			\left\|
			\mathbf U_{\mathbf Y}\mathbf U_{\mathbf Y}^{\top}\mathbf e_i
			\right\|_2 \\
			&\le
			\left\|
			\left(
			\mathbf U_{\mathbf Y}\mathbf U_{\mathbf Y}^{\top}
			-
			\mathbf U_{\mathbf S_Y}\mathbf U_{\mathbf S_Y}^{\top}
			\right)\mathbf e_i
			\right\|_2
			+
			\left\|
			\mathbf U_{\mathbf S_Y}\mathbf U_{\mathbf S_Y}^{\top}\mathbf e_i
			\right\|_2 \\
			&\le
			\left\|
			\mathbf U_{\mathbf Y}\mathbf U_{\mathbf Y}^{\top}
			-
			\mathbf U_{\mathbf S_Y}\mathbf U_{\mathbf S_Y}^{\top}
			\right\|
			+
			\left\|
			\mathbf e_i^{\top}\mathbf U_{\mathbf S_Y}
			\right\|_2.
		\end{aligned}
		\]
		By Proposition \ref{prop:bounduv}(i),
		we have 
		$\left\|
		\mathbf U_{\mathbf Y}
		\right\|_{2\to\infty}
		=
		O_{\mathbb P}
		\left(
		\sqrt{\frac{\log p}{\rho_Y p}}
		\right).$
		The proof for $\mathbf V_{\mathbf Y}$ is identical. 
		
	\end{proof}
	
	\begin{proposition}\label{prop:loo}
		For each vertex \(r\in[p]\), define \(\mathbf Y^{(-r)}\) as follows. For each
		\(t\in[n]\), let \(\mathbf Y^{(-r),(t)}\in\mathbb R^{p\times p}\) be the matrix
		obtained from \(\mathbf Y^{(t)}\) by replacing all entries incident to vertex \(r\)
		by their population counterparts:
		\[
		(\mathbf Y^{(-r),(t)})_{ab}
		=
		\begin{cases}
			(\mathbf S_Y^{(t)})_{ab}, & a=r \text{ or } b=r,\\
			(\mathbf Y^{(t)})_{ab}, & a\neq r,\ b\neq r.
		\end{cases}
		\]
		Set
		\[
		\mathbf Y^{(-r)}
		:=
		[\mathbf Y^{(-r),(1)}|\cdots|\mathbf Y^{(-r),(n)}]
		\in\mathbb R^{p\times np}.
		\]
		Let the rank-\(K_Y\) singular value decomposition of \(\mathbf Y^{(-r)}\) be
		\[
		\mathbf Y^{(-r)}
		=
		\mathbf U_{\mathbf Y}^{(-r)}
		\boldsymbol\Sigma_{\mathbf Y}^{(-r)}
		\mathbf V_{\mathbf Y}^{(-r)\top}
		+
		\mathbf U_{\mathbf Y,\perp}^{(-r)}
		\boldsymbol\Sigma_{\mathbf Y,\perp}^{(-r)}
		\mathbf V_{\mathbf Y,\perp}^{(-r)\top}.
		\]
		Then
		\[
		\max_{r\in[p]}
		\left\|
		\mathbf U_{\mathbf Y}^{(-r)}
		\right\|_{2\to\infty},
		\quad
		\max_{r\in[p]}
		\left\|
		\mathbf V_{\mathbf Y}^{(-r)}
		\right\|_{2\to\infty}
		=
		O_{\mathbb P}
		\left(
		\sqrt{\frac{\log p}{\rho_Y p}}
		\right),
		\]
		and
		\[
		\max_{r\in[p]}
		\left\|
		\left(
		\mathbf Y-\mathbf Y^{(-r)}
		\right)
		\mathbf V_{\mathbf Y}^{(-r)}
		\right\|_F,\max_{r\in[p]}\left\|\left(\mathbf Y-\mathbf Y^{(-r)}\right)^{\top} \mathbf U_Y^{(-r)}\right\|_F
		=
		O_{\mathbb P}(\sqrt{\log p}).
		\]
		The same conclusions hold for \(\mathbf Z\) after replacing \(Y\) by \(Z\) and
		\(\rho_Y\) by \(\rho_Z\).
	\end{proposition}
	
	\begin{proof}
		Throughout the proof, the latent positions are fixed and deterministic, and
		\(\mathbf S_Y\) denotes the corresponding population mean of \(\mathbf Y\). Thus
		$(\mathbf Y^{(t)})_{ab}-(\mathbf S_Y^{(t)})_{ab}$
		is centred for $a\neq b$, the deterministic diagonal mismatch being handled
		separately, and different edge trajectories are independent across unordered
		edge pairs.
		
		We first prove the \(2\to\infty\) bounds. Fix \(r\in[p]\) and \(\ell\in[p]\). By
		the projection argument,
		\[
		\left\|\mathbf e_\ell^{\top}\mathbf U_{\mathbf Y}^{(-r)}\right\|_2
		\le
		\left\|
		\mathbf U_{\mathbf Y}^{(-r)}
		\mathbf U_{\mathbf Y}^{(-r)\top}
		-
		\mathbf U_{\mathbf S_Y}
		\mathbf U_{\mathbf S_Y}^{\top}
		\right\|
		+
		\left\|
		\mathbf U_{\mathbf S_Y}
		\right\|_{2\to\infty}.
		\]
		By Lemma~\ref{lem:davis},
		\[
		\left\|
		\mathbf U_{\mathbf Y}^{(-r)}
		\mathbf U_{\mathbf Y}^{(-r)\top}
		-
		\mathbf U_{\mathbf S_Y}
		\mathbf U_{\mathbf S_Y}^{\top}
		\right\|
		\leq
		\frac{
			2 \sqrt{2K_Y}
			\left(
			2 \sigma_1(\mathbf S_Y)+\|\mathbf{Y}^{(-r)}-\mathbf S_Y\|
			\right)
			\|\mathbf{Y}^{(-r)}-\mathbf S_Y\|
		}{
			\sigma_{K_Y}(\mathbf S_Y)^2
		},
		\]
		on the event where the denominator is positive.
		
		It remains to control \(\|\mathbf Y^{(-r)}-\mathbf S_Y\|\) uniformly over
		\(r\).  We do this directly rather than comparing \(\mathbf Y^{(-r)}\) with
		\(\mathbf Y\) in Frobenius norm.  The off-diagonal part of \(\mathbf Y^{(-r)}-\mathbf S_Y\) is a sum of independent
		centred dyad-block matrices over unordered dyads not incident to \(r\); dyads
		incident to \(r\) contribute zero after the leave-one-row replacement.  As in
		Proposition~\ref{prop:boundy-p}, the deterministic diagonal contribution has
		operator norm \(O(\rho_Y)\), which is dominated by the stochastic bound below.
		For each fixed \(r\), the same dyad-block calculation as in
		Lemma~\ref{lem:fixed-n-dyad-block-concentration} gives variance bounded
		by \(Cnp\rho_Y\) and summand norm bounded by \(C\sqrt n\).  Therefore
		Lemma~\ref{lem:matrixBernstein} gives, for every fixed \(c>0\),
		\[
		\mathbb P\left\{
		\left\|
		\mathbf Y^{(-r)}-\mathbf S_Y
		\right\|
		>
		C_c\left((p\rho_Y\log p)^{1/2}+\log p\right) \right\}
		\le
		p^{-(c+2)},
		\]
		uniformly in \(r\).  A union
		bound over \(r\in[p]\) yields
		\[
		\max_{r\in[p]}
		\left\|
		\mathbf Y^{(-r)}-\mathbf S_Y
		\right\|
		=
		O_{\mathbb P}
		\left(
		(p\rho_Y\log p)^{1/2}+\log p
		\right).
		\]
		Since \(n\) is fixed and Condition~\ref{ass:transition-sparsity} implies
		\(p\rho_Y/\log p\to\infty\), the logarithmic term is dominated.  Hence
		\begin{equation}\label{eq:uniform-loo-operator-bound}
		    	\max_{r\in[p]}
		\left\|
		\mathbf Y^{(-r)}-\mathbf S_Y
		\right\|
		=
		O_{\mathbb P}
		\left(
		(\rho_Yp\log p)^{1/2}
		\right).
		\end{equation}
		Condition~\ref{ass:3} then yields
		\begin{equation}\label{eq:loo}
			\max_{r\in[p]}
			\left\|
			\mathbf U_{\mathbf Y}^{(-r)}
			\mathbf U_{\mathbf Y}^{(-r)\top}
			-
			\mathbf U_{\mathbf S_Y}
			\mathbf U_{\mathbf S_Y}^{\top}
			\right\|
			=
			O_{\mathbb P}
			\left(
			\sqrt{\frac{\log p}{\rho_Y p}}
			\right).
		\end{equation}
		
		Together with Proposition~\ref{prop:uv2}, this gives
		\[
		\max_{r\in[p]}
		\left\|
		\mathbf U_{\mathbf Y}^{(-r)}
		\right\|_{2\to\infty}
		=
		O_{\mathbb P}
		\left(
		\sqrt{\frac{\log p}{\rho_Y p}}
		\right).
		\]
		
		The proof for \(\mathbf V_{\mathbf Y}^{(-r)}\) is identical. 
		
		We next prove the Frobenius norm bound. Fix \(r\in[p]\). By the definition of
		\(\mathbf Y^{(-r)}\), for each \(t\in[n]\),
		\[
		\left(\mathbf Y^{(t)}-\mathbf Y^{(-r),(t)}\right)_{ab}
		=
		\begin{cases}
			(\mathbf Y^{(t)})_{ab}-(\mathbf S_Y^{(t)})_{ab},
			& a=r \text{ or } b=r,\\
			0,
			& a\neq r,\ b\neq r.
		\end{cases}
		\]
		Thus \((\mathbf Y-\mathbf Y^{(-r)})\mathbf V_{\mathbf Y}^{(-r)}\) only has
		nonzero contributions from entries incident to vertex \(r\).
		
		First consider the \(r\)-th row. We have
		\[
		\begin{aligned}
			\mathbf e_r^\top
			\left(
			\mathbf Y-\mathbf Y^{(-r)}
			\right)
			\mathbf V_{\mathbf Y}^{(-r)}
			&=
			\sum_{a\neq r}
			\sum_{t=1}^n
			\left\{
			(\mathbf Y^{(t)})_{ra}
			-
			(\mathbf S_Y^{(t)})_{ra}
			\right\}
			\mathbf e_{(t,a)}^\top
			\mathbf V_{\mathbf Y}^{(-r)}
			\\
			&\quad
			-
			\sum_{t=1}^n
			(\mathbf S_Y^{(t)})_{rr}
			\mathbf e_{(t,r)}^\top
			\mathbf V_{\mathbf Y}^{(-r)}.
		\end{aligned}
		\]
		The second term is the diagonal correction. Since
		$\max_{t,r}|(\mathbf S_Y^{(t)})_{rr}|
		\le C\rho_Y$,
		we have
		\[
		\max_{r\in[p]}
		\left\|
		\sum_{t=1}^n
		(\mathbf S_Y^{(t)})_{rr}
		\mathbf e_{(t,r)}^\top
		\mathbf V_{\mathbf Y}^{(-r)}
		\right\|_2
		=
		O_{\mathbb P}
		\left(
		\sqrt{\frac{\rho_Y\log p}{p}}
		\right)
		=
		O_{\mathbb P}(\sqrt{\log p}).
		\]
		
		Fix \(q\in[K_Y]\).  The \(q\)-th coordinate of the first term in the
		\(r\)-th row is
		\[
		\sum_{a\neq r}
		\sum_{t=1}^n
		\left\{
		(\mathbf Y^{(t)})_{ra}
		-
		(\mathbf S_Y^{(t)})_{ra}
		\right\}
		\left(
		\mathbf e_{(t,a)}^\top
		\mathbf V_{\mathbf Y}^{(-r)}
		\right)_q .
		\]
		Conditional on \(\mathbf Y^{(-r)}\), the
		weights
		$\left(
		\mathbf e_{(t,a)}^\top
		\mathbf V_{\mathbf Y}^{(-r)}
		\right)_q$
		are fixed, while the edge trajectories
		\[
		\left\{
		\left(
		(\mathbf Y^{(1)})_{ra},
		\ldots,
		(\mathbf Y^{(n)})_{ra}
		\right):a\neq r
		\right\}
		\]
		are independent over \(a\neq r\). We do not require independence over time within
		a fixed edge trajectory.
		
		For each \(a\neq r\), the summand
		$\sum_{t=1}^n
		\left\{
		(\mathbf Y^{(t)})_{ra}
		-
		(\mathbf S_Y^{(t)})_{ra}
		\right\}
		\left(
		\mathbf e_{(t,a)}^\top
		\mathbf V_{\mathbf Y}^{(-r)}
		\right)_q$
		is centred and bounded in absolute value by
		$\sum_{t=1}^n
		\left|
		\left(
		\mathbf e_{(t,a)}^\top
		\mathbf V_{\mathbf Y}^{(-r)}
		\right)_q
		\right|.$
		Moreover, by Cauchy--Schwarz and the fact that the \(q\)-th column of
		\(\mathbf V_{\mathbf Y}^{(-r)}\) has unit norm,
		\[
		\sum_{a\neq r}
		\left(
		\sum_{t=1}^n
		\left|
		\left(
		\mathbf e_{(t,a)}^\top
		\mathbf V_{\mathbf Y}^{(-r)}
		\right)_q
		\right|
		\right)^2
		\le
		n
		\sum_{a\neq r}
		\sum_{t=1}^n
		\left(
		\left(
		\mathbf e_{(t,a)}^\top
		\mathbf V_{\mathbf Y}^{(-r)}
		\right)_q
		\right)^2
		\le n.
		\]
		Therefore, by conditional Hoeffding's inequality, there exists a constant
		\(c>0\) such that
		\[
		\mathbb P
		\left(
		\left|
		\sum_{a\neq r}
		\sum_{t=1}^n
		\left\{
		(\mathbf Y^{(t)})_{ra}
		-
		(\mathbf S_Y^{(t)})_{ra}
		\right\}
		\left(
		\mathbf e_{(t,a)}^\top
		\mathbf V_{\mathbf Y}^{(-r)}
		\right)_q
		\right|
		>x
		\ \middle|\
		\mathbf Y^{(-r)}
		\right)
		\le
		2\exp(-cx^2).
		\]
		Taking \(x=C\sqrt{\log p}\) and applying a union bound over
		\(r\in[p]\) and \(q\in[K_Y]\), with \(K_Y\) fixed, gives
		\[
		\max_{r\in[p]}
		\left\|
		\sum_{a\neq r}
		\sum_{t=1}^n
		\left\{
		(\mathbf Y^{(t)})_{ra}
		-
		(\mathbf S_Y^{(t)})_{ra}
		\right\}
		\mathbf e_{(t,a)}^\top
		\mathbf V_{\mathbf Y}^{(-r)}
		\right\|_2
		=
		O_{\mathbb P}(\sqrt{\log p}).
		\]
		Combining this with the diagonal correction,
		\[
		\max_{r\in[p]}
		\left\|
		\mathbf e_r^\top
		\left(
		\mathbf Y-\mathbf Y^{(-r)}
		\right)
		\mathbf V_{\mathbf Y}^{(-r)}
		\right\|_2
		=
		O_{\mathbb P}(\sqrt{\log p}).
		\]
		
		Next consider rows \(a\neq r\). For each \(a\neq r\),
		\[
		\mathbf e_a^\top
		\left(
		\mathbf Y-\mathbf Y^{(-r)}
		\right)
		\mathbf V_{\mathbf Y}^{(-r)}
		=
		\sum_{t=1}^n
		\left\{
		(\mathbf Y^{(t)})_{ar}
		-
		(\mathbf S_Y^{(t)})_{ar}
		\right\}
		\mathbf e_{(t,r)}^\top
		\mathbf V_{\mathbf Y}^{(-r)}.
		\]
		Define
		\[
		X_{a,r}
		:=
		\left\|
		\sum_{t=1}^n
		\left\{
		(\mathbf Y^{(t)})_{ar}
		-
		(\mathbf S_Y^{(t)})_{ar}
		\right\}
		\mathbf e_{(t,r)}^\top
		\mathbf V_{\mathbf Y}^{(-r)}
		\right\|_2^2,
		\qquad a\neq r.
		\]
		Then
		\[
		\sum_{a\neq r}
		\left\|
		\mathbf e_a^\top
		\left(
		\mathbf Y-\mathbf Y^{(-r)}
		\right)
		\mathbf V_{\mathbf Y}^{(-r)}
		\right\|_2^2
		=
		\sum_{a\neq r}X_{a,r}.
		\]
		Conditional on \(\mathbf Y^{(-r)}\), the edge
		trajectories
		\[
		\left\{
		\bigl(
		(\mathbf Y^{(1)})_{ar},\ldots,(\mathbf Y^{(n)})_{ar}
		\bigr):a\neq r
		\right\}
		\]
		are independent over \(a\neq r\). Hence the random variables
		\(\{X_{a,r}:a\neq r\}\) are conditionally independent given
		\(\mathbf Y^{(-r)}\).
		
		Since \(n\) is fixed and all entries are bounded,
		\[
		X_{a,r}
		\le
		C
		\sum_{t=1}^n
		\left\|
		\mathbf e_{(t,r)}^\top
		\mathbf V_{\mathbf Y}^{(-r)}
		\right\|_2^2
		=
		O_{\mathbb P}
		\left(
		\frac{\log p}{\rho_Yp}
		\right).
		\]
		Therefore, on a high-probability event, there exists a constant \(C_1>0\) such
		that uniformly over \(a\neq r\) and \(r\in[p]\),
		\[
		0\le X_{a,r}
		\le
		C_1\frac{\log p}{\rho_Yp}.
		\]
		
		Next, since \(n\) is fixed and the temporal covariance matrix of
		\[
		\left(
		(\mathbf Y^{(1)})_{ar}-(\mathbf S_Y^{(1)})_{ar},
		\ldots,
		(\mathbf Y^{(n)})_{ar}-(\mathbf S_Y^{(n)})_{ar}
		\right)
		\]
		has operator norm \(O(\rho_Y)\), we have
		\[
		\mathbb E[X_{a,r}\mid \mathbf Y^{(-r)}]
		\le
		C\rho_Y
		\sum_{t=1}^n
		\left\|
		\mathbf e_{(t,r)}^\top
		\mathbf V_{\mathbf Y}^{(-r)}
		\right\|_2^2.
		\]
		Thus
		\[
		\sum_{a\neq r}
		\mathbb E[X_{a,r}\mid \mathbf Y^{(-r)}]
		\le
		Cp\rho_Y
		\sum_{t=1}^n
		\left\|
		\mathbf e_{(t,r)}^\top
		\mathbf V_{\mathbf Y}^{(-r)}
		\right\|_2^2
		=
		O_{\mathbb P}(\log p),
		\]
		uniformly over \(r\in[p]\). Hence, on the same high-probability event, there
		exists a constant \(C_2>0\) such that for all \(r\in[p]\),
		$\sum_{a\neq r}
		\mathbb E[X_{a,r}\mid \mathbf Y^{(-r)}]
		\le
		C_2\log p.$
		
		Apply Bernstein's inequality conditionally on \(\mathbf Y^{(-r)}\).  The variance
		term satisfies
		\[
		\sum_{a\neq r}
		\operatorname{Var}(X_{a,r}\mid \mathbf Y^{(-r)})
		\le
		\sum_{a\neq r}
		\mathbb E[X_{a,r}^2\mid \mathbf Y^{(-r)}] 
		\le
		C_1\frac{\log p}{\rho_Yp}
		\sum_{a\neq r}
		\mathbb E[X_{a,r}\mid \mathbf Y^{(-r)}] 
		\le
		C_1C_2\frac{\log^2 p}{\rho_Yp}.
		\]
		Therefore, for any sufficiently large constant \(C_3>C_2\),
		\[
		\begin{aligned}
			\mathbb P\left(
			\sum_{a\neq r}X_{a,r}
			>
			C_3\log p
			\ \middle|\
			\mathbf Y^{(-r)}
			\right) &\le
			\exp\left[
			-
			\frac{
				(C_3-C_2)^2\log^2 p/2
			}{
				C_1C_2\frac{\log^2 p}{\rho_Yp}
				+
				C_1\frac{\log p}{\rho_Yp}(C_3-C_2)\log p/3
			}
			\right] \\
			&\le
			\exp(-c\rho_Yp)
		\end{aligned}
		\]
		for some constant \(c>0\).  A union bound over \(r\in[p]\), together with
		Condition~\ref{ass:transition-sparsity}, which implies
		\(\rho_Yp\gg \log p\), gives
		\[
		\max_{r\in[p]}
		\sum_{a\neq r}
		\left\|
		\mathbf e_a^\top
		\left(
		\mathbf Y-\mathbf Y^{(-r)}
		\right)
		\mathbf V_{\mathbf Y}^{(-r)}
		\right\|_2^2
		=
		O_{\mathbb P}(\log p).
		\]
		Therefore,
		\[
		\max_{r\in[p]}
		\left(
		\sum_{a\neq r}
		\left\|
		\mathbf e_a^\top
		\left(
		\mathbf Y-\mathbf Y^{(-r)}
		\right)
		\mathbf V_{\mathbf Y}^{(-r)}
		\right\|_2^2
		\right)^{1/2}
		=
		O_{\mathbb P}(\sqrt{\log p}).
		\]
		
		Combining the bounds for the \(r\)-th row and the rows \(a\neq r\), we conclude
		that
		\[
		\max_{r\in[p]}
		\left\|
		\left(
		\mathbf Y-\mathbf Y^{(-r)}
		\right)
		\mathbf V_{\mathbf Y}^{(-r)}
		\right\|_F
		=
		O_{\mathbb P}(\sqrt{\log p}).
		\]
		The proof of
		\[
		\max _r\left\|\left(\mathbf Y-\mathbf Y^{(-r)}\right)^{\top}
		\mathbf U_Y^{(-r)}\right\|_F
		=
		O_{\mathbb{P}}(\sqrt{\log p})
		\]
		is identical.
	\end{proof}

	\begin{proposition}\label{prop:last4}
		Define 
		\begin{align*}
			R_1 =& \mathbf{V}_{S_Y}\left(\mathbf{V}_{S_Y}^{\top} \mathbf{V}_{\mathbf{Y}} \boldsymbol{\Sigma}_{\mathbf{Y}}^{1 / 2}-\boldsymbol{\Sigma}_{S_Y}^{1 / 2} \mathbf{W_{YS}}\right);\\
			R_2 =& \left(\mathbf{I}-\mathbf{V}_{S_Y} \mathbf{V}_{S_Y}^{\top}\right)(\mathbf{Y}-\mathbf S_Y)^{\top}\left(\mathbf{U}_{\mathbf{Y}}-\mathbf{U}_{S_Y} \mathbf{W_{YS}}\right) \boldsymbol{\Sigma}_{\mathbf{Y}}^{-1 / 2};\\
			R_3 =& -\mathbf{V}_{S_Y} \mathbf{V}_{S_Y}^{\top}(\mathbf{Y}-\mathbf S_Y)^{\top} \mathbf{U}_{S_Y} \mathbf{W_{YS}} \boldsymbol{\Sigma}_{\mathbf{Y}}^{-1 / 2};\\
			R_4 =& (\mathbf{Y}-\mathbf S_Y)^{\top} \mathbf{U}_{S_Y}\left(\mathbf{W_{YS}} \boldsymbol{\Sigma}_{\mathbf{Y}}^{-1 / 2}-\boldsymbol{\Sigma}_{S_Y}^{-1 / 2} \mathbf{W_{YS}}\right).
		\end{align*}
		The following bounds hold:
		
		(i) $\left\|R_1\right\|_{2 \rightarrow \infty}=O_{\mathbb P}\left(\frac{\log (p)}{\rho_Y^{1 / 2} p}\right)$.
		
		(ii) $\left\|R_2\right\|_{2 \rightarrow \infty}=O_{\mathbb P}\left(\frac{\log ^{3 / 2} p}{\rho_Y p}\right)$.
		
		(iii) $\left\|R_3\right\|_{2 \rightarrow \infty}=O_{\mathbb P}\left(
\frac{\sqrt{\log p}}{p\sqrt{\rho_Y}}
\right)$.
		
		(iv) $\left\|R_4\right\|_{2 \rightarrow \infty}=O_{\mathbb P}\left(\frac{\log ^{3 / 2}(p)}{\rho_Y p}\right)$.
		
	\end{proposition}
	
	\begin{proof}
		(i) Using the relation $\|\mathbf{A B}\|_{2 \rightarrow \infty} \leq\|\mathbf{A}\|_{2 \rightarrow \infty}\|\mathbf{B}\|$, we have 
		\begin{align*}
			&\left\|R_1\right\|_{2 \rightarrow \infty} \leq \left\|\mathbf{V}_{S_Y}\right\|_{2 \rightarrow \infty}\left\|\left(\mathbf{V}_{S_Y}^{\top} \mathbf{V}_{\mathbf{Y}} \boldsymbol{\Sigma}_{\mathbf{Y}}^{1 / 2}-\boldsymbol{\Sigma}_{S_Y}^{1 / 2} \mathbf{W_{YS}}\right)\right\| \\
			\leq &\left\|\mathbf{V}_{S_Y}\right\|_{2 \rightarrow \infty}\left(\left\|\left(\mathbf{V}_{S_Y}^{\top} \mathbf{V}_{\mathbf{Y}}-\mathbf{W_{YS}}\right) \boldsymbol{\Sigma}_{\mathbf{Y}}^{1 / 2}\right\|_F+\left\|\mathbf{W_{YS}} \boldsymbol{\Sigma}_{\mathbf{Y}}^{1 / 2}-\boldsymbol{\Sigma}_{S_Y}^{1 / 2} \mathbf{W_{YS}}\right\|_F\right)\\
			=& O_{\mathbb P}\left(\frac{\log (p)}{\rho_Y^{1 / 2} p}\right)
		\end{align*}
		by Proposition~\ref{prop:sigmapy},~\ref{prop:uvw},~\ref{prop:sigmaw},~\ref{prop:uv2}.
		
		(ii) Decompose $R_2$ into the three terms
		\[
		\begin{aligned}
			-\mathbf M_1
			&=
			-\mathbf{V}_{S_Y} \mathbf{V}_{S_Y}^{\top}
			(\mathbf{Y}-\mathbf S_Y)^{\top}
			\left(\mathbf{U}_{\mathbf{Y}}-\mathbf{U}_{S_Y} \mathbf{W_{YS}}\right)
			\boldsymbol{\Sigma}_{\mathbf{Y}}^{-1 / 2},\\
			\mathbf M_2
			&=
			(\mathbf{Y}-\mathbf S_Y)^{\top}
			\left(\mathbf{I} -\mathbf{U}_{S_Y}\mathbf{U}_{S_Y}^{\top}\right)
			\mathbf{U}_{\mathbf{Y}}
			\boldsymbol{\Sigma}_{\mathbf{Y}}^{-1 / 2},\\
			\mathbf M_3
			&=
			(\mathbf{Y}-\mathbf S_Y)^{\top}\mathbf{U}_{S_Y}
			\left(\mathbf{U}_{S_Y}^{\top}\mathbf{U}_{\mathbf{Y}}-\mathbf{W_{YS}}\right)
			\boldsymbol{\Sigma}_{\mathbf{Y}}^{-1 / 2}.
		\end{aligned}
		\]
		For the first term, Propositions~\ref{prop:boundy-p},~\ref{prop:sigmapy},
		\ref{prop:uvw}, and~\ref{prop:uv2} imply
		$$\|\mathbf{M}_1\|_{2 \rightarrow \infty} \leq \left\|\mathbf{V}_{S_Y}\right\|_{2 \rightarrow \infty}\left\|\mathbf{Y}-\mathbf S_Y\right\|\left\|\mathbf{U}_{\mathbf{Y}}-\mathbf{U}_{S_Y} \mathbf{W_{YS}}\right\|\left\|\boldsymbol{\Sigma}_{\mathbf{Y}}^{-1 / 2}\right\|\leq \left\|\mathbf{U}_{\mathbf{Y}}-\mathbf{U}_{S_Y} \mathbf{W_{YS}}\right\| O_{\mathbb P}(\sqrt{\frac{\log p}{p}}).$$
		Furthermore, Propositions~\ref{prop:bounduv} and~\ref{prop:uvw} give
		\[
		\left\|\mathbf{U}_{\mathbf{Y}}-\mathbf{U}_{S_Y} \mathbf{W_{YS}}\right\|
		\leq
		\left\|\mathbf{U}_{\mathbf{Y}}-\mathbf{U}_{S_Y}
		\mathbf{U}_{S_Y}^{\top} \mathbf{U}_{\mathbf{Y}}\right\|
		+
		\left\|\mathbf{U}_{S_Y}
		\left(\mathbf{U}_{S_Y}^{\top} \mathbf{U}_{\mathbf{Y}}-\mathbf{W_{YS}}\right)\right\|
		=
		O_{\mathbb P}\left(\sqrt{\frac{\log p}{\rho_Y p}}\right).
		\]
		Hence
		\(\|\mathbf{M}_1\|_{2\rightarrow \infty}
		=O_{\mathbb P}\{\log p/(p\rho_Y^{1/2})\}\).
		
		For \(\mathbf M_2\),
		$$\left\|\mathbf{M}_2\right\|_{2\rightarrow\infty} \leq \left\|(\mathbf{Y}-\mathbf S_Y)^{\top}\left(\mathbf{I} -\mathbf{U}_{S_Y}\mathbf{U}_{S_Y}^{\top}\right)\mathbf{U}_{\mathbf{Y}}\mathbf{U}_{\mathbf{Y}}^{\top}\right\|_{2\rightarrow\infty} \|\mathbf{U}_{\mathbf{Y}} \boldsymbol{\Sigma}_{\mathbf{Y}}^{-1 / 2}\|.$$
		Set $\mathbf{M}_4 = (\mathbf{Y}-\mathbf S_Y)^{\top}\left(\mathbf{I} -\mathbf{U}_{S_Y}\mathbf{U}_{S_Y}^{\top}\right)\mathbf{U}_{\mathbf{Y}}\mathbf{U}_{\mathbf{Y}}^{\top}. $ 
		For \(t\in[n]\) and \(r\in[p]\), let \(\mathbf e_{(t,r)}\in\mathbb R^{np}\) denote the
		standard basis vector corresponding to the \(r\)-th column in the \(t\)-th block
		of \(\mathbf Y\). Since \(\mathbf Y^{(-r)}\) replaces all entries incident to
		vertex \(r\) by their population counterparts, we have
		$\mathbf e_{(t,r)}^\top(\mathbf Y-\mathbf S_Y)^\top
		=
		\mathbf e_{(t,r)}^\top(\mathbf Y-\mathbf Y^{(-r)})^\top.$
		Therefore,
		\[
		\begin{aligned}
			\mathbf e_{(t,r)}^\top
			\mathbf M_4
			&=
			\mathbf e_{(t,r)}^\top
			(\mathbf Y-\mathbf Y^{(-r)})^\top
			(\mathbf I-\mathbf U_{\mathbf S_Y}\mathbf U_{\mathbf S_Y}^{\top})
			\mathbf U_{\mathbf Y}^{(-r)}
			\mathbf U_{\mathbf Y}^{(-r)\top} \\
			&\qquad +
			\mathbf e_{(t,r)}^\top
			(\mathbf Y-\mathbf Y^{(-r)})^\top
			(\mathbf I-\mathbf U_{\mathbf S_Y}\mathbf U_{\mathbf S_Y}^{\top})
			\left(
			\mathbf U_{\mathbf Y}\mathbf U_{\mathbf Y}^{\top}
			-
			\mathbf U_{\mathbf Y}^{(-r)}
			\mathbf U_{\mathbf Y}^{(-r)\top}
			\right) \\
			&=:
			\mathbf T_{t,r,1}
			+
			\mathbf T_{t,r,2}.
		\end{aligned}
		\]
		
		We first control \(\mathbf T_{t,r,1}\). Since
		$(\mathbf I-\mathbf U_{\mathbf S_Y}\mathbf U_{\mathbf S_Y}^{\top})
		\mathbf U_{\mathbf S_Y}\mathbf U_{\mathbf S_Y}^{\top}
		=
		\mathbf 0,$
		we have
		\[
		\begin{aligned}
			(\mathbf I-\mathbf U_{\mathbf S_Y}\mathbf U_{\mathbf S_Y}^{\top})
			\mathbf U_{\mathbf Y}^{(-r)}
			& =
			(\mathbf I-\mathbf U_{\mathbf S_Y}\mathbf U_{\mathbf S_Y}^{\top})
			\mathbf U_{\mathbf Y}^{(-r)}
			\mathbf U_{\mathbf Y}^{(-r)\top}
			\mathbf U_{\mathbf Y}^{(-r)} \\
			& =
			(\mathbf I-\mathbf U_{\mathbf S_Y}\mathbf U_{\mathbf S_Y}^{\top})
			\left(
			\mathbf U_{\mathbf Y}^{(-r)}
			\mathbf U_{\mathbf Y}^{(-r)\top}
			-
			\mathbf U_{\mathbf S_Y}\mathbf U_{\mathbf S_Y}^{\top}
			\right)
			\mathbf U_{\mathbf Y}^{(-r)}.
		\end{aligned}
		\]
		Since $\left\|\mathbf{I}-\mathbf{U}_{\mathbf{S}_Y} \mathbf{U}_{\mathbf{S}_Y}^{\top}\right\| \leq 1$,
		\[
		\begin{aligned}
			\left\|
			(\mathbf I-\mathbf U_{\mathbf S_Y}\mathbf U_{\mathbf S_Y}^{\top})
			\mathbf U_{\mathbf Y}^{(-r)}
			\right\|
			&\le
			\left\|
			\mathbf I-\mathbf U_{\mathbf S_Y}\mathbf U_{\mathbf S_Y}^{\top}
			\right\|
			\left\|
			\mathbf U_{\mathbf Y}^{(-r)}
			\mathbf U_{\mathbf Y}^{(-r)\top}
			-
			\mathbf U_{\mathbf S_Y}\mathbf U_{\mathbf S_Y}^{\top}
			\right\|
			\left\|
			\mathbf U_{\mathbf Y}^{(-r)}
			\right\| \\
			&=
			\left\|
			\mathbf U_{\mathbf Y}^{(-r)}
			\mathbf U_{\mathbf Y}^{(-r)\top}
			-
			\mathbf U_{\mathbf S_Y}\mathbf U_{\mathbf S_Y}^{\top}
			\right\|.
		\end{aligned}
		\]
		By Equation (\ref{eq:loo}), we have
		\[
		\max_{r\in[p]}
		\left\|
		(\mathbf I-\mathbf U_{\mathbf S_Y}\mathbf U_{\mathbf S_Y}^{\top})
		\mathbf U_{\mathbf Y}^{(-r)}
		\right\|
		=
		O_{\mathbb P}
		\left(
		\sqrt{\frac{\log p}{\rho_Yp}}
		\right).
		\]
		
		Since \(\|\mathbf U_{\mathbf Y}^{(-r)\top}\|=1\),
		\[
		\|\mathbf T_{t,r,1}\|_2
		\le
		\left\|
		\mathbf e_{(t,r)}^\top
		(\mathbf Y-\mathbf Y^{(-r)})^\top
		(\mathbf I-\mathbf U_{\mathbf S_Y}\mathbf U_{\mathbf S_Y}^{\top})
		\mathbf U_{\mathbf Y}^{(-r)}
		\right\|_2.
		\]
		Conditional on \(\mathbf Y^{(-r)}\), the weights
		in
		$(\mathbf I-\mathbf U_{\mathbf S_Y}\mathbf U_{\mathbf S_Y}^{\top})
		\mathbf U_{\mathbf Y}^{(-r)}$
		are fixed, while the edge variables
		$\{(\mathbf Y^{(t)})_{ar}:a\neq r\}$
		are independent over \(a\neq r\). Therefore, the same conditional Hoeffding
		argument as in Proposition~\ref{prop:loo} yields
		\[
		\max_{t\in[n]}\max_{r\in[p]}
		\|\mathbf T_{t,r,1}\|_2
		=
		O_{\mathbb P}
		\left(
		\frac{\log p}{\sqrt{\rho_Yp}}
		\right).
		\]

		Next we control \(\mathbf T_{t,r,2}\).  Applying Wedin's theorem (for example,
Theorem~2.9 of \citet{chen2021spectral}) to
$\mathbf Y
=
\mathbf Y^{(-r)}
+
(\mathbf Y-\mathbf Y^{(-r)})$,
we have
		\[
		\left\|
		\mathbf U_{\mathbf Y}\mathbf U_{\mathbf Y}^{\top}
		-
		\mathbf U_{\mathbf Y}^{(-r)}
		\mathbf U_{\mathbf Y}^{(-r)\top}
		\right\|_F  \le
		\frac{
			C
			\max \left\{
			\left\|
			(\mathbf Y-\mathbf Y^{(-r)})
			\mathbf V_{\mathbf Y}^{(-r)}
			\right\|_F
			,
			\left\|
			(\mathbf Y-\mathbf Y^{(-r)})^\top
			\mathbf U_{\mathbf Y}^{(-r)}
			\right\|_F
			\right\}
		}{
			\sigma_{K_Y}(\mathbf Y^{(-r)})
		}.
		\]
		By Weyl's inequality, Condition~\ref{ass:3}, and
		\eqref{eq:uniform-loo-operator-bound},
		$\min_{r\in[p]}
		\sigma_{K_Y}(\mathbf Y^{(-r)})
		=
		\Omega_{\mathbb P}(p\rho_Y).$  Proposition~\ref{prop:loo}, together with the
		same argument for left singular vectors, gives
		\[
		\max_{r\in[p]}
		\left\|
		\mathbf U_{\mathbf Y}\mathbf U_{\mathbf Y}^{\top}
		-
		\mathbf U_{\mathbf Y}^{(-r)}
		\mathbf U_{\mathbf Y}^{(-r)\top}
		\right\|_F
		=
		O_{\mathbb P}
		\left(
		\frac{\sqrt{\log p}}{p\rho_Y}
		\right).
		\]
		
		On the other hand, the same Bernstein argument for incident edge variables as in
		Proposition~\ref{prop:loo} gives
		\[
		\max_{t\in[n]}\max_{r\in[p]}
		\left\|
		\mathbf e_{(t,r)}^\top
		(\mathbf Y-\mathbf Y^{(-r)})^\top
		\right\|_2
		=
		O_{\mathbb P}
		\left(
		(\rho_Yp\log p)^{1/2}
\right),
\]
and hence
		\[
		\max_{t\in[n]}\max_{r\in[p]}
		\left\|
		\mathbf e_{(t,r)}^\top
		(\mathbf Y-\mathbf Y^{(-r)})^\top
		(\mathbf I-\mathbf U_{\mathbf S_Y}\mathbf U_{\mathbf S_Y}^{\top})
		\right\|_2
		=
		O_{\mathbb P}
		\left(
		(\rho_Yp\log p)^{1/2}
		\right).
		\]
		Thus
		\[
		\begin{aligned}
			\max_{t\in[n]}\max_{r\in[p]}
			\|\mathbf T_{t,r,2}\|_2
			&\le
			\max_{t,r}
			\left\|
			\mathbf e_{(t,r)}^\top
			(\mathbf Y-\mathbf Y^{(-r)})^\top
			(\mathbf I-\mathbf U_{\mathbf S_Y}\mathbf U_{\mathbf S_Y}^{\top})
			\right\|_2 
			\max_{r\in[p]}
			\left\|
			\mathbf U_{\mathbf Y}\mathbf U_{\mathbf Y}^{\top}
			-
			\mathbf U_{\mathbf Y}^{(-r)}
			\mathbf U_{\mathbf Y}^{(-r)\top}
			\right\|_F \\
			&=
			O_{\mathbb P}
			\left(
			(\rho_Yp\log p)^{1/2}
			\right)
			O_{\mathbb P}
			\left(
			\frac{\sqrt{\log p}}{p\rho_Y}
			\right) \\
			&=
			O_{\mathbb P}
			\left(
			\frac{\log p}{\sqrt{\rho_Yp}}
			\right).
		\end{aligned}
		\]
		
		Combining the bounds for \(\mathbf T_{t,r,1}\) and \(\mathbf T_{t,r,2}\), we obtain
		\[
		\left\|
		\mathbf M_4
		\right\|_{2\to\infty}
		=
		O_{\mathbb P}
		\left(
		\frac{\log p}{\sqrt{\rho_Yp}}
		\right).
		\]
		It follows that
		\[
		\left\|\mathbf{M}_2\right\|_{2\rightarrow\infty}
		=
		O_{\mathbb P}\left(\frac{\log p}{\rho_Y p}\right).
		\]

		For the third term, Propositions~\ref{prop:boundy-p}, \ref{prop:sigmapy},
		and~\ref{prop:uvw} imply
		$$\left\|\mathbf{M}_3\right\|_{2\rightarrow\infty} \leq\left\|\mathbf{M}_3\right\|\leq \left\|\mathbf{Y}-\mathbf S_Y\right\|\left\|\mathbf{U}_{S_Y}\right\|\left\|\mathbf{U}_{S_Y}^{\top}\mathbf{U}_{\mathbf{Y}}-\mathbf{W_{YS}} \right\|\left\|\boldsymbol{\Sigma}_{\mathbf{Y}}^{-1 / 2}\right\|=O_{\mathbb P}(\frac{(\log p)^{3/2}}{p\rho_Y}).$$
		
		Combining the three bounds gives the asserted rate.
		
		(iii) By Propositions~\ref{prop:sigmapy},~\ref{prop:uy-pv}, and~\ref{prop:uv2},
		\begin{align*}
			\left\|R_3\right\|_{2 \rightarrow \infty} \leq & \left\|\mathbf{V}_{S_Y}\right\|_{2 \rightarrow \infty}\left\| \mathbf{V}_{S_Y}^{\top}\left(\mathbf{Y}-\mathbf S_Y\right)^{\top} \mathbf{U}_{S_Y} \mathbf{W_{YS}} \boldsymbol{\Sigma}_{\mathbf{Y}}^{-1 / 2}\right\| \\
			\leq & \left\|\mathbf{V}_{S_Y}\right\|_{2 \rightarrow \infty}\left\| \mathbf{V}_{S_Y}^{\top}\left(\mathbf{Y}-\mathbf S_Y\right)^{\top} \mathbf{U}_{S_Y}\right\|_{F}\left\| \mathbf{W_{YS}} \boldsymbol{\Sigma}_{\mathbf{Y}}^{-1 / 2}\right\|_{F} \\
			= &O_{\mathbb P}\left(
\frac{\sqrt{\log p}}{p\sqrt{\rho_Y}}
\right).
		\end{align*}
		
		(iv) By Propositions~\ref{prop:boundy-p} and~\ref{prop:sigmaw},
		\begin{align*}
			\left\|R_4\right\|_{2 \rightarrow \infty}\leq& \left\|(\mathbf{Y}-\mathbf S_Y)^{\top} \mathbf{U}_{S_Y}\left(\mathbf{W_{YS}} \boldsymbol{\Sigma}_{\mathbf{Y}}^{-1 / 2}-\boldsymbol{\Sigma}_{S_Y}^{-1 / 2} \mathbf{W_{YS}}\right)\right\|_{F}\\
			\leq& \left\|\mathbf{Y}-\mathbf S_Y\right\| \left\| \mathbf{U}_{S_Y}\right\|_F\left\|\mathbf{W_{YS}} \boldsymbol{\Sigma}_{\mathbf{Y}}^{-1 / 2}-\boldsymbol{\Sigma}_{S_Y}^{-1 / 2} \mathbf{W_{YS}}\right\|_{F}\\
			=&O_{\mathbb P}\left(\frac{\log ^{3 / 2}(p)}{\rho_Yp}\right). 
		\end{align*}
	\end{proof}
	
	\begin{proposition}\label{prop:R5-bound}
		$\left\|(\mathbf{Y}-\mathbf S_Y)^{\top} \mathbf{U}_{S_Y} \boldsymbol{\Sigma}_{S_Y}^{-1 / 2} \right\|_{2\rightarrow\infty} = O_{\mathbb P}(\sqrt{\frac{\log p}{p}})$.
	\end{proposition}
	\begin{proof}
		By the smallest singular-value bound,
		\[
		\left\|(\mathbf{Y}-\mathbf S_Y)^{\top}\mathbf{U}_{S_Y}
		\boldsymbol{\Sigma}_{S_Y}^{-1/2}\right\|_{2\to\infty}
		\le
		\sigma_{K_Y}(\mathbf S_Y)^{-1/2}
		\left\|(\mathbf Y-\mathbf S_Y)^\top\mathbf U_{\mathbf S_Y}\right\|_{2\to\infty}.
		\]
		Set
		\(
		\boldsymbol\xi_{t,j}^{\top}
		=
		(\mathbf Y_{\cdot j}^{(t)}-\mathbf S_{Y,\cdot j}^{(t)})^\top
		\mathbf U_{\mathbf S_Y}.
		\)
		Then
		\(
		\left\|(\mathbf Y-\mathbf S_Y)^\top\mathbf U_{\mathbf S_Y}\right\|_{2\to\infty}
		=
		\max_{t,j}\|\boldsymbol\xi_{t,j}\|.
		\)
		Write
		\(
		\boldsymbol\xi_{t,j}
		=
		\sum_{i\ne j}\mathbf M_i+D_j^{(t)}
		\)
		where $\mathbf M_i=
		\{Y_{ij}^{(t)}-(S_Y^{(t)})_{ij}\}
		(\mathbf U_{\mathbf S_Y})_{i\cdot}^{\top}$ and 
		$D_j^{(t)}=-(S_Y^{(t)})_{jj}(\mathbf U_{\mathbf S_Y})_{j\cdot}^{\top}$.
		By Proposition~\ref{prop:uv2},
\(\|D_j^{(t)}\|=O(\rho_Y/\sqrt p)\).  The summands \(\mathbf M_i\) are independent and centred.  Moreover,
		\(
		\|\mathbf M_i\|
\le
\|\mathbf U_{\mathbf S_Y}\|_{2\to\infty}
=O(p^{-1/2}),
		\)
		and, using \(\operatorname{Var}(Y_{ij}^{(t)})\le \rho_Y\), we have 
		\(
		\left\|\sum_{i\ne j}\mathbb E(\mathbf M_i\mathbf M_i^\top)\right\|
		\le \rho_Y,
		\sum_{i\ne j}\mathbb E(\mathbf M_i^\top\mathbf M_i)
		\le \rho_Y K_Y .
		\)
		Thus the Bernstein variance term is bounded by \(v\le \rho_YK_Y\), and
		Lemma~\ref{lem:matrixBernstein} gives
		\[
		\mathbb{P}\!\left(\left\|\sum_{i\neq j}\mathbf{M}_i\right\|\ge\gamma\right)
		\le
		(K_Y+1)
		\exp\left\{
		\frac{-\gamma^2/2}{\rho_YK_Y+C_L\gamma/(3\sqrt p)}
		\right\}.
		\]
		Fix \(c_0>0\), and choose
		\(\gamma=C_{\gamma,c_0}\sqrt{\rho_Y\log p}\), where
		\(C_{\gamma,c_0}\) is sufficiently large satisfying
		\(
		C_{\gamma,c_0}
		>
		\frac{
		(c_0+1)C_L+
		\sqrt{(c_0+1)^2C_L^2+18(c_0+1)K_Y}
		}{3}.\)  Since
		Condition~\ref{ass:transition-sparsity} implies \(\rho_Y>\log p/p\) for large
		\(p\),
		\begin{align*}
			\mathbb{P}(\max_{t,j}\|\sum_{i\neq j}\mathbf{M}_i\| \geq \gamma) \leq & np\left(K_Y+1\right) \exp \left(\frac{-\gamma^2 / 2}{\rho_YK_Y+C_L\gamma / 3\sqrt{p}}\right) \\
			\leq & np\left(K_Y+1\right) \exp \left(\frac{-C_{\gamma,c_0}^2\rho_Y\log p / 2}{\rho_YK_Y+\frac{C_LC_{\gamma,c_0}}{3}\sqrt{\frac{\log p\rho_Y}{p}}}\right)\\
			\leq & np\left(K_Y+1\right) \exp \left(\frac{-C_{\gamma,c_0}^2\rho_Y\log p / 2}{\rho_Y(K_Y+\frac{C_LC_{\gamma,c_0}}{3})}\right)\\
			\leq & np\left(K_Y+1\right) \exp \left(\frac{-C_{\gamma,c_0}^2}{2(K_Y+\frac{C_LC_{\gamma,c_0}}{3})}\log p\right)\\
			=&n(K_Y+1)p^{1-\frac{3C_{\gamma,c_0}^2}{6K_Y+2C_LC_{\gamma,c_0}}}\\
            \leq&n(K_Y+1)p^{-c_0}.
		\end{align*}
		Consequently,
		\[
		\left\|(\mathbf Y-\mathbf S_Y)^\top\mathbf U_{\mathbf S_Y}\right\|_{2\to\infty}
		=O_{\mathbb P}(\sqrt{\rho_Y\log p}).
		\]
		Condition~\ref{ass:3} gives 
		\(\sigma_{K_Y}(\mathbf S_Y)^{-1/2}=O\{(p\rho_Y)^{-1/2}\}\), and
		the displayed bound follows.
	\end{proof}

	\begin{proof}[Proof of Theorem~\ref{thm1}]
		We use
		\[
		\mathbf{Y}^{\top} \mathbf{U}_{\mathbf{Y}} \boldsymbol{\Sigma}_{\mathbf{Y}}^{-1 / 2}
		=
		\mathbf{V}_{\mathbf{Y}} \boldsymbol{\Sigma}_{\mathbf{Y}}^{1 / 2},
		\qquad
		\mathbf{V}_{S_Y} \mathbf{V}_{S_Y}^{\top} \mathbf S_Y^{\top}
		=
		\mathbf S_Y^{\top}.
		\]
        Then
\begin{align*}
\mathbf R_Y-\mathbf R_{S_Y}\mathbf W_{YS}
&=
\mathbf V_Y\boldsymbol\Sigma_Y^{1/2}
-
\mathbf V_{S_Y}\boldsymbol\Sigma_{S_Y}^{1/2}
\mathbf W_{YS} \\
&=
\mathbf V_Y\boldsymbol\Sigma_Y^{1/2}
-
\mathbf V_{S_Y}\mathbf V_{S_Y}^{\top}
\mathbf V_Y\boldsymbol\Sigma_Y^{1/2}
+
R_1 \\
&=
\mathbf Y^\top\mathbf U_Y\boldsymbol\Sigma_Y^{-1/2}
-
\mathbf V_{S_Y}\mathbf V_{S_Y}^{\top}
\mathbf Y^\top\mathbf U_Y\boldsymbol\Sigma_Y^{-1/2}
+
R_1 \\
&=
(\mathbf Y-\mathbf S_Y)^\top
\mathbf U_Y\boldsymbol\Sigma_Y^{-1/2} \\
&\quad-
\left(
\mathbf V_{S_Y}\mathbf V_{S_Y}^{\top}\mathbf Y^\top
-
\mathbf S_Y^\top
\right)
\mathbf U_Y\boldsymbol\Sigma_Y^{-1/2}
+
R_1 \\
&=
\left(
\mathbf I-
\mathbf V_{S_Y}\mathbf V_{S_Y}^{\top}
\right)
(\mathbf Y-\mathbf S_Y)^\top
\mathbf U_Y\boldsymbol\Sigma_Y^{-1/2}
+
R_1 \\
&=
\left(
\mathbf I-
\mathbf V_{S_Y}\mathbf V_{S_Y}^{\top}
\right)
(\mathbf Y-\mathbf S_Y)^\top
\left\{
\mathbf U_{S_Y}\mathbf W_{YS}
+
\left(
\mathbf U_Y-\mathbf U_{S_Y}\mathbf W_{YS}
\right)
\right\}
\boldsymbol\Sigma_Y^{-1/2}
+
R_1 \\
&=
(\mathbf Y-\mathbf S_Y)^\top
\mathbf U_{S_Y}\mathbf W_{YS}
\boldsymbol\Sigma_Y^{-1/2}
+
R_3+R_2+R_1 \\
&=
R_5+R_4+R_3+R_2+R_1,
\end{align*}
where
\[
R_5
=
(\mathbf Y-\mathbf S_Y)^\top
\mathbf U_{S_Y}
\boldsymbol\Sigma_{S_Y}^{-1/2}
\mathbf W_{YS}.
\]
By the preceding proposition and
Proposition~\ref{prop:last4},
\[
\begin{aligned}
\left\|
\mathbf R_Y-\mathbf R_{S_Y}\mathbf W_{YS}
\right\|_{2\to\infty}
\le
\sum_{j=1}^{5}\|R_j\|_{2\to\infty} 
=
O_{\mathbb P}
\left(
\sqrt{\frac{\log p}{p}}
\right),
\end{aligned}
\]
where Condition~\ref{ass:transition-sparsity} ensures that the
rates for \(R_1,\ldots,R_4\) are no larger than the displayed rate.
Since the blocks partition the rows of the unfolded embedding,
\[
\max_{1\le t\le n}
\left\|
\mathbf R_Y^{(t)}
-
\mathbf R_{S_Y}^{(t)}\mathbf W_{YS}
\right\|_{2\to\infty}
=
\left\|
\mathbf R_Y-\mathbf R_{S_Y}\mathbf W_{YS}
\right\|_{2\to\infty}.
\]
The \(Z\)-channel result follows by the same argument. Fix \(c_0>0\). By the polynomial-tail versions of
Proposition~\ref{prop:last4} and Proposition~\ref{prop:R5-bound},
a union bound over \(R_1,\ldots,R_5\) gives constants
\(C_{Y,c_0},C'_{c_0}<\infty\) such that
\[
\mathbb P\left\{
\left\|
\mathbf R_Y-\mathbf R_{S_Y}\mathbf W_{YS}
\right\|_{2\to\infty}
>
C_{Y,c_0}\sqrt{\frac{\log p}{p}}
\right\}
\le
C'_{c_0}p^{-c_0}.
\]
The same argument gives the corresponding bound for the \(Z\) channel.
A union bound over the two channels proves the simultaneous bounds in
Theorem~\ref{thm1}.
	\end{proof}
	
	\subsection*{S.2 Supporting results and proofs for Section~3}

Appendix~S.2.1 proves Lemma~\ref{lem:sparse-xi-main}.
Appendix~S.2.2 proves Theorem~\ref{thm:2v0} and
Corollary~\ref{cor:perfect-v}.
Appendix~S.2.3 establishes
Proposition~\ref{prop:memory-geometric-decay},
Theorem~\ref{thm:memory-echo-main}, and
Corollary~\ref{cor:vertex-memory-burnin-main}.
Appendix~S.2.4 proves the network-level detection result,
Theorem~\ref{thm:2g}.
    
	In the following proofs, let
\(\mathcal E_{\rm UASE}(c_0)\) denote the event on which the bound
in Theorem~\ref{thm1} holds.

\subsubsection*{S.2.1 Proof of Lemma~\ref{lem:sparse-xi-main}}
	
\begin{proof}
		For every \(j\ne i\) such that \(\delta_{Y,j}^{(t^*-1,t)}>0\),
		$ \|\boldsymbol\ell_{Y,j}^{(t)}-\boldsymbol\ell_{Y,j}^{(t^*-1)}\|_2
		\leq
		2\sqrt d\,\ell.$
		Also
		$\|\boldsymbol\ell_{Y,i}^{(t)}\|_2
		\leq
		\sqrt d\,\ell.$
		Hence
		$  \left|
		(\boldsymbol\ell_{Y,j}^{(t)}-\boldsymbol\ell_{Y,j}^{(t^*-1)})^\top
		\boldsymbol\ell_{Y,i}^{(t)}
		\right|
		\leq
		2d\ell^2 .$
		Since at most \(m_Y^{(t^*-1,t)}\) terms are nonzero,
		\[
		\begin{aligned}
			\xi_{Y,i}^{(t^*-1,t)}
			=
			\sqrt{\rho_Y/p}
			\left\|
			\Delta\mathbf L_{Y,-i}^{(t^*-1,t)}
			\boldsymbol\ell_{Y,i}^{(t)}
			\right\|_2        
			\leq
			\sqrt{\rho_Y/p}
			\left(
			m_Y^{(t^*-1,t)}
			\cdot
			4d^2\ell^4
			\right)^{1/2}        
			=
			2d\ell^2
			\sqrt{\frac{\rho_Y m_Y^{(t^*-1,t)}}{p}}.
		\end{aligned}
		\]
		The proof for \(Z\) is identical. Hence, for either
		\(h\in\mathcal H\), the displayed bound holds uniformly in \(i\).
		If \(\rho_hm_h^{(t^*-1,t)}=o(\log p)\), then, uniformly in
		\(i\),
		\[
		\frac{\xi_{h,i}^{(t^*-1,t)}}{\epsilon_p}
		\leq
		C\left\{\frac{\rho_hm_h^{(t^*-1,t)}}{\log p}\right\}^{1/2}
		=o(1),
		\]
		which proves the stated conclusion.
	\end{proof}
	
	\subsubsection*{S.2.2 Proofs of Theorem~\ref{thm:2v0}
and Corollary~\ref{cor:perfect-v}}
	
		\begin{proposition}\label{prop:column-wise-SY-bound}
		For \(t\geq t^*\), let $\xi_{Y,i}^{(t^*-1,t)}$, $M_{Y,i}^{(t^*-1,t)}$, and
		$\delta_{Y,i}^{(t^*-1,t)}$ be defined as in the main paper.
		Assume 
		Condition~\ref{ass:column-visibility} holds.
		Then
		\[
		\frac{1}{\sqrt{p}}
		\left\|\left(\mathbf S_Y^{(t)}-\mathbf S_Y^{(t^*-1)}\right)\mathbf e_i\right\|_2
		\ge
		\rho_Y\mu_{Y,i}\,\delta_{Y,i}^{(t^*-1,t)}
		-\sqrt{\rho_Y}\,\xi_{Y,i}^{(t^*-1,t)}
		-\sqrt{\rho_Y}\,M_{Y,i}^{(t^*-1,t)},
		\]
		\[
		\frac{1}{\sqrt{p}}
		\left\|\left(\mathbf S_Y^{(t)}-\mathbf S_Y^{(t^*-1)}\right)\mathbf e_i\right\|_2
		\le		\rho_Y\sqrt{d}\,\ell\,\delta_{Y,i}^{(t^*-1,t)}+\sqrt{\rho_Y}\xi_{Y,i}^{(t^*-1,t)}
		+\sqrt{\rho_Y}\,M_{Y,i}^{(t^*-1,t)}
		+\frac{\rho_Y}{\sqrt p}.
		\]
		The same bounds hold for the \(Z\) channel after replacing the corresponding
		\(Y\)-channel quantities by their \(Z\)-channel counterparts.
	\end{proposition}
	
	\begin{proof}
		The population difference decomposes as
		\[
		\begin{aligned}
			\mathbf S_Y^{(t)}-\mathbf S_Y^{(t^*-1)}
			&=
			\left\{
			\mathbf P_Y^{(t)}
			-
			\mathbf P_Y^{(t^*-1)}
			\right\}
			\circ
			\left\{
			\mathbf 1\mathbf 1^\top-\mathbf\Pi^{(t^*-1)}
			\right\}        \\
			&\qquad
			-
			\mathbf P_Y^{(t)}
			\circ
			\left\{
			\mathbf\Pi^{(t-1)}
			-
			\mathbf\Pi^{(t^*-1)}
			\right\}.
		\end{aligned}
		\]
		Multiplying by \(\mathbf e_i\) gives
		\[
		\begin{aligned}
			\left(
			\mathbf S_Y^{(t)}-\mathbf S_Y^{(t^*-1)}
			\right)\mathbf e_i
			&=
			\left[
			\left\{
			\mathbf P_Y^{(t)}
			-
			\mathbf P_Y^{(t^*-1)}
			\right\}
			\circ
			\left\{
			\mathbf 1\mathbf 1^\top-\mathbf\Pi^{(t^*-1)}
			\right\}
			\right]\mathbf e_i        \\
			&\qquad
			-
			\left[
			\mathbf P_Y^{(t)}
			\circ
			\left\{
			\mathbf\Pi^{(t-1)}
			-
			\mathbf\Pi^{(t^*-1)}
			\right\}
			\right]\mathbf e_i .
		\end{aligned}
		\]
		
		Define
		\[
		\begin{aligned}
			A_{Y,i}^{(t^*-1,t)}
			&:=
			\frac{1}{\sqrt p}
			\left\|
			\left(I-\mathbf e_i \mathbf e_i^\top\right)
			\left[
			\left\{
			\mathbf P_Y^{(t)}
			-
			\mathbf P_Y^{(t^*-1)}
			\right\}
			\circ
			\left\{
			\mathbf 1\mathbf 1^\top-\mathbf\Pi^{(t^*-1)}
			\right\}
			\right]\mathbf e_i
			\right\|_2,        \\
			B_{Y,i}^{(t^*-1,t)}
			&:=
			\frac{1}{\sqrt p}
			\left\|
			\left(I-\mathbf e_i \mathbf e_i^\top\right)
			\left[
			\mathbf P_Y^{(t)}
			\circ
			\left\{
			\mathbf\Pi^{(t-1)}
			-
			\mathbf\Pi^{(t^*-1)}
			\right\}
			\right]\mathbf e_i
			\right\|_2 .
		\end{aligned}
		\]
		The reverse triangle inequality applied after removing the completed diagonal
		coordinate yields
		\[
		\begin{aligned}
			\frac{1}{\sqrt p}
			\left\|
			\left(
			\mathbf S_Y^{(t)}-\mathbf S_Y^{(t^*-1)}
			\right)\mathbf e_i
			\right\|_2
			&\geq
			\frac{1}{\sqrt p}
			\left\|
			\left(I-\mathbf e_i \mathbf e_i^\top\right)
			\left(
			\mathbf S_Y^{(t)}-\mathbf S_Y^{(t^*-1)}
			\right)\mathbf e_i
			\right\|_2\\
			\geq
			A_{Y,i}^{(t^*-1,t)}
			-
			B_{Y,i}^{(t^*-1,t)} .
		\end{aligned}
		\]
		
		We first lower bound \(A_{Y,i}^{(t^*-1,t)}\).  For \(j\neq i\),
		\[
		\begin{aligned}
			\alpha_{ji}^{(t)}-\alpha_{ji}^{(t^*-1)}
			&=
			\rho_Y
			\left\{
			(\boldsymbol\ell_{Y,j}^{(t)})^\top \boldsymbol\ell_{Y,i}^{(t)}
			-
			(\boldsymbol\ell_{Y,j}^{(t^*-1)})^\top
			\boldsymbol\ell_{Y,i}^{(t^*-1)}
			\right\}        \\
			&=
			\rho_Y
			\left\{
			(\boldsymbol\ell_{Y,j}^{(t^*-1)})^\top
			\Delta\boldsymbol\ell_{Y,i}^{(t^*-1,t)}
			+
			\left(
			\boldsymbol\ell_{Y,j}^{(t)}
			-
			\boldsymbol\ell_{Y,j}^{(t^*-1)}
			\right)^\top
			\boldsymbol\ell_{Y,i}^{(t)}
			\right\}.
		\end{aligned}
		\]
		Hence
		\[
		\begin{aligned}
			A_{Y,i}^{(t^*-1,t)}
			&\geq
			\rho_Y
			\frac{1}{\sqrt p}
			\left\|
			\operatorname{diag}
			\left(
			\{1-\Pi_{ji}^{(t^*-1)}\}_{j\neq i}
			\right)
			\mathbf L_{Y,-i}^{(t^*-1)}
			\Delta\boldsymbol\ell_{Y,i}^{(t^*-1,t)}
			\right\|_2        \\
			&\qquad
			-
			\rho_Y
			\frac{1}{\sqrt p}
			\left\|
			\operatorname{diag}
			\left(
			\{1-\Pi_{ji}^{(t^*-1)}\}_{j\neq i}
			\right)
			\Delta\mathbf L_{Y,-i}^{(t^*-1,t)}
			\boldsymbol\ell_{Y,i}^{(t)}
			\right\|_2 .
		\end{aligned}
		\]
		By Condition~\ref{ass:column-visibility},
		\[
		\frac{1}{\sqrt p}
		\left\|
		\operatorname{diag}
		\left(
		\{1-\Pi_{ji}^{(t^*-1)}\}_{j\neq i}
		\right)
		\mathbf L_{Y,-i}^{(t^*-1)}
		\Delta\boldsymbol\ell_{Y,i}^{(t^*-1,t)}
		\right\|_2
		\geq
		\mu_{Y,i}\delta_{Y,i}^{(t^*-1,t)}.
		\]
		Moreover, since \(0\leq 1-\Pi_{ji}^{(t^*-1)}\leq 1\),
		\[
		\begin{aligned}
			\frac{1}{\sqrt p}
			\left\|
			\operatorname{diag}
			\left(
			\{1-\Pi_{ji}^{(t^*-1)}\}_{j\neq i}
			\right)
			\Delta\mathbf L_{Y,-i}^{(t^*-1,t)}
			\boldsymbol\ell_{Y,i}^{(t)}
			\right\|_2
			\leq
			\frac{1}{\sqrt p}
			\left\|
			\Delta\mathbf L_{Y,-i}^{(t^*-1,t)}
			\boldsymbol\ell_{Y,i}^{(t)}
			\right\|_2        
			=
			\xi_{Y,i}^{(t^*-1,t)}/\sqrt{\rho_Y}.
		\end{aligned}
		\]
		Therefore, we have
		\(
		A_{Y,i}^{(t^*-1,t)}
		\geq
		\rho_Y\mu_{Y,i}\delta_{Y,i}^{(t^*-1,t)}
		-
		\sqrt{\rho_Y}\xi_{Y,i}^{(t^*-1,t)}.
		\)
		
		Next, by the definition of \(M_{Y,i}^{(t^*-1,t)}\),
		\[
		\begin{aligned}
			B_{Y,i}^{(t^*-1,t)}
			&=
			\frac{1}{\sqrt p}
			\left[
			\sum_{j\neq i}
			\left\{
			\alpha_{ji}^{(t)}
			\left(
			\Pi_{ji}^{(t-1)}
			-
			\Pi_{ji}^{(t^*-1)}
			\right)
			\right\}^2
			\right]^{1/2}        \\
			&=
			\sqrt{\rho_Y}
			\left[
			\frac{1}{p}
			\sum_{j\neq i}
			\left\{
			\frac{\alpha_{ji}^{(t)}}{\sqrt{\rho_Y}}
			\left(
			\Pi_{ji}^{(t-1)}
			-
			\Pi_{ji}^{(t^*-1)}
			\right)
			\right\}^2
			\right]^{1/2}        \\
			&=
			\sqrt{\rho_Y}\,
			M_{Y,i}^{(t^*-1,t)} .
		\end{aligned}
		\]
		Combining the bounds for \(A_{Y,i}^{(t^*-1,t)}\) and \(B_{Y,i}^{(t^*-1,t)}\), we obtain
		\[
		\begin{aligned}
			\frac{1}{\sqrt p}
			\left\|
			\left(
			\mathbf S_Y^{(t)}-\mathbf S_Y^{(t^*-1)}
			\right)\mathbf e_i
			\right\|_2
			&\geq
			\rho_Y\mu_{Y,i}\delta_{Y,i}^{(t^*-1,t)}
			-
			\sqrt{\rho_Y}\xi_{Y,i}^{(t^*-1,t)}
			-
			\sqrt{\rho_Y}\,
			M_{Y,i}^{(t^*-1,t)} .
		\end{aligned}
		\]
		This establishes the lower bound.
		
		For the upper bound, the off-diagonal coordinates satisfy the same triangle
		inequality.  Since the completed population entries satisfy
		\(
		0\leq (\mathbf S_Y^{(t)})_{ii}, (\mathbf S_Y^{(t^*-1)})_{ii}\leq\rho_Y,
		\left|(\mathbf S_Y^{(t)})_{ii}-(\mathbf S_Y^{(t^*-1)})_{ii}\right|
		\leq\rho_Y,
		\)
		we have
		\[
		\begin{aligned}
			\frac{1}{\sqrt p}
			\left\|
			\left(
			\mathbf S_Y^{(t)}-\mathbf S_Y^{(t^*-1)}
			\right)\mathbf e_i
			\right\|_2
			&\leq
			A_{Y,i}^{(t^*-1,t)}
			+
			B_{Y,i}^{(t^*-1,t)}
			+
			\frac{\rho_Y}{\sqrt p} .
		\end{aligned}
		\]
		Using the same decomposition of
		\(\alpha_{ji}^{(t)}-\alpha_{ji}^{(t^*-1)}\), we have
		\[
		\begin{aligned}
			A_{Y,i}^{(t^*-1,t)}
			&\leq
			\rho_Y
			\frac{1}{\sqrt p}
			\left\|
			\operatorname{diag}
			\left(
			\{1-\Pi_{ji}^{(t^*-1)}\}_{j\neq i}
			\right)
			\mathbf L_{Y,-i}^{(t^*-1)}
			\Delta\boldsymbol\ell_{Y,i}^{(t^*-1,t)}
			\right\|_2        \\
			&\qquad
			+
			\rho_Y
			\frac{1}{\sqrt p}
			\left\|
			\operatorname{diag}
			\left(
			\{1-\Pi_{ji}^{(t^*-1)}\}_{j\neq i}
			\right)
			\Delta\mathbf L_{Y,-i}^{(t^*-1,t)}
			\boldsymbol\ell_{Y,i}^{(t)}
			\right\|_2 .
		\end{aligned}
		\]
		Since \(0\leq 1-\Pi_{ji}^{(t^*-1)}\leq 1\) and
		\(\|\boldsymbol\ell_{Y,j}^{(t^*-1)}\|_2\leq \sqrt d\,\ell\),
		\[
		\frac{1}{\sqrt p}
		\left\|
		\operatorname{diag}
		\left(
		\{1-\Pi_{ji}^{(t^*-1)}\}_{j\neq i}
		\right)
		\mathbf L_{Y,-i}^{(t^*-1)}
		\Delta\boldsymbol\ell_{Y,i}^{(t^*-1,t)}
		\right\|_2
		\leq
		\sqrt d\,\ell\,\delta_{Y,i}^{(t^*-1,t)}.
		\]
		Also,
		\[
		\frac{1}{\sqrt p}
		\left\|
		\operatorname{diag}
		\left(
		\{1-\Pi_{ji}^{(t^*-1)}\}_{j\neq i}
		\right)
		\Delta\mathbf L_{Y,-i}^{(t^*-1,t)}
		\boldsymbol\ell_{Y,i}^{(t)}
		\right\|_2
		\leq
		\xi_{Y,i}^{(t^*-1,t)}/\sqrt{\rho_Y}.
		\]
		Thus, we have
		\(
		A_{Y,i}^{(t^*-1,t)}
		\leq
		\rho_Y
		\sqrt d\,\ell\,\delta_{Y,i}^{(t^*-1,t)}
		+
		\sqrt{\rho_Y}\xi_{Y,i}^{(t^*-1,t)}.
		\)
		As shown above,
		\(
		B_{Y,i}^{(t^*-1,t)}
		=
		\sqrt{\rho_Y}\,
		M_{Y,i}^{(t^*-1,t)}.
		\)
		Therefore,
		\[
		\begin{aligned}
			\frac{1}{\sqrt p}
			\left\|
			\left(
			\mathbf S_Y^{(t)}-\mathbf S_Y^{(t^*-1)}
			\right)\mathbf e_i
			\right\|_2
			&\leq
			\rho_Y
			\sqrt d\,\ell\,\delta_{Y,i}^{(t^*-1,t)}
			+
			\sqrt{\rho_Y}\xi_{Y,i}^{(t^*-1,t)}
			+
			\sqrt{\rho_Y}\,
			M_{Y,i}^{(t^*-1,t)}
			+
			\frac{\rho_Y}{\sqrt p} .
		\end{aligned}
		\]
		Together the lower and upper estimates prove the proposition.  The proof for
		\(Z\) is identical.
	\end{proof}

	\begin{proposition}
		\label{prop:pointwise-S-to-R}
		For every \(i\in[p]\),
		\[
		\left\|
		\mathbf e_i^\top
		\left(
		\mathbf R_{\mathbf S_Y}^{(t)}
		-
		\mathbf R_{\mathbf S_Y}^{(t^*-1)}
		\right)
		\right\|_2
		\geq
		\frac{1}{\sqrt{C_{Y,S}\rho_Y}}
		\frac{1}{\sqrt p}
		\left\|
		\left(
		\mathbf S_Y^{(t)}
		-
		\mathbf S_Y^{(t^*-1)}
		\right)\mathbf e_i
		\right\|_2.
		\]
		Moreover,
		\[
		\left\|
		\mathbf e_i^\top
		\left(
		\mathbf R_{\mathbf S_Y}^{(t)}
		-
		\mathbf R_{\mathbf S_Y}^{(t^*-1)}
		\right)
		\right\|_2
		\leq
		\frac{1}{\sqrt{c_{Y,S}\rho_Y}}
		\frac{1}{\sqrt p}
		\left\|
		\left(
		\mathbf S_Y^{(t)}
		-
		\mathbf S_Y^{(t^*-1)}
		\right)\mathbf e_i
		\right\|_2.
		\]
		The same statement holds for \(\mathbf S_Z\), with
		\(\rho_Y,c_{Y,S},C_{Y,S}\) replaced by \(\rho_Z,c_{Z,S},C_{Z,S}\).
	\end{proposition}
	
	\begin{proof}
		We prove the statement for \(Y\). The proof for \(Z\) is identical.
		
		By the population SVD, for each block \(t\), we have
		\(
		\mathbf S_Y^{(t)}
		=
		\mathbf U_{S_Y}\mathbf \Sigma_{S_Y}^{1/2}
		\left(
		\mathbf R_{\mathbf S_Y}^{(t)}
		\right)^\top .
		\)
		Hence
		\[
		\mathbf S_Y^{(t)}-\mathbf S_Y^{(t^*-1)}
		=
		\mathbf U_{S_Y}\mathbf \Sigma_{S_Y}^{1/2}
		\left(
		\mathbf R_{\mathbf S_Y}^{(t)}
		-
		\mathbf R_{\mathbf S_Y}^{(t^*-1)}
		\right)^\top .
		\]
		Multiplying by \(\mathbf e_i\) gives
		\[
		\left(
		\mathbf S_Y^{(t)}-\mathbf S_Y^{(t^*-1)}
		\right)\mathbf e_i
		=
		U_{S_Y}\mathbf \Sigma_{S_Y}^{1/2}
		\left[
		\mathbf e_i^\top
		\left(
		\mathbf R_{\mathbf S_Y}^{(t)}
		-
		\mathbf R_{\mathbf S_Y}^{(t^*-1)}
		\right)
		\right]^\top .
		\]
        
		Then the two singular-value bounds in Condition~\ref{ass:3}, give the two inequalities.
	\end{proof}

	\begin{proof}[Proof of Theorem~\ref{thm:2v0}]
		By Theorem~\ref{thm1},
		\(\mathbb P\{\mathcal E_{\rm UASE}(c_0)^c\}\le C_{c_0}p^{-c_0}\), and on
		\(\mathcal E_{\rm UASE}(c_0)\) the rowwise embedding bounds hold in both
		channels for every block used below.
		
		(1) Under \(H_{0,i}^{(t^*-1,t)}\), the own latent displacement of vertex \(i\)
		vanishes in both transition channels.  Thus the own-vertex visible signal terms in
		Proposition~\ref{prop:column-wise-SY-bound} vanish; the remaining population
		row displacement is bounded by the interference and memory terms.
		By Proposition~\ref{prop:column-wise-SY-bound} and Proposition~\ref{prop:pointwise-S-to-R},
		\[
		\begin{aligned}
			\left\|
			\mathbf e_i^\top
			\left(
			\mathbf R_{\mathbf S_Y}^{(t)}
			-
			\mathbf R_{\mathbf S_Y}^{(t^*-1)}
			\right)
			\right\|_2
			&\le
			\frac{1}{\sqrt{c_{Y,S}}}
			\left(
			\xi_{Y,i}^{(t^*-1,t)}
			+
			M_{Y,i}^{(t^*-1,t)}
			+
			\sqrt{\frac{\rho_Y}{p}}
			\right),
			\\
			\left\|
			\mathbf e_i^\top
			\left(
			\mathbf R_{\mathbf S_Z}^{(t)}
			-
			\mathbf R_{\mathbf S_Z}^{(t^*-1)}
			\right)
			\right\|_2
			&\le
			\frac{1}{\sqrt{c_{Z,S}}}
			\left(
			\xi_{Z,i}^{(t^*-1,t)}
			+
			M_{Z,i}^{(t^*-1,t)}
			+
			\sqrt{\frac{\rho_Z}{p}}
			\right).
		\end{aligned}
		\]
		Hence the null condition implies
		\[
		\max
		\left\{
		\left\|
		\mathbf e_i^\top
		\left(
		\mathbf R_{\mathbf S_Y}^{(t)}
		-
		\mathbf R_{\mathbf S_Y}^{(t^*-1)}
		\right)
		\right\|_2,
		\left\|
		\mathbf e_i^\top
		\left(
		\mathbf R_{\mathbf S_Z}^{(t)}
		-
		\mathbf R_{\mathbf S_Z}^{(t^*-1)}
		\right)
		\right\|_2
		\right\}
		<
		(\gamma-2)\epsilon_p .
		\]
		On \(\mathcal E_{\rm UASE}(c_0)\), the embedding error contributes at most \(2\epsilon_p\), so
		\[
		T_i^{(t^*-1,t)}
		<
		(\gamma-2)\epsilon_p+2\epsilon_p
		=
		\gamma\epsilon_p
		=
		\tau_p.
		\]
		Thus \(\phi_i^{(t^*-1,t)}=0\) on \(\mathcal E_{\rm UASE}(c_0)\). Therefore
		\[
		\mathbb P_{H_{0,i}^{(t^*-1,t)}}
		\left(
		\phi_i^{(t^*-1,t)}=1
		\right)
		\le
		\mathbb P\{\mathcal E_{\rm UASE}(c_0)^c\}
		\le
		C_{c_0}p^{-c_0}.
		\]
		
		(2) Under \(H_{1,i}^{(t^*-1,t)}\), Proposition~\ref{prop:column-wise-SY-bound} and
		Proposition~\ref{prop:pointwise-S-to-R} give
		\[
		\begin{aligned}
			\left\|
			\mathbf e_i^\top
			\left(
			\mathbf R_{\mathbf S_Y}^{(t)}
			-
			\mathbf R_{\mathbf S_Y}^{(t^*-1)}
			\right)
			\right\|_2
			&\ge
			\frac{1}{\sqrt{C_{Y,S}}}
			\left\{
			\sqrt{\rho_Y}
			\mu_{Y,i}\delta_{Y,i}^{(t^*-1,t)}
			-
			\xi_{Y,i}^{(t^*-1,t)}
			-
			M_{Y,i}^{(t^*-1,t)}
			\right\},
			\\
			\left\|
			\mathbf e_i^\top
			\left(
			\mathbf R_{\mathbf S_Z}^{(t)}
			-
			\mathbf R_{\mathbf S_Z}^{(t^*-1)}
			\right)
			\right\|_2
			&\ge
			\frac{1}{\sqrt{C_{Z,S}}}
			\left\{
			\sqrt{\rho_Z}
			\mu_{Z,i}\delta_{Z,i}^{(t^*-1,t)}
			-
			\xi_{Z,i}^{(t^*-1,t)}
			-
			M_{Z,i}^{(t^*-1,t)}
			\right\}.
		\end{aligned}
		\]
		Hence the separation condition \eqref{eq:vertex_sep} implies
		\[
		\max
		\left\{
		\left\|
		\mathbf e_i^\top
		\left(
		\mathbf R_{\mathbf S_Y}^{(t)}
		-
		\mathbf R_{\mathbf S_Y}^{(t^*-1)}
		\right)
		\right\|_2,
		\left\|
		\mathbf e_i^\top
		\left(
		\mathbf R_{\mathbf S_Z}^{(t)}
		-
		\mathbf R_{\mathbf S_Z}^{(t^*-1)}
		\right)
		\right\|_2
		\right\}
		>
		(\gamma+2)\epsilon_p .
		\]
		On \(\mathcal E_{\rm UASE}(c_0)\), the embedding error is at most \(2\epsilon_p\), so
		\[
		T_i^{(t^*-1,t)}
		>
		(\gamma+2)\epsilon_p-2\epsilon_p
		=
		\gamma\epsilon_p
		=
		\tau_p.
		\]
		Thus \(\phi_i^{(t^*-1,t)}=1\) on \(\mathcal E_{\rm UASE}(c_0)\). Therefore
		\[
		\mathbb P_{H_{1,i}^{(t^*-1,t)}}
		\left(
		\phi_i^{(t^*-1,t)}=1
		\right)
		\ge
		\mathbb P\{\mathcal E_{\rm UASE}(c_0)\}
		\ge
		1-C_{c_0}p^{-c_0}.
		\]
	\end{proof}
	
	\begin{proof}[Proof of Corollary~\ref{cor:perfect-v}]
		The two displays in Theorem~\ref{thm:2v0} are per-vertex statements. Under the
		stated hypotheses the null condition~\eqref{eq:null_drift_control} holds for
		every $i\notin\mathcal A^{(t^*-1,t)}$ and the separation
		condition~\eqref{eq:vertex_sep} holds for every $i\in\mathcal A^{(t^*-1,t)}$.
		On the common embedding event \(\mathcal E_{\rm UASE}(c_0)\), the proof of
		Theorem~\ref{thm:2v0} therefore classifies every null and anomalous vertex
		correctly.  Hence the recovery failure is contained in
		\(\mathcal E_{\rm UASE}(c_0)^c\), and
		\[
		\mathbb P\!\left(\widehat{\mathcal A}^{(t^*-1,t)}\neq\mathcal A^{(t^*-1,t)}\right)
		\le
		\mathbb P\{\mathcal E_{\rm UASE}(c_0)^c\}
		\le
		C_{c_0}p^{-c_0}.
		\]
		Taking complements gives
		$\mathbb P\!\left(\widehat{\mathcal A}^{(t^*-1,t)}=\mathcal A^{(t^*-1,t)}\right)\ge
		1-C_{c_0}p^{-c_0}$. This lower bound tends to one as $p\to\infty$. Applying the same argument to the channel-specific tests
		based on \(T_{Y,i}^{(t^*-1,t)}\) and \(T_{Z,i}^{(t^*-1,t)}\), and taking the
		intersection of their common embedding events, gives exact recovery of
		the formation- and dissolution-anomaly sets. This proves the stated
		channel-attribution conclusion.
	\end{proof}

\subsubsection*{S.2.3 Proofs for Section~\ref{sec:vertex-memory-main-insert}}

	\begin{proof}[Proof of Proposition~\ref{prop:memory-geometric-decay}]
		Write \(\lambda_{ij}=1-\alpha_{ij}-\beta_{ij}\).  After the isolated anomaly, iterating the baseline recursion
		gives, for $i\ne j$ and $t>t^*$,
		\[
			\Pi_{ij}^{(t-1)}-\Pi_{ij}^{(t^*-1)}
			=
			\lambda_{ij}^{t-t^*-1}
			\left(
			\Pi_{ij}^{(t^*)}-\Pi_{ij}^{(t^*-1)}
			\right).
		\]
		Therefore,
		\[
			\{M_{Y,i}^{(t^*-1,t)}\}^2
			=
			\frac1p\sum_{j\ne i}
			\frac{\alpha_{ji}^2}{\rho_Y}
			|\lambda_{ji}|^{2(t-t^*-1)}
			\left(
			\Pi_{ji}^{(t^*)}-\Pi_{ji}^{(t^*-1)}
			\right)^2
			\le
			\rho_Y\bar r^{2(t-t^*-1)}(\Delta_i^{\rm mean})^2.
		\]
		Taking square roots gives the formation bound.  The dissolution bound
		is identical.  	If \(i,j\notin\mathcal I^{(t^*)}\), the mechanism of dyad \((i,j)\)
		does not change at \(t^*\), so baseline stationarity gives
		\(\Pi_{ij}^{(t^*)}-\Pi_{ij}^{(t^*-1)}=0\). Thus, for every
		\(i\notin\mathcal I^{(t^*)}\), at most \(m^{(t^*)}\) summands in
		\(\Delta_i^{\rm mean}\) are nonzero, and
		\(
		\Delta_i^{\rm mean}\leq\sqrt{\frac{m^{(t^*)}}{p}}.
		\)
		Substitution in the channelwise geometric bound gives
		\[
		M_{h,i}^{(t^*-1,t)}
		\leq
		\left(\frac{\rho_hm^{(t^*)}}{p}\right)^{1/2}
		\bar r^{\,t-t^*-1},\qquad h\in\mathcal H.
		\]
		Finally, if
		\(\rho_hm^{(t^*)}=o(\log p)\), then
		\(\bar r^{\,t-t^*-1}\leq1\), and
		\(\epsilon_p\asymp\sqrt{\log p/p}\) imply
		\(M_{h,i}^{(t^*-1,t)}=o(\epsilon_p)\) uniformly over every non-anomalous
		vertex and every \(t>t^*\).
	\end{proof}

	\begin{proof}[Proof of Theorem~\ref{thm:memory-echo-main}]

    Under the isolated-anomaly data-generating process, retain the
		already defined notation \(\mathbf S_h\) and \(\mathbf R_{S_h}\) for the
		population unfolded matrix and its right embedding. For use only inside this proof, define
		\[
		D_{h,i}^{(t^*-1,t)}
		:=
		\left\|
		\mathbf e_i^\top\mathbf R_{S_h}^{(t)}
		-
		\mathbf e_i^\top\mathbf R_{S_h}^{(t^*-1)}
		\right\|_2,
		\qquad h\in\mathcal H.
		\]

		We first establish deterministic population bounds for this distance. 
		Multiplying the SVD identity on the left by
		\(\boldsymbol\Sigma_Y^{-1/2}\mathbf U_Y^\top\) gives
		\(
		\mathbf R_{S_Y}^\top
		=
		\boldsymbol\Sigma_Y^{-1/2}\mathbf U_Y^\top\mathbf S_Y.
		\)
		Taking the difference between the \(t\) and \(t^*-1\)
		blocks and then multiplying by \(\mathbf e_i\) therefore yields the exact
		identity
		\[
		D_{Y,i}^{(t^*-1,t)}
		=
		\left\|
		\boldsymbol\Sigma_Y^{-1/2}\mathbf U_Y^\top
		\left\{
		\mathbf S_Y^{(t)}-\mathbf S_Y^{(t^*-1)}
		\right\}\mathbf e_i
		\right\|_2.
		\]

		For \(j\ne i\), the returned-baseline formation blocks
		satisfy
		\((\mathbf S_Y^{(s)})_{ji}=\alpha_{ji}(1-\Pi_{ji}^{(s-1)})\).
		Consequently,
		\[
		\mathbf e_j^\top
		\left\{
		\mathbf S_Y^{(t)}-\mathbf S_Y^{(t^*-1)}
		\right\}\mathbf e_i
		=
		-\alpha_{ji}
		\left(\Pi_{ji}^{(t-1)}-\Pi_{ji}^{(t^*-1)}\right).
		\]
		By the definition of \(M_{Y,i}^{(t^*-1,t)}\), the off-diagonal
		part of this column has Euclidean norm
		\(\sqrt{p\rho_Y}\,M_{Y,i}^{(t^*-1,t)}\). The completed diagonal changes
		its norm by at most \(C\rho_Y\). Hence
		\[
	\sqrt{p\rho_Y}\,M_{Y,i}^{(t^*-1,t)}
		\le
		\left\|
		\left\{\mathbf S_Y^{(t)}-\mathbf S_Y^{(t^*-1)}\right\}
		\mathbf e_i
		\right\|_2
		\le
		\sqrt{p\rho_Y}\,M_{Y,i}^{(t^*-1,t)}+C\rho_Y.
		\]

		For the upper bound, the first identity and
		\(\sigma_{K_Y}(\mathbf S_Y)\ge c_{Y,S}p\rho_Y\) from
		Condition~\ref{ass:3} imply
		\[
		\begin{aligned}
		D_{Y,i}^{(t^*-1,t)}
		\le
		\frac{
		\left\|\left\{\mathbf S_Y^{(t)}-
		\mathbf S_Y^{(t^*-1)}\right\}\mathbf e_i\right\|_2
		}{\sqrt{\sigma_{K_Y}(\mathbf S_Y)}}
		\le
		CM_{Y,i}^{(t^*-1,t)}
		+C\sqrt{\frac{\rho_Y}{p}}.
		\end{aligned}
		\]
		For the lower bound, every column of each block difference
		belongs to the column space of \(\mathbf U_Y\), because
		\[
		\mathbf U_Y\mathbf U_Y^\top
		\left\{\mathbf S_Y^{(t)}-\mathbf S_Y^{(t^*-1)}\right\}
		\mathbf e_i
		=
		\left\{\mathbf S_Y^{(t)}-\mathbf S_Y^{(t^*-1)}\right\}
		\mathbf e_i.
		\]
		Thus \(\mathbf U_Y^\top\) preserves the norm of this vector.
		Using \(\sigma_1(\mathbf S_Y)\le C_{Y,S}p\rho_Y\) and the lower column
		bound above gives
		\[
		\begin{aligned}
		D_{Y,i}^{(t^*-1,t)}
		&\ge
		\frac{
		\left\|\mathbf U_Y^\top
		\left\{\mathbf S_Y^{(t)}-\mathbf S_Y^{(t^*-1)}\right\}
		\mathbf e_i\right\|_2
		}{\sqrt{\sigma_1(\mathbf S_Y)}}\\
		&=
		\frac{
		\left\|\left\{\mathbf S_Y^{(t)}-
		\mathbf S_Y^{(t^*-1)}\right\}\mathbf e_i\right\|_2
		}{\sqrt{\sigma_1(\mathbf S_Y)}}
		\ge
		cM_{Y,i}^{(t^*-1,t)}.
		\end{aligned}
		\]
		The same calculation applies to the \(Z\) channel. Since
		\(\mathcal H=\{Y,Z\}\) is finite, the channelwise constants may be replaced
		by common constants \(0<c<C<\infty\). We have therefore shown
		\[
		cM_{h,i}^{(t^*-1,t)}
		\le
		D_{h,i}^{(t^*-1,t)}
		\le
		CM_{h,i}^{(t^*-1,t)}
		+C\sqrt{\frac{\rho_h}{p}},
		\qquad h\in\mathcal H.
		\]

		It remains to transfer these deterministic bounds to the
		observed statistic. On \(\mathcal E_{\rm UASE}(c_0)\), let
		\(\mathbf W_{hS}\) be the common alignment matrix from
		Theorem~\ref{thm1}. Orthogonality gives
		\[
		D_{h,i}^{(t^*-1,t)}
		=
		\left\|
		\mathbf e_i^\top
		\left\{\mathbf R_{S_h}^{(t)}-\mathbf R_{S_h}^{(t^*-1)}\right\}
		\mathbf W_{hS}
		\right\|_2.
		\]
		The reverse triangle inequality, followed by the triangle
		inequality and the two rowwise bounds from Theorem~\ref{thm1}, yields
		\[
		\begin{aligned}
		\left|T_{h,i}^{(t^*-1,t)}-D_{h,i}^{(t^*-1,t)}\right|
		&\le
		\left\|
		\mathbf e_i^\top
		\left[
		\left\{\mathbf R_h^{(t)}-\mathbf R_{S_h}^{(t)}\mathbf W_{hS}\right\}
		-
		\left\{\mathbf R_h^{(t^*-1)}-
		\mathbf R_{S_h}^{(t^*-1)}\mathbf W_{hS}\right\}
		\right]
		\right\|_2\\
		&\le
		\left\|
		\mathbf e_i^\top
		\left\{\mathbf R_h^{(t)}-\mathbf R_{S_h}^{(t)}\mathbf W_{hS}\right\}
		\right\|_2
		+
		\left\|
		\mathbf e_i^\top
		\left\{\mathbf R_h^{(t^*-1)}-
		\mathbf R_{S_h}^{(t^*-1)}\mathbf W_{hS}\right\}
		\right\|_2\\
		&\le
		2\epsilon_p.
		\end{aligned}
		\]
		Combining this approximation with the population upper and
		lower bounds, and using \(T_{h,i}^{(t^*-1,t)}\ge0\), gives on the same event
		\[
		\max\left\{cM_{h,i}^{(t^*-1,t)}-2\epsilon_p,0\right\}
		\le
	T_{h,i}^{(t^*-1,t)}
		\le
		CM_{h,i}^{(t^*-1,t)}
		+C\sqrt{\frac{\rho_h}{p}}+2\epsilon_p.
		\]
		The event is common to all \(h\), \(i\), and time blocks, and
		Theorem~\ref{thm1} gives
		\(\mathbb P\{\mathcal E_{\rm UASE}(c_0)^c\}
		\le C_{c_0}p^{-c_0}\). This proves the asserted simultaneous sample bound.
		Finally, if the right-hand side of \eqref{eq:memory-decay} is
		\(O(\epsilon_p)\), then \(M_{h,i}^{(t^*-1,t)}=O(\epsilon_p)\) and
		\(\sqrt{\rho_h/p}=O(\epsilon_p)\). The sample upper bound consequently
		implies \(T_{h,i}^{(t^*-1,t)}=O_{\mathbb P}(\epsilon_p)\).

	\end{proof}

		\begin{proof}[Proof of Corollary~\ref{cor:vertex-memory-burnin-main}]
		Because \(\log(1/\bar r)>0\), the assumed lower bound on
		\(t-t^*-1\) implies
		\[
		\bar r^{\,t-t^*-1}
		\leq
		\left[
		\max\left\{1,
		\frac{\sqrt{\max\{\rho_Y,\rho_Z\}}\,\Delta_i^{\rm mean}}
		{\epsilon_p}\right\}
		\right]^{-1}.
		\]
		Proposition~\ref{prop:memory-geometric-decay} therefore gives
		\[
		\max_{h\in\mathcal H}M_{h,i}^{(t^*-1,t)}
		\leq
		\sqrt{\max\{\rho_Y,\rho_Z\}}\,
		\bar r^{\,t-t^*-1}\Delta_i^{\rm mean}
		\leq\epsilon_p.
		\]
		Thus the right-hand side of \eqref{eq:memory-decay} is
		\(O(\epsilon_p)\), and Theorem~\ref{thm:memory-echo-main} directly yields
		\(\max_{h\in\mathcal H}T_{h,i}^{(t^*-1,t)}
		=O_{\mathbb P}(\epsilon_p)\), completing the proof.
	\end{proof}
	\subsubsection*{S.2.4 Proof of Theorem~\ref{thm:2g}}
	
		\begin{proposition}\label{prop:g0}
		Define the network-level AR(1) memory functionals $M_{Y,G}^{(t^*-1,t)}$ and
		$M_{Z,G}^{(t^*-1,t)}$ as in the main paper.
		Assume  Condition~\ref{ass:graph-visibility} holds. Then
		\[
		\frac{1}{p}\left\|\mathbf S_Y^{(t)}-\mathbf S_Y^{(t^*-1)}\right\|_F
		\;\in\;
		\left[\rho_Y\mu_Y\Delta_Y^{(t^*-1,t)}-\sqrt{\rho_Y}M_{Y,G}^{(t^*-1,t)}
		-\frac{\rho_Y}{\sqrt p},\;
		2\rho_Y\sqrt{d}\,\ell\,\Delta_Y^{(t^*-1,t)}+\sqrt{\rho_Y}M_{Y,G}^{(t^*-1,t)}
		+\frac{\rho_Y}{\sqrt p}\right],
		\]
		and analogously for $Z$.
	\end{proposition}
	
	\begin{proof}
		We first prove the bounds for \(\mathbf S_Y^{(t)}\). By definition,
		\[
		\begin{aligned}
			\mathbf S_Y^{(t)}-\mathbf S_Y^{(t^*-1)}
			&=
			\left\{
			\mathbf P_Y^{(t)}
			-
			\mathbf P_Y^{(t^*-1)}
			\right\}
			\circ
			\left\{
			\mathbf 1\mathbf 1^\top-\mathbf\Pi^{(t^*-1)}
			\right\}        \\
			&\qquad
			-
			\mathbf P_Y^{(t)}
			\circ
			\left\{
			\mathbf\Pi^{(t-1)}
			-
			\mathbf\Pi^{(t^*-1)}
			\right\}.
		\end{aligned}
		\]
		Hence, by the reverse triangle inequality,
		\[
		\begin{aligned}
			\frac{1}{p}
			\left\|
			\mathbf S_Y^{(t)}-\mathbf S_Y^{(t^*-1)}
			\right\|_F
			&\geq
			A_{Y,G}^{(t^*-1,t)}
			-
			B_{Y,G}^{(t^*-1,t)},
		\end{aligned}
		\]
		where
		\[
		\begin{aligned}
			A_{Y,G}^{(t^*-1,t)}
			&:=
			\frac{1}{p}
			\left\|
			\left\{
			\mathbf P_Y^{(t)}
			-
			\mathbf P_Y^{(t^*-1)}
			\right\}
			\circ
			\left\{
			\mathbf 1\mathbf 1^\top-\mathbf\Pi^{(t^*-1)}
			\right\}
			\right\|_F,        \\
			B_{Y,G}^{(t^*-1,t)}
			&:=
			\frac{1}{p}
			\left\|
			\mathbf P_Y^{(t)}
			\circ
			\left\{
			\mathbf\Pi^{(t-1)}
			-
			\mathbf\Pi^{(t^*-1)}
			\right\}
			\right\|_F .
		\end{aligned}
		\]
		
		For the direct signal term, since
		$\mathbf P_Y^{(t)}
		=
		\rho_Y\mathbf L_Y^{(t)}\mathbf L_Y^{(t)\top},$
		we have
		\[
		\begin{aligned}
			A_{Y,G}^{(t^*-1,t)}
			&=
			\rho_Y
			\frac{1}{p}
			\left\|
			\left[
			\left\{
			\mathbf L_Y^{(t)}\mathbf L_Y^{(t)\top}
			-
			\mathbf L_Y^{(t^*-1)}\mathbf L_Y^{(t^*-1)\top}
			\right\}
			\circ
			\left\{
			\mathbf 1\mathbf 1^\top-\mathbf\Pi^{(t^*-1)}
			\right\}
			\right]
			\right\|_F        \\
			&\geq
			\rho_Y\mu_Y\Delta_Y^{(t^*-1,t)}
		\end{aligned}
		\]
		by Condition~\ref{ass:graph-visibility}.
		
		For the memory term, separate the completed diagonal from the true
		off-diagonal probabilities.  Since
		\(0\leq\alpha_{ii}^{(t)}\leq\rho_Y\) and
		\(0\leq\Pi_{ii}^{(t)}\leq1\), the triangle inequality gives
		\[
		\begin{aligned}
			B_{Y,G}^{(t^*-1,t)}
			&\leq
			\frac{1}{p}
			\left[
			\sum_{i\neq j}
			\left\{
			\alpha_{ij}^{(t)}
			\left(
			\Pi_{ij}^{(t-1)}
			-
			\Pi_{ij}^{(t^*-1)}
			\right)
			\right\}^2
			\right]^{1/2}        \\
			&\qquad
			+
			\frac{1}{p}
			\left[
			\sum_i
			\left\{
			\alpha_{ii}^{(t)}
			\left(
			\Pi_{ii}^{(t-1)}
			-
			\Pi_{ii}^{(t^*-1)}
			\right)
			\right\}^2
			\right]^{1/2}        \\
			&=
			\sqrt{\rho_Y}
			\left[
			\frac{1}{p^2}
			\sum_{i\neq j}
			\left\{
			\frac{\alpha_{ij}^{(t)}}{\sqrt{\rho_Y}}
			\left(
			\Pi_{ij}^{(t-1)}
			-
			\Pi_{ij}^{(t^*-1)}
			\right)
			\right\}^2
			\right]^{1/2}
			+
			\frac{1}{p}
			\left[
			\sum_i
			\left\{
			\alpha_{ii}^{(t)}
			\left(
			\Pi_{ii}^{(t-1)}
			-
			\Pi_{ii}^{(t^*-1)}
			\right)
			\right\}^2
			\right]^{1/2}        \\
			&\leq
			\sqrt{\rho_Y}\,
			M_{Y,G}^{(t^*-1,t)}
			+
			\frac{\rho_Y}{\sqrt p} .
		\end{aligned}
		\]
		Combining the two estimates gives
		\[
		\begin{aligned}
			\frac{1}{p}
			\left\|
			\mathbf S_Y^{(t)}-\mathbf S_Y^{(t^*-1)}
			\right\|_F
			&\geq
			\rho_Y\mu_Y\Delta_Y^{(t^*-1,t)}
			- 
			\sqrt{\rho_Y}\,
			M_{Y,G}^{(t^*-1,t)}
			-
			\frac{\rho_Y}{\sqrt p} .
		\end{aligned}
		\]
		This establishes the lower bound for \(\mathbf S_Y^{(t)}\).
		
		For the upper bound, the triangle inequality gives
		\[
		\begin{aligned}
			\frac{1}{p}
			\left\|
			\mathbf S_Y^{(t)}-\mathbf S_Y^{(t^*-1)}
			\right\|_F
			&\leq
			A_{Y,G}^{(t^*-1,t)}
			+
			B_{Y,G}^{(t^*-1,t)} .
		\end{aligned}
		\]
		It remains to upper bound \(A_{Y,G}^{(t^*-1,t)}\).  Write
		$\Delta\mathbf L_Y^{(t^*-1,t)}
		:=
		\mathbf L_Y^{(t)}
		-
		\mathbf L_Y^{(t^*-1)}.$
		Then
		\[
		\begin{aligned}
			\mathbf L_Y^{(t)}\mathbf L_Y^{(t)\top}
			-
			\mathbf L_Y^{(t^*-1)}\mathbf L_Y^{(t^*-1)\top}
			&=
			\Delta\mathbf L_Y^{(t^*-1,t)}
			\mathbf L_Y^{(t)\top}
			+
			\mathbf L_Y^{(t^*-1)}
			\Delta\mathbf L_Y^{(t^*-1,t)\top}.
		\end{aligned}
		\]
		Since Hadamard multiplication by
		\(\mathbf 1\mathbf 1^\top-\mathbf\Pi^{(t^*-1)}\) cannot increase the Frobenius norm,
		\[
		\begin{aligned}
			A_{Y,G}^{(t^*-1,t)}
			&\leq
			\frac{\rho_Y}{p}
			\left\|
			\mathbf L_Y^{(t)}\mathbf L_Y^{(t)\top}
			-
			\mathbf L_Y^{(t^*-1)}\mathbf L_Y^{(t^*-1)\top}
			\right\|_F        \\
			&\leq
			\frac{\rho_Y}{p}
			\left\{
			\left\|
			\Delta\mathbf L_Y^{(t^*-1,t)}
			\mathbf L_Y^{(t)\top}
			\right\|_F
			+
			\left\|
			\mathbf L_Y^{(t^*-1)}
			\Delta\mathbf L_Y^{(t^*-1,t)\top}
			\right\|_F
			\right\}        \\
			&\leq
			\frac{\rho_Y}{p}
			\left\{
			\left\|
			\Delta\mathbf L_Y^{(t^*-1,t)}
			\right\|_F
			\left\|
			\mathbf L_Y^{(t)}
			\right\|_F
			+
			\left\|
			\mathbf L_Y^{(t^*-1)}
			\right\|_F
			\left\|
			\Delta\mathbf L_Y^{(t^*-1,t)}
			\right\|_F
			\right\}.
		\end{aligned}
		\]
		The uniform latent-position bound gives
		$\left\|\mathbf L_Y^{(t)}\right\|_F
		\leq
		\sqrt p\,\sqrt d\,\ell,
		\left\|\mathbf L_Y^{(t^*-1)}\right\|_F
		\leq
		\sqrt p\,\sqrt d\,\ell,$
		and, by definition, 
		$ \left\|
		\Delta\mathbf L_Y^{(t^*-1,t)}
		\right\|_F
		=
		\sqrt p\,\Delta_Y^{(t^*-1,t)}.$
		Consequently,
		\[
		A_{Y,G}^{(t^*-1,t)}
		\leq
		2\rho_Y\sqrt d\,\ell\,\Delta_Y^{(t^*-1,t)}.
		\]
		Together with
		\[
		B_{Y,G}^{(t^*-1,t)}
		\leq
		\sqrt{\rho_Y}\,
		M_{Y,G}^{(t^*-1,t)}
		+
		\frac{\rho_Y}{\sqrt p},
		\]
		we obtain
		\[
		\begin{aligned}
			\frac{1}{p}
			\left\|
			\mathbf S_Y^{(t)}-\mathbf S_Y^{(t^*-1)}
			\right\|_F
			&\leq
			2\rho_Y\sqrt d\,\ell\,\Delta_Y^{(t^*-1,t)}
			+
			\sqrt{\rho_Y}\,
			M_{Y,G}^{(t^*-1,t)}
			+
			\frac{\rho_Y}{\sqrt p} .
		\end{aligned}
		\]
		
		The bounds for \(\mathbf S_Z^{(t)}\) are proved analogously. Indeed,
		\[
		\begin{aligned}
			\mathbf S_Z^{(t)}-\mathbf S_Z^{(t^*-1)}
			&=
			\left\{
			\mathbf P_Z^{(t)}
			-
			\mathbf P_Z^{(t^*-1)}
			\right\}
			\circ
			\mathbf\Pi^{(t^*-1)}
			+
			\mathbf P_Z^{(t)}
			\circ
			\left\{
			\mathbf\Pi^{(t-1)}
			-
			\mathbf\Pi^{(t^*-1)}
			\right\}.
		\end{aligned}
		\]
		Condition~\ref{ass:graph-visibility} and the same Frobenius-norm argument give
		\[
		\begin{aligned}
			\frac{1}{p}
			\left\|
			\mathbf S_Z^{(t)}-\mathbf S_Z^{(t^*-1)}
			\right\|_F
			&\geq
			\rho_Z\mu_Z\Delta_Z^{(t^*-1,t)}
			-
			\sqrt{\rho_Z}\,
			M_{Z,G}^{(t^*-1,t)}
			-
			\frac{\rho_Z}{\sqrt p},
		\end{aligned}
		\]
		and
		\[
		\begin{aligned}
			\frac{1}{p}
			\left\|
			\mathbf S_Z^{(t)}-\mathbf S_Z^{(t^*-1)}
			\right\|_F
			&\leq
			2\rho_Z\sqrt d\,\ell\,\Delta_Z^{(t^*-1,t)}
			+
			\sqrt{\rho_Z}\,
			M_{Z,G}^{(t^*-1,t)}
			+
			\frac{\rho_Z}{\sqrt p} .
		\end{aligned}
		\]
		The two channelwise bounds prove the proposition.
	\end{proof}

	\begin{proposition}
		\label{prop:graph-S-to-R}
		It holds that
		\[
		\frac{1}{\sqrt p}
		\left\|
		\mathbf R_{\mathbf S_Y}^{(t)}
		-
		\mathbf R_{\mathbf S_Y}^{(t^*-1)}
		\right\|_F
		\geq
		\frac{1}{\sqrt{C_{Y,S}\rho_Y}}
		\frac{1}{p}
		\left\|
		\mathbf S_Y^{(t)}
		-
		\mathbf S_Y^{(t^*-1)}
		\right\|_F .
		\]
		Moreover,
		\[
		\frac{1}{\sqrt p}
		\left\|
		\mathbf R_{\mathbf S_Y}^{(t)}
		-
		\mathbf R_{\mathbf S_Y}^{(t^*-1)}
		\right\|_F
		\leq
		\frac{1}{\sqrt{c_{Y,S}\rho_Y}}
		\frac{1}{p}
		\left\|
		\mathbf S_Y^{(t)}
		-
		\mathbf S_Y^{(t^*-1)}
		\right\|_F .
		\]
		The same statement holds for \(\mathbf S_Z\), with
		\(\rho_Y,c_{Y,S},C_{Y,S}\) replaced by
		\(\rho_Z,c_{Z,S},C_{Z,S}\).
	\end{proposition}
	
	\begin{proof}
		We prove the statement for \(Y\). The proof for \(Z\) is identical.
		
		By the population SVD, for each block \(t\),
		\[
		\mathbf S_Y^{(t)}
		=
		\mathbf U_{S_Y}
		\mathbf \Sigma_{S_Y}^{1/2}
		\left(
		\mathbf R_{\mathbf S_Y}^{(t)}
		\right)^\top .
		\]
		Hence
		\[
		\mathbf S_Y^{(t)}
		-
		\mathbf S_Y^{(t^*-1)}
		=
		\mathbf U_{S_Y}
		\mathbf \Sigma_{S_Y}^{1/2}
		\left(
		\mathbf R_{\mathbf S_Y}^{(t)}
		-
		\mathbf R_{\mathbf S_Y}^{(t^*-1)}
		\right)^\top .
		\]
		Since \(\mathbf U_{S_Y}\) has orthonormal columns,
		\[
		\begin{aligned}
			\left\|
			\mathbf S_Y^{(t)}
			-
			\mathbf S_Y^{(t^*-1)}
			\right\|_F
			&=
			\left\|
			\mathbf \Sigma_{S_Y}^{1/2}
			\left(
			\mathbf R_{\mathbf S_Y}^{(t)}
			-
			\mathbf R_{\mathbf S_Y}^{(t^*-1)}
			\right)^\top
			\right\|_F .
		\end{aligned}
		\]
		Therefore,
		\[
		\sqrt{\sigma_{K_Y}(\mathbf S_Y)}
		\left\|
		\mathbf R_{\mathbf S_Y}^{(t)}
		-
		\mathbf R_{\mathbf S_Y}^{(t^*-1)}
		\right\|_F
		\leq
		\left\|
		\mathbf S_Y^{(t)}
		-
		\mathbf S_Y^{(t^*-1)}
		\right\|_F
		\]
		and
		\[
		\left\|
		\mathbf S_Y^{(t)}
		-
		\mathbf S_Y^{(t^*-1)}
		\right\|_F
		\leq
		\sqrt{\sigma_1(\mathbf S_Y)}
		\left\|
		\mathbf R_{\mathbf S_Y}^{(t)}
		-
		\mathbf R_{\mathbf S_Y}^{(t^*-1)}
		\right\|_F .
		\]
		By Condition~\ref{ass:3}, we have
		\[
		\left\|
		\mathbf R_{\mathbf S_Y}^{(t)}
		-
		\mathbf R_{\mathbf S_Y}^{(t^*-1)}
		\right\|_F
		\geq
		\frac{1}{\sqrt{C_{Y,S}p\rho_Y}}
		\left\|
		\mathbf S_Y^{(t)}
		-
		\mathbf S_Y^{(t^*-1)}
		\right\|_F ,
		\]
		and
		\[
		\left\|
		\mathbf R_{\mathbf S_Y}^{(t)}
		-
		\mathbf R_{\mathbf S_Y}^{(t^*-1)}
		\right\|_F
		\leq
		\frac{1}{\sqrt{c_{Y,S}p\rho_Y}}
		\left\|
		\mathbf S_Y^{(t)}
		-
		\mathbf S_Y^{(t^*-1)}
		\right\|_F .
		\]
		Dividing both inequalities by \(\sqrt p\) gives
		\[
		\frac{1}{\sqrt p}
		\left\|
		\mathbf R_{\mathbf S_Y}^{(t)}
		-
		\mathbf R_{\mathbf S_Y}^{(t^*-1)}
		\right\|_F
		\geq
		\frac{1}{\sqrt{C_{Y,S}\rho_Y}}
		\frac{1}{p}
		\left\|
		\mathbf S_Y^{(t)}
		-
		\mathbf S_Y^{(t^*-1)}
		\right\|_F ,
		\]
		and
		\[
		\frac{1}{\sqrt p}
		\left\|
		\mathbf R_{\mathbf S_Y}^{(t)}
		-
		\mathbf R_{\mathbf S_Y}^{(t^*-1)}
		\right\|_F
		\leq
		\frac{1}{\sqrt{c_{Y,S}\rho_Y}}
		\frac{1}{p}
		\left\|
		\mathbf S_Y^{(t)}
		-
		\mathbf S_Y^{(t^*-1)}
		\right\|_F .
		\]
		The proposition follows.
	\end{proof}
	
		\begin{proof}[Proof of Theorem~\ref{thm:2g}]
		By Theorem~\ref{thm1},
		\(\mathbb P\{\mathcal E_{\rm UASE}(c_0)^c\}\le C_{c_0}p^{-c_0}\), and on
		\(\mathcal E_{\rm UASE}(c_0)\) the rowwise embedding bounds hold in both
		channels for every block used below.
		
		(1) Under \(H_0^{(t^*-1,t)}\), the aggregate visible transition-probability
		displacement vanishes; the remaining population displacement is bounded by
		the network-level memory terms.
		By Proposition~\ref{prop:g0} and \ref{prop:graph-S-to-R},
		\[
		\begin{aligned}
			\frac{1}{\sqrt p}
			\left\|
			\mathbf R_{\mathbf S_Y}^{(t)}
			-
			\mathbf R_{\mathbf S_Y}^{(t^*-1)}
			\right\|_F
			&\le
			\frac{1}{\sqrt{c_{Y,S}}}
			\left(
			M_{Y,G}^{(t^*-1,t)}
			+
			\sqrt{\frac{\rho_Y}{p}}
			\right),\\
			\frac{1}{\sqrt p}
			\left\|
			\mathbf R_{\mathbf S_Z}^{(t)}
			-
			\mathbf R_{\mathbf S_Z}^{(t^*-1)}
			\right\|_F
			&\le
			\frac{1}{\sqrt{c_{Z,S}}}
			\left(
			M_{Z,G}^{(t^*-1,t)}
			+
			\sqrt{\frac{\rho_Z}{p}}
			\right).
		\end{aligned}
		\]
		Hence the network-level null condition implies
		\[
		\max
		\left\{
		\frac{1}{\sqrt p}
		\left\|
		\mathbf R_{\mathbf S_Y}^{(t)}
		-
		\mathbf R_{\mathbf S_Y}^{(t^*-1)}
		\right\|_F,
		\frac{1}{\sqrt p}
		\left\|
		\mathbf R_{\mathbf S_Z}^{(t)}
		-
		\mathbf R_{\mathbf S_Z}^{(t^*-1)}
		\right\|_F
		\right\}
		<
		(\gamma-2)\epsilon_p .
		\]
		On \(\mathcal E_{\rm UASE}(c_0)\), the embedding error contributes at most \(2\epsilon_p\) after
		normalisation by \(p^{-1/2}\). Therefore
		\[
		T_G^{(t^*-1,t)}
		<
		(\gamma-2)\epsilon_p+2\epsilon_p
		=
		\gamma\epsilon_p
		=
		\tau_p.
		\]
		Thus \(\phi^{(t^*-1,t)}=0\) on \(\mathcal E_{\rm UASE}(c_0)\). Hence
		\[
		\mathbb P_{H_0^{(t^*-1,t)}}
		\left(
		\phi^{(t^*-1,t)}=1
		\right)
		\le
		\mathbb P\{\mathcal E_{\rm UASE}(c_0)^c\}
		\le
		C_{c_0}p^{-c_0}.
		\]
		
		(2) Under \(H_1^{(t^*-1,t)}\), Proposition~\ref{prop:g0} and \ref{prop:graph-S-to-R} give
		\[
		\begin{aligned}
			\frac{1}{\sqrt p}
			\left\|
			\mathbf R_{\mathbf S_Y}^{(t)}
			-
			\mathbf R_{\mathbf S_Y}^{(t^*-1)}
			\right\|_F
			&\ge
			\frac{1}{\sqrt{C_{Y,S}}}
			\left[
			\sqrt{\rho_Y}\mu_Y\Delta_Y^{(t^*-1,t)}
			-
			M_{Y,G}^{(t^*-1,t)}
			-
			\sqrt{\frac{\rho_Y}{p}}
			\right],\\
			\frac{1}{\sqrt p}
			\left\|
			\mathbf R_{\mathbf S_Z}^{(t)}
			-
			\mathbf R_{\mathbf S_Z}^{(t^*-1)}
			\right\|_F
			&\ge
			\frac{1}{\sqrt{C_{Z,S}}}
			\left[
			\sqrt{\rho_Z}\mu_Z\Delta_Z^{(t^*-1,t)}
			-
			M_{Z,G}^{(t^*-1,t)}
			-
			\sqrt{\frac{\rho_Z}{p}}
			\right].
		\end{aligned}
		\]
		Hence the separation condition implies
		\[
		\max
		\left\{
		\frac{1}{\sqrt p}
		\left\|
		\mathbf R_{\mathbf S_Y}^{(t)}
		-
		\mathbf R_{\mathbf S_Y}^{(t^*-1)}
		\right\|_F,
		\frac{1}{\sqrt p}
		\left\|
		\mathbf R_{\mathbf S_Z}^{(t)}
		-
		\mathbf R_{\mathbf S_Z}^{(t^*-1)}
		\right\|_F
		\right\}
		>
		(\gamma+2)\epsilon_p .
		\]
		On \(\mathcal E_{\rm UASE}(c_0)\), the embedding error contributes at most \(2\epsilon_p\), so
		\[
		T_G^{(t^*-1,t)}
		>
		(\gamma+2)\epsilon_p-2\epsilon_p
		=
		\gamma\epsilon_p
		=
		\tau_p.
		\]
		Thus \(\phi^{(t^*-1,t)}=1\) on \(\mathcal E_{\rm UASE}(c_0)\). Therefore
		\[
		\mathbb P_{H_1^{(t^*-1,t)}}
		\left(
		\phi^{(t^*-1,t)}=1
		\right)
		\ge
		\mathbb P\{\mathcal E_{\rm UASE}(c_0)\}
		\ge
		1-C_{c_0}p^{-c_0}.
		\]
        Applying the same argument separately to
		\(T_{Y,G}^{(t^*-1,t)}\) and \(T_{Z,G}^{(t^*-1,t)}\) on the common UASE
		event gives the channelwise conclusions simultaneously, proving the
		stated formation--dissolution attribution result.
	\end{proof}
		\subsection*{S.3 Supporting results and proofs for Section 4}
	
	Appendix~S.3.1 proves Proposition~\ref{prop:eta}.
Appendix~S.3.2 proves
Theorem~\ref{thm:dcsbm_community_exact_recovery} and
Corollary~\ref{cor:pooled_community_exact_recovery}.
Appendix S.3.3 proves the community-reassignment recovery result, Appendix S.3.4 proves Theorem 6 for the three community-level tests, and Appendix S.3.5 gives an additional corollary linking \(\widetilde\eta_{S,t}\) to autoregressive memory.
	
	The singular-value bounds used below are those in
Condition~\ref{ass:3}, while the population leverage bounds follow
from Proposition~\ref{prop:uv2}. As in Appendix~S.1, all
\(O_{\mathbb P}\) statements are interpreted in the polynomial
high-probability sense. The community-recovery and testing arguments
below are deterministic on the corresponding embedding events.
	
	\subsubsection*{S.3.1 Proof of Proposition \ref{prop:eta}}

\begin{lemma}\label{lem:eta-channel-reduction}
		In the AR(1)-DCSBM model, suppose Condition~\ref{ass:3} holds and
		\(m_{\min,S,t}>0\). Then
		\begin{equation}\label{eq:eta_memory_channel_reduction}
			\widetilde\eta_{S,t}
			\le
			\frac{C}{m_{\min,S,t}}
			\max_{a\in[q_t]}
			\max_{i,j\in C_a^{(t)}}
			\left[
			\frac{
				\left\|
				\mathbf S_Y^{(t)}
				\left(
				\frac{\mathbf e_i}{\psi_{Y,i}^{(t)}}
				-
				\frac{\mathbf e_j}{\psi_{Y,j}^{(t)}}
				\right)
				\right\|_2^2
			}{p\rho_Yc_{Y,S}}
			+
			\frac{
				\left\|
				\mathbf S_Z^{(t)}
				\left(
				\frac{\mathbf e_i}{\psi_{Z,i}^{(t)}}
				-
				\frac{\mathbf e_j}{\psi_{Z,j}^{(t)}}
				\right)
				\right\|_2^2
			}{p\rho_Zc_{Z,S}}
			\right]^{1/2}.
		\end{equation}
	\end{lemma}
	\begin{proof}
		Fix \(a\in[q_t]\), \(i,j\in C_a^{(t)}\), and \(h\in\mathcal H\).
		The population SVD identity and Condition~\ref{ass:3} give
		\[
		\left\|
		\frac{\mathbf e_i^\top\mathbf R_{S_h}^{(t)}}
		{\psi_{h,i}^{(t)}}
		-
		\frac{\mathbf e_j^\top\mathbf R_{S_h}^{(t)}}
		{\psi_{h,j}^{(t)}}
		\right\|_2
		\le
		\frac{
		\left\|
		\mathbf S_h^{(t)}
		\left(
		\frac{\mathbf e_i}{\psi_{h,i}^{(t)}}
		-
		\frac{\mathbf e_j}{\psi_{h,j}^{(t)}}
		\right)
		\right\|_2
		}{\sqrt{p\rho_hc_{h,S}}}.
		\]
		Because \(\psi_{h,r}^{(t)}>0\), division by
		\(\psi_{h,r}^{(t)}\) does not change the normalised row. Moreover,
		\[
		\min_{r\in\{i,j\}}
		\left\|
		\frac{\mathbf e_r^\top\mathbf R_{S_h}^{(t)}}
		{\psi_{h,r}^{(t)}}
		\right\|_2
		\ge
		\frac{m_{h,S,t}}{\ell}.
		\]
		Applying
		\[
		\left\|
		\frac{\boldsymbol x}{\|\boldsymbol x\|_2}
		-
		\frac{\boldsymbol y}{\|\boldsymbol y\|_2}
		\right\|_2
		\le
		\frac{2\|\boldsymbol x-\boldsymbol y\|_2}
		{\min\{\|\boldsymbol x\|_2,\|\boldsymbol y\|_2\}}
		\]
		with
		\(
		\boldsymbol x=\mathbf e_i^\top\mathbf R_{S_h}^{(t)}/\psi_{h,i}^{(t)}
		\)
		and
		\(
		\boldsymbol y=\mathbf e_j^\top\mathbf R_{S_h}^{(t)}/\psi_{h,j}^{(t)}
		\), and using the orthogonality of \(\mathbf W_{hS}\), gives
		\[
		\left\|
		\widetilde{\mathbf R}_{S,h,i}^{(t)}
		-
		\widetilde{\mathbf R}_{S,h,j}^{(t)}
		\right\|_2
		\le
		\frac{C}{m_{h,S,t}}
		\frac{
		\left\|
		\mathbf S_h^{(t)}
		\left(
		\mathbf e_i/\psi_{h,i}^{(t)}
		-
		\mathbf e_j/\psi_{h,j}^{(t)}
		\right)
		\right\|_2
		}{\sqrt{p\rho_hc_{h,S}}}.
		\] Concatenating the two channel rows then yields
		\[
		\left\|
		\widetilde{\mathbf J}_{S,i}^{(t)}
		-
		\widetilde{\mathbf J}_{S,j}^{(t)}
		\right\|_2
		\le
		\frac{C}{m_{\min,S,t}}
		\left[
		\frac{
		\left\|\mathbf S_Y^{(t)}
		(\mathbf e_i/\psi_{Y,i}^{(t)}-
		\mathbf e_j/\psi_{Y,j}^{(t)})\right\|_2^2
		}{p\rho_Yc_{Y,S}}
		+
		\frac{
		\left\|\mathbf S_Z^{(t)}
		(\mathbf e_i/\psi_{Z,i}^{(t)}-
		\mathbf e_j/\psi_{Z,j}^{(t)})\right\|_2^2
		}{p\rho_Zc_{Z,S}}
		\right]^{1/2}.
		\]
		Finally,
		\[
		\left\|
		\widetilde{\mathbf J}_{S,i}^{(t)}-
		\widetilde{\boldsymbol c}_{S,a}^{(t)}
		\right\|_2
		\le
		\max_{j\in C_a^{(t)}}
		\left\|
		\widetilde{\mathbf J}_{S,i}^{(t)}-
		\widetilde{\mathbf J}_{S,j}^{(t)}
		\right\|_2.
		\]
		Taking the maximum over \(i\) and \(a\) proves the lemma.
	\end{proof}
    
	\begin{proof}
		By Lemma~\ref{lem:eta-channel-reduction}, it remains to
		bound the two channelwise numerators. Fix \(a\in[q_t]\) and
		\(i,j\in C_a^{(t)}\).
		For \(k\notin\{i,j\}\), the vector inside \(\mathbf S_Y^{(t)}\)
		has nonzero entries only at coordinates \(i\) and \(j\). Moreover,
		\[
		S_{Y,kr}^{(t)}
		=
		\alpha_{kr}^{(t)}\bigl(1-\Pi_{kr}^{(t-1)}\bigr),
		\qquad
		\alpha_{kr}^{(t)}
		=
		\rho_Y\psi_{Y,k}^{(t)}\psi_{Y,r}^{(t)}
		\left(\boldsymbol\theta_{Y,\nu^{(t)}(k)}^{(t)}\right)^\top
		\boldsymbol\theta_{Y,\nu^{(t)}(r)}^{(t)}.
		\]
		Since \(i,j\in C_a^{(t)}\), we have
		\(\nu^{(t)}(i)=\nu^{(t)}(j)=a\). Therefore,
		\[
		\begin{aligned}
		\frac{1}{\sqrt{\rho_Y}}
		\left[
		\mathbf S_Y^{(t)}
		\left(
		\frac{\mathbf e_i}{\psi_{Y,i}^{(t)}}
		-
		\frac{\mathbf e_j}{\psi_{Y,j}^{(t)}}
		\right)
		\right]_k
		&=
		\frac{1}{\sqrt{\rho_Y}}
		\left\{
		\frac{S_{Y,ki}^{(t)}}{\psi_{Y,i}^{(t)}}
		-
		\frac{S_{Y,kj}^{(t)}}{\psi_{Y,j}^{(t)}}
		\right\}\\
		&=
		\sqrt{\rho_Y}\,\psi_{Y,k}^{(t)}
		\left(\boldsymbol\theta_{Y,\nu^{(t)}(k)}^{(t)}\right)^\top
		\boldsymbol\theta_{Y,a}^{(t)}
		\left\{
		\bigl(1-\Pi_{ki}^{(t-1)}\bigr)
		-
		\bigl(1-\Pi_{kj}^{(t-1)}\bigr)
		\right\}\\
		&=
		\sqrt{\rho_Y}\,\psi_{Y,k}^{(t)}
		\left(\boldsymbol\theta_{Y,\nu^{(t)}(k)}^{(t)}\right)^\top
		\boldsymbol\theta_{Y,a}^{(t)}
		\left(
		\Pi_{kj}^{(t-1)}-\Pi_{ki}^{(t-1)}
		\right).
		\end{aligned}
		\]
		For \(k\in\{i,j\}\), the completed diagonal satisfies
		\(0\le S_{Y,kk}^{(t)}\le
		\rho_Y(\psi_{Y,k}^{(t)})^2\), while
		\(0\le S_{Y,ij}^{(t)}\le
		\rho_Y\psi_{Y,i}^{(t)}\psi_{Y,j}^{(t)}\). Thus each exceptional
		coordinate is \(O(\rho_Y)\), and their combined contribution after
		division by \(\sqrt{p\rho_Y}\) is
		\(O(\sqrt{\rho_Y/p})=O(p^{-1/2})\), since \(\rho_Y\le1\).
		For \(k\notin\{i,j\}\), the preceding identity, together with
		\(\rho_Y\le1\), \(\psi_{Y,k}^{(t)}\le\ell\), and
		\(\|\boldsymbol\theta_{Y,b}^{(t)}\|_2=1\), bounds the remaining
		coordinates by \(C\omega_{\Pi,S,t-1}^{(t)}\). Combining the two
		bounds gives
		\[
		\frac{
			\left\|
			\mathbf S_Y^{(t)}
			\left(
			\frac{\mathbf e_i}{\psi_{Y,i}^{(t)}}
			-
			\frac{\mathbf e_j}{\psi_{Y,j}^{(t)}}
			\right)
			\right\|_2
		}{
			\sqrt{p\rho_Y}
		}
		\le
		C\left(
		\omega_{\Pi,S,t-1}^{(t)}+p^{-1/2}
		\right).
		\]
		The same coordinate calculation for the \(Z\) channel
		gives the analogous bound with \(\rho_Z\). Substitution into
		Lemma~\ref{lem:eta-channel-reduction} proves the result.
	\end{proof}
	\subsubsection*{S.3.2 Proofs of Theorem~\ref{thm:dcsbm_community_exact_recovery}
and Corollary~\ref{cor:pooled_community_exact_recovery}}
	\begin{lemma}\label{lem:normalized_joint_rowwise}
		Under Conditions~\ref{ass:transition-sparsity}, \ref{ass:initial-finite-rank},
		and~\ref{ass:3} in the AR(1)-DCSBM model, fix \(t\in[n]\). Recall that
		\[
		m_{Y,S,t}
		=
		\min_{i\in[p]}
		\|\mathbf e_i^\top\mathbf R_{S_Y}^{(t)}\|_2,
		\qquad
		m_{Z,S,t}
		=
		\min_{i\in[p]}
		\|\mathbf e_i^\top\mathbf R_{S_Z}^{(t)}\|_2,
		\]
		and define
		\[
		E_t
		=
		\left\{
		\|\mathbf R_Y^{(t)}
		-\mathbf R_{S_Y}^{(t)}\mathbf W_{YS}\|_{2\to\infty}
		\le\epsilon_p
		\right\}
		\cap
		\left\{
		\|\mathbf R_Z^{(t)}
		-\mathbf R_{S_Z}^{(t)}\mathbf W_{ZS}\|_{2\to\infty}
		\le\epsilon_p
		\right\}.
		\]
		Then, for every fixed \(c_0>0\) and all sufficiently large \(p\),
		\(
		\mathbb P(E_t^c)\le C_{c_0}p^{-c_0}.
		\)
		Assume \(m_{Y,S,t}>\epsilon_p\) and
		\(m_{Z,S,t}>\epsilon_p\). On \(E_t\), the sample normalised rows
		are well-defined and
		\[
		\|\widetilde{\mathbf R}_Y^{(t)}
		-\widetilde{\mathbf R}_{S_Y}^{(t)}\|_{2\to\infty}
		\le
		\frac{2\epsilon_p}{m_{Y,S,t}},
		\qquad
		\|\widetilde{\mathbf R}_Z^{(t)}
		-\widetilde{\mathbf R}_{S_Z}^{(t)}\|_{2\to\infty}
		\le
		\frac{2\epsilon_p}{m_{Z,S,t}}.
		\]
		Consequently,
		\begin{equation*}
			\|\widetilde{\mathbf J}^{(t)}
			-\widetilde{\mathbf J}_S^{(t)}\|_{2\to\infty}
			\le
			\sqrt{2}\,\epsilon_p
			\left(
			m_{Y,S,t}^{-2}+m_{Z,S,t}^{-2}
			\right)^{1/2}
			=
			\widetilde\epsilon_t,
		\end{equation*}
        which is $\tilde\epsilon_t$ of Theorem 5.
	\end{lemma}
	
	\begin{proof}
		The probability bound follows from Theorem~\ref{thm1}. On \(E_t\),
		for every \(i\in[p]\),
		\[
		\|\mathbf e_i^\top\mathbf R_Y^{(t)}\|_2
		\ge
		\|\mathbf e_i^\top\mathbf R_{S_Y}^{(t)}\|_2
		-\epsilon_p
		\ge
		m_{Y,S,t}-\epsilon_p
		>
		0.
		\]
		Thus the sample normalised rows in the \(Y\) channel are well-defined.
		The inequality
		\(
		\left\|
		\frac{\boldsymbol x}{\|\boldsymbol x\|_2}
		-
		\frac{\boldsymbol y}{\|\boldsymbol y\|_2}
		\right\|_2
		\le
		\frac{2\|\boldsymbol x-\boldsymbol y\|_2}{\|\boldsymbol y\|_2}
		\)
		with
		\(\boldsymbol x=\mathbf e_i^\top\mathbf R_Y^{(t)}\) and
		\(\boldsymbol y=\mathbf e_i^\top
		\mathbf R_{S_Y}^{(t)}\mathbf W_{YS}\) gives
		\[
		\|\widetilde{\mathbf R}_{Y,i}^{(t)}
		-\widetilde{\mathbf R}_{S_Y,i}^{(t)}\|_2
		\le
		\frac{2\epsilon_p}{m_{Y,S,t}}.
		\]
		The same argument for the \(Z\) channel gives
		\[
		\|\widetilde{\mathbf R}_{Z,i}^{(t)}
		-\widetilde{\mathbf R}_{S_Z,i}^{(t)}\|_2
		\le
		\frac{2\epsilon_p}{m_{Z,S,t}}.
		\]
		Taking the maximum over \(i\) proves the two channelwise bounds. Finally,
		the factor \(1/\sqrt2\) in the concatenation gives
		\[
		\begin{aligned}
			\|\widetilde{\mathbf J}^{(t)}
			-\widetilde{\mathbf J}_S^{(t)}\|_{2\to\infty}
			&\le
			\frac{1}{\sqrt2}
			\left\{
			\left(\frac{2\epsilon_p}{m_{Y,S,t}}\right)^2
			+
			\left(\frac{2\epsilon_p}{m_{Z,S,t}}\right)^2
			\right\}^{1/2}\\
			&=
			\sqrt2\,\epsilon_p
			\left(
			m_{Y,S,t}^{-2}+m_{Z,S,t}^{-2}
			\right)^{1/2}.
		\end{aligned}
		\]
	\end{proof}

	\begin{proof}[Proof of Theorem \ref{thm:dcsbm_community_exact_recovery}]
		Define
		$
		\widetilde{\mathbf{M}}_S^{(t)}
		:=
		\bigl[
		\widetilde{\boldsymbol{c}}_{S,\nu^{(t)}(1)}^{(t)}
		\mid
		\cdots
		\mid
		\widetilde{\boldsymbol{c}}_{S,\nu^{(t)}(p)}^{(t)}
		\bigr]^\top
		\in\mathbb{R}^{p\times (K_Y+K_Z)}.
		$
		Define the fitted-centre matrix by
		$
		\widehat{\mathbf M}^{(t)}
		:=
		\bigl[
		\widehat{\boldsymbol m}_{\widehat\nu^{(t)}(1)}^{(t)}
		\mid\cdots\mid
		\widehat{\boldsymbol m}_{\widehat\nu^{(t)}(p)}^{(t)}
		\bigr]^\top.
		$
		A fitted centre is called nonempty if its fitted cluster
		\(\{i\in[p]:\widehat\nu^{(t)}(i)=b\}\) contains at least one row.
		Since $\widehat{\mathbf{M}}^{(t)}$ minimises the $q_t$-means objective,
		$$
		\|\widetilde{\mathbf{J}}^{(t)}-\widehat{\mathbf{M}}^{(t)}\|_F
		\le
		\|\widetilde{\mathbf{J}}^{(t)}-\widetilde{\mathbf{M}}_S^{(t)}\|_F,
		$$
		and hence
		$$
		\|\widehat{\mathbf{M}}^{(t)}-\widetilde{\mathbf{M}}_S^{(t)}\|_F
		\le
		\|\widetilde{\mathbf{J}}^{(t)}-\widetilde{\mathbf{M}}_S^{(t)}\|_F
		+
		\|\widetilde{\mathbf{J}}^{(t)}-\widehat{\mathbf{M}}^{(t)}\|_F
		\le
		2\|\widetilde{\mathbf{J}}^{(t)}-\widetilde{\mathbf{M}}_S^{(t)}\|_F.
		$$
		
		For each $i\in C_a^{(t)}$,
		$$
		\|\widetilde{\mathbf{J}}_i^{(t)}-\widetilde{\boldsymbol{c}}_{S,a}^{(t)}\|_2
		\le
		\|\widetilde{\mathbf{J}}_i^{(t)}-\widetilde{\mathbf{J}}_{S,i}^{(t)}\|_2
		+
		\|\widetilde{\mathbf{J}}_{S,i}^{(t)}-\widetilde{\boldsymbol{c}}_{S,a}^{(t)}\|_2
		\le
		\widetilde{\epsilon}_t+\widetilde{\eta}_{S,t}.
		$$
		Therefore
		$$
		\|\widetilde{\mathbf{J}}^{(t)}-\widetilde{\mathbf{M}}_S^{(t)}\|_F
		\le
		\sqrt{p}\,(\widetilde{\epsilon}_t+\widetilde{\eta}_{S,t}),
		$$
		so
		$$
		\|\widehat{\mathbf{M}}^{(t)}-\widetilde{\mathbf{M}}_S^{(t)}\|_F
		\le
		2\sqrt{p}\,(\widetilde{\epsilon}_t+\widetilde{\eta}_{S,t}).
		$$
		
		Set
		$
		\widetilde r_t
		:=
		2(\widetilde\epsilon_t+\widetilde\eta_{S,t})\sqrt{\frac{p}{s_t^{\min}}}.
		$
		For each $a\in[q_t]$, define the ball
		$$
		\widetilde B_a
		:=
		\left\{
		\boldsymbol x\in\mathbb R^{K_Y+K_Z}:
		\|\boldsymbol x-\widetilde{\boldsymbol c}_{S,a}^{(t)}\|_2\le \widetilde r_t
		\right\}.
		$$
		
		We first show that each $\widetilde B_a$ contains at least one nonempty fitted centre of
		$\widehat{\mathbf M}^{(t)}$.
		Suppose, to the contrary, that for some $a\in[q_t]$, the ball $\widetilde B_a$
		contains no nonempty fitted centre. Then, for every $i\in C_a^{(t)}$, since
		$\mathbf e_i^\top \widetilde{\mathbf M}_S^{(t)}=\widetilde{\boldsymbol c}_{S,a}^{(t)}$, we have
		$$
		\|\mathbf e_i^\top(\widehat{\mathbf M}^{(t)}-\widetilde{\mathbf M}_S^{(t)})\|_2
		=
		\|\mathbf e_i^\top\widehat{\mathbf M}^{(t)}-\widetilde{\boldsymbol c}_{S,a}^{(t)}\|_2
		>
		\widetilde r_t.
		$$
		Therefore,
		$$
		\|\widehat{\mathbf M}^{(t)}-\widetilde{\mathbf M}_S^{(t)}\|_F^2
		\ge
		\sum_{i\in C_a^{(t)}}
		\|\mathbf e_i^\top(\widehat{\mathbf M}^{(t)}-\widetilde{\mathbf M}_S^{(t)})\|_2^2
		>
		|C_a^{(t)}|\,\widetilde r_t^2
		\ge
		s_t^{\min}\,\widetilde r_t^2
		=
		4p(\widetilde\epsilon_t+\widetilde\eta_{S,t})^2,
		$$
		which contradicts the bound
		$$
		\|\widehat{\mathbf M}^{(t)}-\widetilde{\mathbf M}_S^{(t)}\|_F
		\le
		2\sqrt p\,(\widetilde\epsilon_t+\widetilde\eta_{S,t}).
		$$
		Hence every ball $\widetilde B_a$ contains at least one nonempty fitted centre.
		
		Next, by \eqref{eq:dcsbm_community_sep_condition},
		$
		\widetilde\Delta_{S,t}>3\widetilde r_t,
		$
		so the balls $\widetilde B_1,\ldots,\widetilde B_{q_t}$ are pairwise disjoint.
		The fitted partition has at most \(q_t\) nonempty centres, while the \(q_t\)
		disjoint balls each contain at least one nonempty fitted centre.  It follows
		that each ball $\widetilde B_a$ contains exactly one nonempty fitted centre.
		Hence there exists a permutation
		$\pi_t\in \mathfrak S_{q_t}$ such that the fitted centre indexed by $\pi_t(a)$ lies in
		$\widetilde B_a$ for every $a\in[q_t]$.

		Fix $i\in C_a^{(t)}$, and let
		$\widehat{\boldsymbol m}_{\pi_t(a)}^{(t)}$
		be the fitted centre in $\widetilde{B}_a$. Then
		$$
		\|\widetilde{\mathbf{J}}_i^{(t)}
		-
		\widehat{\boldsymbol m}_{\pi_t(a)}^{(t)}\|_2
		\le
		\|\widetilde{\mathbf{J}}_i^{(t)}-\widetilde{\boldsymbol{c}}_{S,a}^{(t)}\|_2
		+
		\|\widetilde{\boldsymbol{c}}_{S,a}^{(t)}
		-
		\widehat{\boldsymbol m}_{\pi_t(a)}^{(t)}\|_2
		\le
		\widetilde{\epsilon}_t+\widetilde{\eta}_{S,t}+\widetilde{r}_t.
		$$
		For any $b\neq a$, let
		$\widehat{\boldsymbol m}_{\pi_t(b)}^{(t)}$
		be the fitted centre in $\widetilde{B}_b$. Then
		$$
		\|\widetilde{\mathbf{J}}_i^{(t)}
		-
		\widehat{\boldsymbol m}_{\pi_t(b)}^{(t)}\|_2
		\ge
		\|\widetilde{\boldsymbol{c}}_{S,a}^{(t)}-\widetilde{\boldsymbol{c}}_{S,b}^{(t)}\|_2
		-
		\|\widetilde{\boldsymbol{c}}_{S,b}^{(t)}
		-
		\widehat{\boldsymbol m}_{\pi_t(b)}^{(t)}\|_2
		-
		\|\widetilde{\mathbf{J}}_i^{(t)}-\widetilde{\boldsymbol{c}}_{S,a}^{(t)}\|_2
		\ge
		\widetilde{\Delta}_{S,t}
		-
		\widetilde{r}_t
		-
		\widetilde{\epsilon}_t
		-
		\widetilde{\eta}_{S,t}.
		$$
		Since $\widetilde{r}_t\ge 2(\widetilde{\epsilon}_t+\widetilde{\eta}_{S,t})$ and
		$\widetilde{\Delta}_{S,t}>3\widetilde{r}_t$, we have
		$$
		\widetilde{\epsilon}_t+\widetilde{\eta}_{S,t}+\widetilde{r}_t
		<
		\widetilde{\Delta}_{S,t}-\widetilde{r}_t-\widetilde{\epsilon}_t-\widetilde{\eta}_{S,t}.
		$$
		Hence every vertex is strictly closer to the fitted centre matched to its true community
		than to any other fitted centre. Therefore
		$
		\widehat{\nu}^{(t)}(i)=\pi_t(\nu^{(t)}(i)),
		 i\in[p],
		$
		so exact recovery holds on \(E_t\). By
		Lemma~\ref{lem:normalized_joint_rowwise}, for all sufficiently large \(p\),
		\(
		\mathbb P(E_t)
		\ge
		1-C_{c_0}p^{-c_0}.
		\)
		Therefore,
		\[
		\mathbb P\left\{
		\min_{\pi\in\mathfrak S_{q_t}}
		\left|
		\left\{
		i\in[p]:
		\widehat{\nu}^{(t)}(i)
		\neq
		\pi\bigl(\nu^{(t)}(i)\bigr)
		\right\}
		\right|
		=0
		\right\}
		\ge
		1-C_{c_0}p^{-c_0}.
		\]
		The preceding argument applies to every \(t\in[n]\).  Since \(n\) is
		fixed, a union bound gives
		\[
		\mathbb P\left\{
		\bigcap_{t=1}^{n}
		\left[
		\min_{\pi\in\mathfrak S_{q_t}}
		\left|
		\left\{
		i\in[p]:
		\widehat{\nu}^{(t)}(i)
		\neq
		\pi\bigl(\nu^{(t)}(i)\bigr)
		\right\}
		\right|
		=0
		\right]
		\right\}
		\ge
		1-C_{c_0}n p^{-c_0},
		\]
		which proves simultaneous exact recovery over all \(t\in[n]\).
	\end{proof}

	\begin{proof}[Proof of Corollary~\ref{cor:pooled_community_exact_recovery}]
		On \(\mathcal E_{\rm UASE}(c_0)\), Lemma~\ref{lem:normalized_joint_rowwise}
		and the triangle inequality give
		\[
		\begin{aligned}
			\|\overline{\widetilde{\mathbf J}}_{\mathcal T}
			-\overline{\widetilde{\mathbf J}}_{S,\mathcal T}\|_{2\to\infty}
			&\le
			\frac{1}{|\mathcal T|}
			\sum_{t\in\mathcal T}
			\|\widetilde{\mathbf J}^{(t)}
			-\widetilde{\mathbf J}_S^{(t)}\|_{2\to\infty}\\
			&\le
			\frac{1}{|\mathcal T|}
			\sum_{t\in\mathcal T}\widetilde\epsilon_t
			=
			\widetilde\epsilon_{\mathcal T}.
		\end{aligned}
		\]
		Applying the deterministic argument in the proof of
		Theorem~\ref{thm:dcsbm_community_exact_recovery} to these averaged
		rows shows that the stated separation condition gives exact recovery
		on \(\mathcal E_{\rm UASE}(c_0)\). Finally,
		\(\mathbb P\{\mathcal E_{\rm UASE}(c_0)^c\}
		\le C_{c_0}p^{-c_0}\) by Theorem~\ref{thm1}.
	\end{proof}
	
	\subsubsection*{S.3.3 Proof of Corollary \ref{thm:dcsbm_jump_set_exact_recovery}}
	\begin{lemma}
		\label{lem:dcsbm_endpoint_center_error}
		Suppose the conditions of
		Corollary~\ref{cor:pooled_community_exact_recovery} hold with
		\(\mathcal T=B\), so that the baseline partition is recovered exactly up
		to a permutation of labels.  On the exact-recovery event and
		\(E_{t^\ast-1}\), after relabelling the estimated communities,
		\[
		\max_{a\in[q_B]}
		\left\|
		\widetilde{\boldsymbol c}_{B,a}^{(t^\ast-1)}
		-
		\widetilde{\boldsymbol c}_{S,B,a}^{(t^\ast-1)}
		\right\|_2
		\le
		\widetilde\epsilon_{t^\ast-1}.
		\]
	\end{lemma}
	
	\begin{proof}
		On the exact-recovery event, after relabelling,
		\(\widehat C_{B,a}=C_{S,B,a}\) for all \(a\in[q_B]\).  Hence
		\[
		\widetilde{\boldsymbol c}_{B,a}^{(t^\ast-1)}
		-
		\widetilde{\boldsymbol c}_{S,B,a}^{(t^\ast-1)}
		=
		\frac{1}{|C_{S,B,a}|}
		\sum_{i\in C_{S,B,a}}
		\left(
		\widetilde{\mathbf J}_i^{(t^\ast-1)}
		-
		\widetilde{\mathbf J}_{S,i}^{(t^\ast-1)}
		\right).
		\]
		The triangle inequality and the rowwise bound on \(E_{t^\ast-1}\)
		give
		\[
		\left\|
		\widetilde{\boldsymbol c}_{B,a}^{(t^\ast-1)}
		-
		\widetilde{\boldsymbol c}_{S,B,a}^{(t^\ast-1)}
		\right\|_2
		\le
		\frac{1}{|C_{S,B,a}|}
		\sum_{i\in C_{S,B,a}}
		\left\|
		\widetilde{\mathbf J}_i^{(t^\ast-1)}
		-
		\widetilde{\mathbf J}_{S,i}^{(t^\ast-1)}
		\right\|_2
		\le
		\widetilde\epsilon_{t^\ast-1}.
		\]
		Taking the maximum over \(a\) proves the claim.
	\end{proof}
	
	\begin{proof}[Proof of Corollary~\ref{thm:dcsbm_jump_set_exact_recovery}]
    On \(\bigcap_{s\in B}E_s\), the proof of
		Corollary~\ref{cor:pooled_community_exact_recovery} gives exact recovery
		of the baseline partition. After relabelling, we have 
		\(
		\widehat C_{B,a}=C_{S,B,a}, a\in[q_B].
		\)
        Define the population nearest-centre label by
		\[
		\tau_{S,B}^{(t^*-1,t)}(i)
		:=
		\arg\min_{b\in[q_B]}
		\left\|
		\widetilde{\mathbf J}_{S,i}^{(t)}
		-
		\widetilde{\mathbf c}_{S,B,b}^{(t^*-1)}
		\right\|_2.
		\]
		Fix any $i\in[p]$ and any $b\in[q_B]$. By the reverse triangle inequality,
		$$
		\left|
		\left\|\widetilde{\mathbf{J}}_i^{(t)}
		-\widetilde{\boldsymbol c}_{B,b}^{(t^\ast-1)}\right\|_2
		-
		\left\|\widetilde{\mathbf{J}}_{S,i}^{(t)}
		-\widetilde{\boldsymbol c}_{S,B,b}^{(t^\ast-1)}\right\|_2
		\right|
		\le
		\|\widetilde{\mathbf{J}}_i^{(t)}-\widetilde{\mathbf{J}}_{S,i}^{(t)}\|_2
		+
		\left\|\widetilde{\boldsymbol c}_{B,b}^{(t^\ast-1)}
		-\widetilde{\boldsymbol c}_{S,B,b}^{(t^\ast-1)}\right\|_2.
		$$
		On the event $E_t\cap\bigcap_{s\in B}E_s$,
		Lemma~\ref{lem:normalized_joint_rowwise} and
		Lemma~\ref{lem:dcsbm_endpoint_center_error} imply
		$$
		\|\widetilde{\mathbf{J}}_i^{(t)}-\widetilde{\mathbf{J}}_{S,i}^{(t)}\|_2
		\le
		\widetilde{\epsilon}_t,
		\qquad
		\left\|\widetilde{\boldsymbol c}_{B,b}^{(t^\ast-1)}
		-\widetilde{\boldsymbol c}_{S,B,b}^{(t^\ast-1)}\right\|_2
		\le
		\widetilde\epsilon_{t^\ast-1},
		$$
		and therefore
		\begin{equation}\label{eq:dcsbm_node_destination_perturb}
			\left|
			\left\|\widetilde{\mathbf{J}}_i^{(t)}
			-\widetilde{\boldsymbol c}_{B,b}^{(t^\ast-1)}\right\|_2
			-
			\left\|\widetilde{\mathbf{J}}_{S,i}^{(t)}
			-\widetilde{\boldsymbol c}_{S,B,b}^{(t^\ast-1)}\right\|_2
			\right|
			\le
			\overline{\epsilon}_{t^\ast-1,t}.
		\end{equation}
		
		Fix \(i\in\mathcal A^{(t^*-1,t)}\) and write
		\(b^\star=\nu^{(t)}(i)\).
		By the definition of \(\widetilde\eta_{S,t}\) and
		\eqref{eq:dcsbm_centre_drift},
		\begin{align*}
		\left\|
		\widetilde{\mathbf J}_{S,i}^{(t)}
		-
		\widetilde{\boldsymbol c}_{S,B,b^\star}^{(t^\ast-1)}
		\right\|_2
		\le
		\left\|
		\widetilde{\mathbf J}_{S,i}^{(t)}
		-
		\widetilde{\boldsymbol c}_{S,b^\star}^{(t)}
		\right\|_2
		+
		\left\|
		\widetilde{\boldsymbol c}_{S,b^\star}^{(t)}
		-
		\widetilde{\boldsymbol c}_{S,B,b^\star}^{(t^\ast-1)}
		\right\|_2
		\le
		\widetilde\eta_{S,t}+\widetilde\zeta_{S,B,t}.
		\end{align*}
		For every \(c\neq b^\star\), the triangle inequality gives
		\begin{align*}
		\left\|
		\widetilde{\mathbf J}_{S,i}^{(t)}
		-
		\widetilde{\boldsymbol c}_{S,B,c}^{(t^\ast-1)}
		\right\|_2
		&\ge
		\left\|
		\widetilde{\boldsymbol c}_{S,B,b^\star}^{(t^\ast-1)}
		-
		\widetilde{\boldsymbol c}_{S,B,c}^{(t^\ast-1)}
		\right\|_2
		-
		\left\|
		\widetilde{\mathbf J}_{S,i}^{(t)}
		-
		\widetilde{\boldsymbol c}_{S,B,b^\star}^{(t^\ast-1)}
		\right\|_2\\
		&\ge
		\min_{a\neq b}
		\left\|
		\widetilde{\boldsymbol c}_{S,B,a}^{(t^\ast-1)}
		-
		\widetilde{\boldsymbol c}_{S,B,b}^{(t^\ast-1)}
		\right\|_2
		-\widetilde\eta_{S,t}
		-\widetilde\zeta_{S,B,t}.
		\end{align*}
		Hence
		\begin{align*}
		&\left\|
		\widetilde{\mathbf J}_{S,i}^{(t)}
		-
		\widetilde{\boldsymbol c}_{S,B,c}^{(t^\ast-1)}
		\right\|_2
		-
		\left\|
		\widetilde{\mathbf J}_{S,i}^{(t)}
		-
		\widetilde{\boldsymbol c}_{S,B,b^\star}^{(t^\ast-1)}
		\right\|_2\\
		&\ge
		\min_{a\neq b}
		\left\|
		\widetilde{\boldsymbol c}_{S,B,a}^{(t^\ast-1)}
		-
		\widetilde{\boldsymbol c}_{S,B,b}^{(t^\ast-1)}
		\right\|_2
		-2\widetilde\eta_{S,t}
		-2\widetilde\zeta_{S,B,t}
		>
		2\overline\epsilon_{t^\ast-1,t},
		\qquad c\neq b^\star,
		\end{align*}
		where the strict inequality follows from
		\eqref{eq:dcsbm_jump_set_margin}. Thus
		\(\tau_{S,B}^{(t^*-1,t)}(i)=b^\star=\nu^{(t)}(i)\).
		Moreover, \eqref{eq:dcsbm_node_destination_perturb} yields
		\begin{align*}
		\left\|
		\widetilde{\mathbf J}_i^{(t)}
		-
		\widetilde{\boldsymbol c}_{B,b^\star}^{(t^\ast-1)}
		\right\|_2
		&\le
		\left\|
		\widetilde{\mathbf J}_{S,i}^{(t)}
		-
		\widetilde{\boldsymbol c}_{S,B,b^\star}^{(t^\ast-1)}
		\right\|_2
		+\overline\epsilon_{t^\ast-1,t}\\
		&<
		\left\|
		\widetilde{\mathbf J}_{S,i}^{(t)}
		-
		\widetilde{\boldsymbol c}_{S,B,c}^{(t^\ast-1)}
		\right\|_2
		-\overline\epsilon_{t^\ast-1,t}\\
		&\le
		\left\|
		\widetilde{\mathbf J}_i^{(t)}
		-
		\widetilde{\boldsymbol c}_{B,c}^{(t^\ast-1)}
		\right\|_2,
		\qquad c\neq b^\star.
		\end{align*}
		Therefore, 
		\(
		\widehat\tau_B^{(t^*-1,t)}(i)
		=
		\tau_{S,B}^{(t^*-1,t)}(i)
		=
		\nu^{(t)}(i), i\in\mathcal A^{(t^*-1,t)}.
		\)

		Consider the event $\mathcal{E}_{\mathrm{jump}}^{(t^*-1,t)}
		:=
		\left\{
		\widehat{\mathcal A}^{(t^*-1,t)}=\mathcal A^{(t^*-1,t)}
		\right\}
		\cap
		E_t
		\cap
		\bigcap_{s\in B}E_s.$ 
		On this event,
		Corollary~\ref{cor:pooled_community_exact_recovery} gives
		\(
		\widehat{\nu}_B(i)=\nu_{S,B}(i), i\in[p],
		\)
			after relabelling, and the preceding argument gives
		\(
		\widehat\tau_B^{(t^*-1,t)}(i)=\nu^{(t)}(i), i\in\mathcal A^{(t^*-1,t)}.
		\)
		Since every membership-changing vertex belongs to
		\(\mathcal A^{(t^*-1,t)}\), the exact recovery of
		\(\mathcal A^{(t^*-1,t)}\), \(\nu_{S,B}\), and the target nearest-centre
		assignment implies, on \(\mathcal E_{\mathrm{jump}}^{(t^*-1,t)}\),
		\[
		\widehat{\mathcal J}_B^{(t^*-1,t)}
		=
		\left\{i\in\mathcal A^{(t^*-1,t)}:
		\nu^{(t)}(i)\neq\nu_{S,B}(i)\right\}
		=
		\mathcal J_{S,B}^{(t^*-1,t)}.
		\]
		
		Finally, Lemma~\ref{lem:normalized_joint_rowwise} gives, for every
		\(u\in B\cup\{t\}\),
		\(
		\mathbb P(E_u^c)
		\le
		C_{c_0}p^{-c_0}.
		\)
		Therefore,
		\[
		\begin{aligned}
		\mathbb P\left[
		\left(
		E_t\cap\bigcap_{u\in B}E_u
		\right)^c
		\right]
		&\le
		\mathbb P(E_t^c)
		+
		\sum_{u\in B}\mathbb P(E_u^c) \\
		&\le
		C_{c_0}(|B|+1)p^{-c_0}.
		\end{aligned}
		\]
		By Corollary~\ref{cor:perfect-v},
		\[
		\mathbb P\left\{
		\widehat{\mathcal A}^{(t^*-1,t)}
		\neq
		\mathcal A^{(t^*-1,t)}
		\right\}
		\le
		C_{c_0}p^{-c_0}.
		\]
		Hence, after enlarging \(C_{c_0}\),
		\[
		\begin{aligned}
		\mathbb P\left\{
		\left(
		\mathcal E_{\mathrm{jump}}^{(t^*-1,t)}
		\right)^c
		\right\}
		&\le
		\mathbb P\left\{
		\widehat{\mathcal A}^{(t^*-1,t)}
		\neq
		\mathcal A^{(t^*-1,t)}
		\right\} +
		\mathbb P\left[
		\left(
		E_t\cap\bigcap_{u\in B}E_u
		\right)^c
		\right] \\
		&\le
		C_{c_0}(|B|+1)p^{-c_0}.
		\end{aligned}
		\]
		Since the two reassignment sets agree on
		\(\mathcal E_{\mathrm{jump}}^{(t^*-1,t)}\),
		\[
		\mathbb P\left\{
		\widehat{\mathcal J}_B^{(t^*-1,t)}
		=
		\mathcal J_{S,B}^{(t^*-1,t)}
		\right\}
		\ge
		1-C_{c_0}(|B|+1)p^{-c_0}.
		\]
	\end{proof}
	\subsubsection*{S.3.4 Proof of Theorem \ref{thm:community_score_detection}}
	\begin{lemma}\label{thm:dcsbm_community_level_test}
		Suppose the conditions of
		Corollary~\ref{cor:pooled_community_exact_recovery} hold with
		\(\mathcal T=B\), and
		Conditions~\ref{ass:transition-sparsity}--\ref{ass:3} hold at time
		\(t\).  Assume that \(m_{Y,S,t}>\epsilon_p\) and
		\(m_{Z,S,t}>\epsilon_p\).
		Then, after relabelling the estimated baseline communities, simultaneously
		for every \(a\in[q_B]\) and every prespecified
		\(G\subseteq[q_B]\) with \(|G|\ge2\),
		\[
		\left|
		\widetilde T_{B,a}^{\mathrm{cent},(t^*-1,t)}
		-
		\widetilde T_{S,B,a}^{\mathrm{cent},(t^*-1,t)}
		\right|
		\le
		\overline\epsilon_{t^\ast-1,t},
		\]
		\[
		\left|
		\widetilde T_{B,a}^{\mathrm{split},(t^*-1,t)}
		-
		\widetilde T_{S,B,a}^{\mathrm{split},(t^*-1,t)}
		\right|
		\le
		\overline\epsilon_{t^\ast-1,t},
		\]
		and
		\[
		\left|
		\widetilde T_{B,G}^{\mathrm{merge},(t^*-1,t)}
		-
		\widetilde T_{S,B,G}^{\mathrm{merge},(t^*-1,t)}
		\right|
		\le
		\overline\epsilon_{t^\ast-1,t}.
		\]
		These inequalities hold with probability at least
		\(1-C_{c_0}(|B|+1)p^{-c_0}\) for every fixed \(c_0>0\) and all
		sufficiently large \(p\).

		More generally, if the same statistics are formed with a reference time
		\(u\) and a target time \(v\), then, on the corresponding exact-recovery
		event and \(E_u\cap E_v\), each sample--population test statistic error is at most
		\(
		\overline\epsilon_{u,v}
		:=
		\widetilde\epsilon_u+\widetilde\epsilon_v.
		\)
	\end{lemma}
	
	\begin{proof}
		On the baseline exact-recovery event, relabel the estimated communities so
		that \(\widehat C_{B,a}=C_{S,B,a}\).  On
		\(E_{t^\ast-1}\cap E_t\), Lemma~\ref{lem:normalized_joint_rowwise}
		gives, for \(s\in\{t^\ast-1,t\}\), we have 
		\(
		\max_{i\in[p]}
		\left\|
		\widetilde{\mathbf J}_i^{(s)}
		-
		\widetilde{\mathbf J}_{S,i}^{(s)}
		\right\|_2
		\le
		\widetilde\epsilon_s.
		\)
		Averaging over the common vertex set \(C_{S,B,a}\) therefore yields
		\[
		\max_{a\in[q_B]}
		\left\|
		\widetilde{\boldsymbol c}_{B,a}^{(s)}
		-
		\widetilde{\boldsymbol c}_{S,B,a}^{(s)}
		\right\|_2
		\le
		\widetilde\epsilon_s,
		\qquad
		s\in\{t^\ast-1,t\}.
		\]
		The centre-shift bound follows immediately from the reverse triangle
		inequality.

		For the split statistic, the definition of the vertex displacement gives
		\[
		\left\|
		\widetilde{\boldsymbol d}_{B,i}^{(t^*-1,t)}
		-
		\widetilde{\boldsymbol d}_{S,B,i}^{(t^*-1,t)}
		\right\|_2
		\le
		\left\|
		\widetilde{\mathbf J}_i^{(t)}
		-
		\widetilde{\mathbf J}_{S,i}^{(t)}
		\right\|_2
		+
		\left\|
		\widetilde{\mathbf J}_i^{(t^\ast-1)}
		-
		\widetilde{\mathbf J}_{S,i}^{(t^\ast-1)}
		\right\|_2
		\le
		\overline\epsilon_{t^\ast-1,t}.
		\]
		On exact recovery, the sample and population maximisations have the same
		candidate subsets after relabelling.  For
		\(A\in\mathcal P_{S,B,a}\), write
		\(
		\gamma_A
		:=
		\left\{
		\frac{|A|\,|C_{S,B,a}\setminus A|}
		{|C_{S,B,a}|^2}
		\right\}^{1/2}.
		\)
		Each subgroup mean displacement differs from its population counterpart
		by at most \(\overline\epsilon_{t^\ast-1,t}\).  Thus
		\[
		\left|
		\gamma_A
		\left\|
		\widetilde{\boldsymbol d}_{B,a,A}^{(t^*-1,t)}
		-
		\widetilde{\boldsymbol d}_{B,a,C_{S,B,a}\setminus A}^{(t^*-1,t)}
		\right\|_2
		\right.\\[-0.2em]
		\left.
		-
		\gamma_A
		\left\|
		\widetilde{\boldsymbol d}_{S,B,a,A}^{(t^*-1,t)}
		-
		\widetilde{\boldsymbol d}_{S,B,a,C_{S,B,a}\setminus A}^{(t^*-1,t)}
		\right\|_2
		\right|
		\le
		2\gamma_A\overline\epsilon_{t^\ast-1,t}
		\le
		\overline\epsilon_{t^\ast-1,t},
		\]
		because \(\gamma_A\le1/2\).  Applying
		\(\left|\max_A f_A-\max_A g_A\right|
		\le\max_A|f_A-g_A|\) proves the split bound.

		For the merge statistic, the reverse triangle inequality gives
		\begin{align*}
		\left|
		\widetilde D_{B,G}^{(s)}
		-
		\widetilde D_{S,B,G}^{(s)}
		\right|
		&\le
		\left\{
		\frac{1}{|G|}\sum_{a\in G}
		\left\|
		\widetilde{\boldsymbol c}_{B,a}^{(s)}
		-
		\widetilde{\boldsymbol c}_{S,B,a}^{(s)}
		-
		\frac{1}{|G|}\sum_{b\in G}
		\left(
		\widetilde{\boldsymbol c}_{B,b}^{(s)}
		-
		\widetilde{\boldsymbol c}_{S,B,b}^{(s)}
		\right)
		\right\|_2^2
		\right\}^{1/2}\\
		&=
		\left\{
		\frac{1}{|G|}\sum_{a\in G}
		\left\|
		\widetilde{\boldsymbol c}_{B,a}^{(s)}
		-
		\widetilde{\boldsymbol c}_{S,B,a}^{(s)}
		\right\|_2^2
		-
		\left\|
		\frac{1}{|G|}\sum_{a\in G}
		\left(
		\widetilde{\boldsymbol c}_{B,a}^{(s)}
		-
		\widetilde{\boldsymbol c}_{S,B,a}^{(s)}
		\right)
		\right\|_2^2
		\right\}^{1/2}\\
		&\le
		\left\{
		\frac{1}{|G|}\sum_{a\in G}
		\left\|
		\widetilde{\boldsymbol c}_{B,a}^{(s)}
		-
		\widetilde{\boldsymbol c}_{S,B,a}^{(s)}
		\right\|_2^2
		\right\}^{1/2}
		\le
		\widetilde\epsilon_s,
		\qquad
		s\in\{t^\ast-1,t\}.
		\end{align*}
		Using the definition of merge statistics directly,
		\begin{align*}
		\left|
		\widetilde T_{B,G}^{\mathrm{merge},(t^*-1,t)}
		-
		\widetilde T_{S,B,G}^{\mathrm{merge},(t^*-1,t)}
		\right|
		&=
		\left|
		\left\{
		\widetilde D_{B,G}^{(t^\ast-1)}
		-
		\widetilde D_{B,G}^{(t)}
		\right\}
		-
		\left\{
		\widetilde D_{S,B,G}^{(t^\ast-1)}
		-
		\widetilde D_{S,B,G}^{(t)}
		\right\}
		\right|\\
		&\le
		\left|
		\widetilde D_{B,G}^{(t^\ast-1)}
		-
		\widetilde D_{S,B,G}^{(t^\ast-1)}
		\right|
		+
		\left|
		\widetilde D_{B,G}^{(t)}
		-
		\widetilde D_{S,B,G}^{(t)}
		\right|\\
		&\le
		\overline\epsilon_{t^\ast-1,t}.
		\end{align*}

		The same proof with \(t^\ast-1\) and \(t\) replaced by \(u\) and
		\(v\) proves the general pairwise statement.  Finally, intersecting the
		baseline exact-recovery event with the rowwise events over \(B\) and at
		\(t\), and applying a union bound, gives the stated probability.
	\end{proof}

	\begin{proof}[Proof of Theorem \ref{thm:community_score_detection}]
		Lemma~\ref{thm:dcsbm_community_level_test} gives, simultaneously for the
		centre-shift, split, and merge statistics,
		\[
		\left|
		T_C^{(t^*-1,t)}
		-
		T_{S,C}^{(t^*-1,t)}
		\right|
		\le
		\overline\epsilon_{t^\ast-1,t}
		\]
		on an event with probability at least
		\(1-C_{c_0}(|B|+1)p^{-c_0}\).

		Under the centre-shift and split null hypotheses,
		\(T_{S,C}^{(t^*-1,t)}=0\), while under the merge null hypothesis,
		\(T_{S,C}^{(t^*-1,t)}\le0\).  Thus, under any of the three corresponding
		null hypotheses,
		\[
		T_C^{(t^*-1,t)}
		\le
		T_{S,C}^{(t^*-1,t)}
		+\overline\epsilon_{t^\ast-1,t}
		\le
		\overline\epsilon_{t^\ast-1,t}
		\le\tau.
		\]
		Because the test rejects only when \(T_C^{(t^*-1,t)}>\tau\), it does not
		reject on this event.

		Conversely, if
		\(T_{S,C}^{(t^*-1,t)}>
		\tau+\overline\epsilon_{t^\ast-1,t}\), then
		\[
		T_C^{(t^*-1,t)}
		\ge
		T_{S,C}^{(t^*-1,t)}
		-\overline\epsilon_{t^\ast-1,t}
		>\tau,
		\]
		so the corresponding test rejects on the same event.  This proves both
		claims.
	\end{proof}

    \subsubsection*{S.3.5 An additional corollary linking
\(\widetilde\eta_{S,t}\) to autoregressive memory}
	For a target time \(t\geq t^*\) after a stationary baseline ended at $t^*-1$, Proposition~\ref{prop:eta} can be refined by
separating the baseline within-community variation from the AR(1) memory
generated after the baseline endpoint.

\begin{corollary}
\label{cor:eta}
Suppose Conditions~\ref{ass:initial-finite-rank}--\ref{ass:3} hold and
\(m_{\min,S,t}>0\). Then there is a constant \(C>0\) such that
\[
\widetilde\eta_{S,t}
\leq
\frac{C}{m_{\min,S,t}}
\left[
\omega_{\Pi,S,t^*-1}^{(t)}
+2\max_{i\in[p]}
\left\{
\frac{\{M_{Y,i}^{(t^*-1,t)}\}^2}
{c_{Y,S}(\psi_{Y,i}^{(t)})^2}
+
\frac{\{M_{Z,i}^{(t^*-1,t)}\}^2}
{c_{Z,S}(\psi_{Z,i}^{(t)})^2}
\right\}^{1/2}
+p^{-1/2}
\right].
\]
\end{corollary}
 The first term is the variation already
present at the baseline endpoint; the \(M\)-terms are the additional
within-community spread inherited from the previous network state.
\begin{proof}
		By Lemma~\ref{lem:eta-channel-reduction}, it remains only
		to refine the two channelwise numerator bounds.		
		Continue with fixed \(a\in[q_t]\) and \(i,j\in C_a^{(t)}\). Since
		\(\nu^{(t)}(i)=\nu^{(t)}(j)=a\), for the \(Y\) channel and
		\(k\notin\{i,j\}\),
		\[
		\begin{aligned}
		\left[
		\mathbf S_Y^{(t)}
		\left(
		\frac{\mathbf e_i}{\psi_{Y,i}^{(t)}}
		-
		\frac{\mathbf e_j}{\psi_{Y,j}^{(t)}}
		\right)
		\right]_k
		&=
		\frac{\alpha_{ki}^{(t)}}{\psi_{Y,i}^{(t)}}
		\bigl(1-\Pi_{ki}^{(t-1)}\bigr)
		-
		\frac{\alpha_{kj}^{(t)}}{\psi_{Y,j}^{(t)}}
		\bigl(1-\Pi_{kj}^{(t-1)}\bigr)\\
		&=
		\left\{
		\frac{\alpha_{ki}^{(t)}}{\psi_{Y,i}^{(t)}}
		\bigl(1-\Pi_{ki}^{(t^\ast-1)}\bigr)
		-
		\frac{\alpha_{kj}^{(t)}}{\psi_{Y,j}^{(t)}}
		\bigl(1-\Pi_{kj}^{(t^\ast-1)}\bigr)
		\right\}\\
		&\qquad-
		\frac{\alpha_{ki}^{(t)}}{\psi_{Y,i}^{(t)}}
		\bigl(\Pi_{ki}^{(t-1)}-\Pi_{ki}^{(t^\ast-1)}\bigr)
		+
		\frac{\alpha_{kj}^{(t)}}{\psi_{Y,j}^{(t)}}
		\bigl(\Pi_{kj}^{(t-1)}-\Pi_{kj}^{(t^\ast-1)}\bigr).
		\end{aligned}
		\]
		The term in braces is the baseline part. Since
		\[
		\frac{\alpha_{ki}^{(t)}}{\psi_{Y,i}^{(t)}}
		=
		\frac{\alpha_{kj}^{(t)}}{\psi_{Y,j}^{(t)}}
		=
		\rho_Y\psi_{Y,k}^{(t)}
		\left(\boldsymbol\theta_{Y,\nu^{(t)}(k)}^{(t)}\right)^\top
		\boldsymbol\theta_{Y,a}^{(t)},
		\]
		its \(k\)th coordinate, after division by \(\sqrt{\rho_Y}\), is
		\(
		\sqrt{\rho_Y}\,\psi_{Y,k}^{(t)}
		\left(\boldsymbol\theta_{Y,\nu^{(t)}(k)}^{(t)}\right)^\top
		\boldsymbol\theta_{Y,a}^{(t)}
		\left(
		\Pi_{kj}^{(t^\ast-1)}-\Pi_{ki}^{(t^\ast-1)}
		\right).
		\)
		Therefore, using \(\rho_Y\le1\), \(\psi_{Y,k}^{(t)}\le\ell\),
		and \(\|\boldsymbol\theta_{Y,b}^{(t)}\|_2=1\), its off-diagonal
		coordinates, together with the two completed diagonal coordinates
		\(k=i,j\), are bounded by
		\[
		\begin{aligned}
		&\left[
		\frac{\rho_Y}{p c_{Y,S}}
		\sum_{k\notin\{i,j\}}
		\bigl(\psi_{Y,k}^{(t)}\bigr)^2
		\left\{
		\left(\boldsymbol\theta_{Y,\nu^{(t)}(k)}^{(t)}\right)^\top
		\boldsymbol\theta_{Y,a}^{(t)}
		\right\}^2
		\left(
		\Pi_{kj}^{(t^\ast-1)}-\Pi_{ki}^{(t^\ast-1)}
		\right)^2
		\right]^{1/2}
		+\frac{C}{\sqrt p}\\
		&\quad\le
		C\left(
		\omega_{\Pi,S,t^\ast-1}^{(t)}+p^{-1/2}
		\right).
		\end{aligned}
		\]
		
		For the first memory term, the definition of
		\(M_{Y,i}^{(t^*-1,t)}\) gives
		\[
		\begin{aligned}
		\frac{1}{\sqrt{p\rho_Yc_{Y,S}}}
		\left[
		\sum_{k\ne i}
		\left\{
		\frac{\alpha_{ki}^{(t)}}{\psi_{Y,i}^{(t)}}
		\bigl(\Pi_{ki}^{(t-1)}-\Pi_{ki}^{(t^\ast-1)}\bigr)
		\right\}^2
		\right]^{1/2}
		&=
		\frac{M_{Y,i}^{(t^*-1,t)}}
		{\psi_{Y,i}^{(t)}\sqrt{c_{Y,S}}}.
		\end{aligned}
		\]
		The second memory term is bounded in the same way with \(i\) replaced
		by \(j\); the exceptional diagonal coordinates are already included
		in the \(p^{-1/2}\) term. Hence
		\[
		\begin{aligned}
		\frac{
		\left\|
		\mathbf S_Y^{(t)}
		\left(
		\frac{\mathbf e_i}{\psi_{Y,i}^{(t)}}
		-
		\frac{\mathbf e_j}{\psi_{Y,j}^{(t)}}
		\right)
		\right\|_2
		}{\sqrt{p\rho_Yc_{Y,S}}}
		&\le
		C\left(
		\omega_{\Pi,S,t^\ast-1}^{(t)}+p^{-1/2}
		\right)
		+
		\frac{M_{Y,i}^{(t^*-1,t)}}
		{\psi_{Y,i}^{(t)}\sqrt{c_{Y,S}}}
		+
		\frac{M_{Y,j}^{(t^*-1,t)}}
		{\psi_{Y,j}^{(t)}\sqrt{c_{Y,S}}}.
		\end{aligned}
		\]
		The same argument for the \(Z\) channel, using
		\(S_{Z,ki}^{(t)}=\beta_{ki}^{(t)}\Pi_{ki}^{(t-1)}\), gives
		\[
		\begin{aligned}
		\frac{
		\left\|
		\mathbf S_Z^{(t)}
		\left(
		\frac{\mathbf e_i}{\psi_{Z,i}^{(t)}}
		-
		\frac{\mathbf e_j}{\psi_{Z,j}^{(t)}}
		\right)
		\right\|_2
		}{\sqrt{p\rho_Zc_{Z,S}}}
		&\le
		C\left(
		\omega_{\Pi,S,t^\ast-1}^{(t)}+p^{-1/2}
		\right)
		+
		\frac{M_{Z,i}^{(t^*-1,t)}}
		{\psi_{Z,i}^{(t)}\sqrt{c_{Z,S}}}
		+
		\frac{M_{Z,j}^{(t^*-1,t)}}
		{\psi_{Z,j}^{(t)}\sqrt{c_{Z,S}}}.
		\end{aligned}
		\]
		Combining these two bounds by the triangle inequality in
		\(\mathbb R^2\) yields
		\[
		\begin{aligned}
		&\left[
		\frac{
		\left\|
		\mathbf S_Y^{(t)}
		\left(
		\frac{\mathbf e_i}{\psi_{Y,i}^{(t)}}
		-
		\frac{\mathbf e_j}{\psi_{Y,j}^{(t)}}
		\right)
		\right\|_2^2
		}{p\rho_Yc_{Y,S}}
		+
		\frac{
		\left\|
		\mathbf S_Z^{(t)}
		\left(
		\frac{\mathbf e_i}{\psi_{Z,i}^{(t)}}
		-
		\frac{\mathbf e_j}{\psi_{Z,j}^{(t)}}
		\right)
		\right\|_2^2
		}{p\rho_Zc_{Z,S}}
		\right]^{1/2}\\
		&\quad\le
		C\left[
		\omega_{\Pi,S,t^\ast-1}^{(t)}+p^{-1/2}
		+
		2\max_{r\in[p]}
		\left\{
		\frac{\{M_{Y,r}^{(t^*-1,t)}\}^2}
		{c_{Y,S}(\psi_{Y,r}^{(t)})^2}
		+
		\frac{\{M_{Z,r}^{(t^*-1,t)}\}^2}
		{c_{Z,S}(\psi_{Z,r}^{(t)})^2}
		\right\}^{1/2}
		\right].
		\end{aligned}
		\]
		Substitution into \eqref{eq:eta_memory_channel_reduction} proves
		the claim.
	\end{proof}
\section*{S.4 Algorithms}

The theoretical threshold \(\tau_p=\gamma\epsilon_p\) contains population
constants. The following procedures retain the main-paper statistics but
construct empirical reference distributions from pairs of stationary baseline
transitions.

\begin{algorithm}[H]
\caption{Baseline stationarity check}
\label{alg:baseline-calib}
\begin{enumerate}
\item[] \textbf{Input:} Transition-event right embeddings over the candidate baseline
\(B\); false discovery rate \(\alpha\); and, for community analysis, \(q_B\)
and a prespecified community object.
\item \textbf{For each \(r\in B\), do:}
  \item Set \(B_{-r}=B\setminus\{r\}\).
  \item Compute
  \[
  \overline T_G^{(r;B_{-r})}
  =\frac{1}{|B|-1}\sum_{b\in B_{-r}}T_G^{(b,r)},\qquad
  \overline T_i^{(r;B_{-r})}
  =\frac{1}{|B|-1}\sum_{b\in B_{-r}}T_i^{(b,r)}.
  \]
  \item Set
  \[
  p_G^{\rm stationary}(r)
  =
  \frac{
  1+\#\{(a,b):a<b,\ a,b\in B_{-r},\
  T_G^{(a,b)}\geq\overline T_G^{(r;B_{-r})}\}}
  {1+\binom{|B|-1}{2}}.
  \]
  \item\textbf{For each \(i\in[p]\), do:}
    \item Set
    \[
    p_i^{\rm stationary}(r)
    =
    \frac{
    1+\#\{(a,b,j):a<b,\ a,b\in B_{-r},\ j\in[p],\
    T_j^{(a,b)}\geq\overline T_i^{(r;B_{-r})}\}}
    {1+p\binom{|B|-1}{2}}.
    \]
  \item\textbf{End for.}
  \item For community analysis, estimate the partition from \(B_{-r}\),
  compute
  \(\overline T_C^{(r;B_{-r})}
  =(|B|-1)^{-1}\sum_{b\in B_{-r}}T_C^{(b,r)}\), and set
  \[
  p_C^{\rm stationary}(r)
  =
  \frac{
  1+\#\{(u,v):u<v,\ u,v\in B_{-r},\
  T_C^{(u,v)}\geq\overline T_C^{(r;B_{-r})}\}}
  {1+\binom{|B|-1}{2}}.
  \]
\item \textbf{End for.}
\item Apply the Benjamini--Hochberg adjustment separately to the network-level,
vertex-level, and selected community-level families.
\item[] \textbf{Output:} Raw and adjusted \(p\)-values and a pass/fail stationarity decision.
\end{enumerate}
\end{algorithm}

\begin{algorithm}[H]
\caption{Empirical \(p\)-values for anomaly detection}
\label{alg:baseline-pair-calibration}
\begin{enumerate}
\item[] \textbf{Input:} Transition-event right embeddings over the analysis window \(W\);
stationary baseline \(B\); target \(t\); level \(\alpha\); and, for community
analysis, \(q_B\) and a prespecified community object.
\item Compute
\[
\overline T_G^{(t;B)}
=|B|^{-1}\sum_{b\in B}T_G^{(b,t)}
\]
and evaluate \(p_G(t)\) by \eqref{eq:network-empirical-p}.
\item \textbf{For each \(i\in[p]\), do:}
  \item Compute
  \(\overline T_i^{(t;B)}
  =|B|^{-1}\sum_{b\in B}T_i^{(b,t)}\)
  and evaluate \(p_i(t)\) by \eqref{eq:vertex-empirical-p}.
\item \textbf{End for.}
\item For community analysis, estimate the common partition from \(B\),
compute
\(\overline T_C^{(t;B)}
=|B|^{-1}\sum_{b\in B}T_C^{(b,t)}\),
and evaluate \(p_C(t)\) by \eqref{eq:community-empirical-p}.
\item Flag the target network, vertex \(i\), or selected community object
when its pointwise empirical \(p\)-value is below \(\alpha\).
\item[] \textbf{Output:} Empirical \(p\)-values and pointwise anomaly flags.
\end{enumerate}
\end{algorithm}

\section*{S.5 Additional simulation studies}

\subsection*{S.5.1 Dynamic community recovery}

We generate an AR(1)-DCSBM with \(p=300\), \(d=6\), embedding dimension
\(D=10\), \(n=25\), and \(\rho_Y=\rho_Z=0.7\). Degree parameters are drawn
independently from \(\operatorname{Unif}(0.20,0.80)\). The number of
communities follows
\(
2\rightarrow3\rightarrow4\rightarrow3\rightarrow2
\)
with changes at \(t=5,10,15,20\). The community directions evolve as
\[
\boldsymbol\theta_k^{(t)}
=
\frac{\mathbf e_k+a_t\mathbf 1_6}
{\|\mathbf e_k+a_t\mathbf 1_6\|_2},
\qquad
a_t=0.06\{1+\sin(2\pi t/25)\},
\]
and memberships evolve according to
\[
\nu^{(t)}(i)
=1+\left\lfloor
q_t\left\{\left(\frac{i-1}{p}+\frac{0.4t}{25}\right)\bmod1\right\}
\right\rfloor .
\]
Thus the vertices move around a unit circle while the number and separation of
communities change. Applying \(q_t\)-means to the normalised joint event
embedding gives an adjusted Rand index between \(0.89\) and \(1.00\) over the
representative trajectory. Recovery is weaker when \(a_t\) is large because
the community directions are less separated.

\begin{figure}[t]
\centering
\includegraphics[width=\textwidth,height=0.75\textheight,keepaspectratio]
{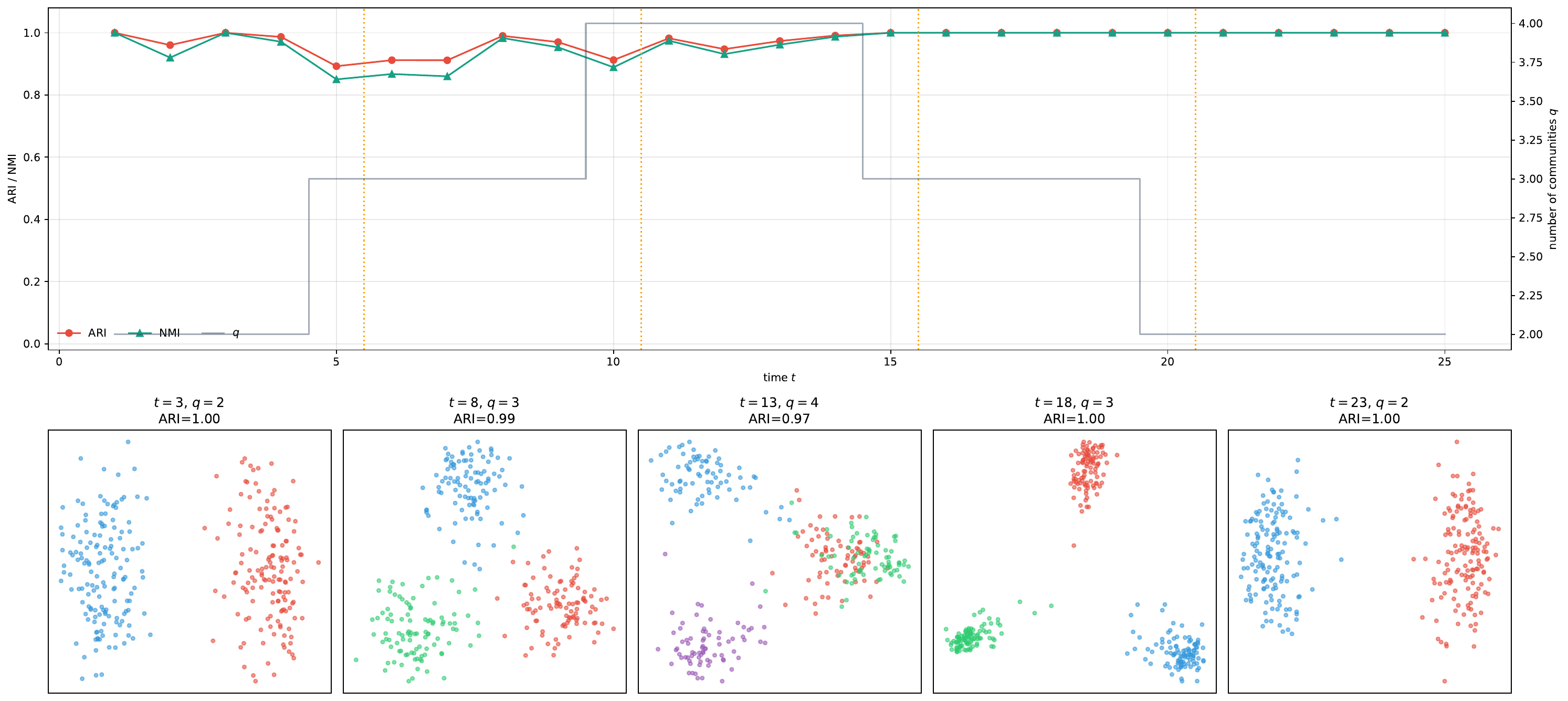}
\caption{Dynamic community recovery. Top: adjusted Rand index and normalised
mutual information over time as the community count follows
\(2\to3\to4\to3\to2\). Bottom: the channelwise-normalised joint embedding at
representative times.}
\label{fig:sim-community-recovery-main}
\end{figure}

\subsection*{S.5.2 Vertex anomalies with and without community reassignment}

We generate an AR(1)-DCSBM with \(p=300\), three balanced communities,
20 transition layers, \(\rho_Y=\rho_Z=0.20\), and \(D=6\). Baseline degree
parameters are drawn independently from \(\operatorname{Unif}(0.60,0.70)\); before row
normalisation, a community direction has an entry of one in its own coordinate and
0.27 in each off-community coordinate. At \(t=15\), a fraction
\(\gamma\in\{0.02,0.04,0.06,0.10\}\) of vertices is anomalous. Half move
cyclically to the next community and half retain their original community.
All anomalous vertices change their channel degree parameters to
\(\psi_{Y,i}^{(15)}=2.19\) and \(\psi_{Z,i}^{(15)}=0.40\).

We use the baseline \(B=\{1,\ldots,12\}\) and 100 Monte Carlo replications. As
\(\gamma\) increases from \(0.02\) to \(0.10\), the vertex detection rate
increases from \(0.917\) to \(0.995\) for reassigned anomalies and from
\(0.883\) to \(0.988\) for same-community degree anomalies. The true-negative
rate for normal vertices decreases from \(0.985\) to \(0.919\), reflecting the
increasing cross-vertex interference.

\begin{figure}[t]
\centering
\includegraphics[width=\textwidth]{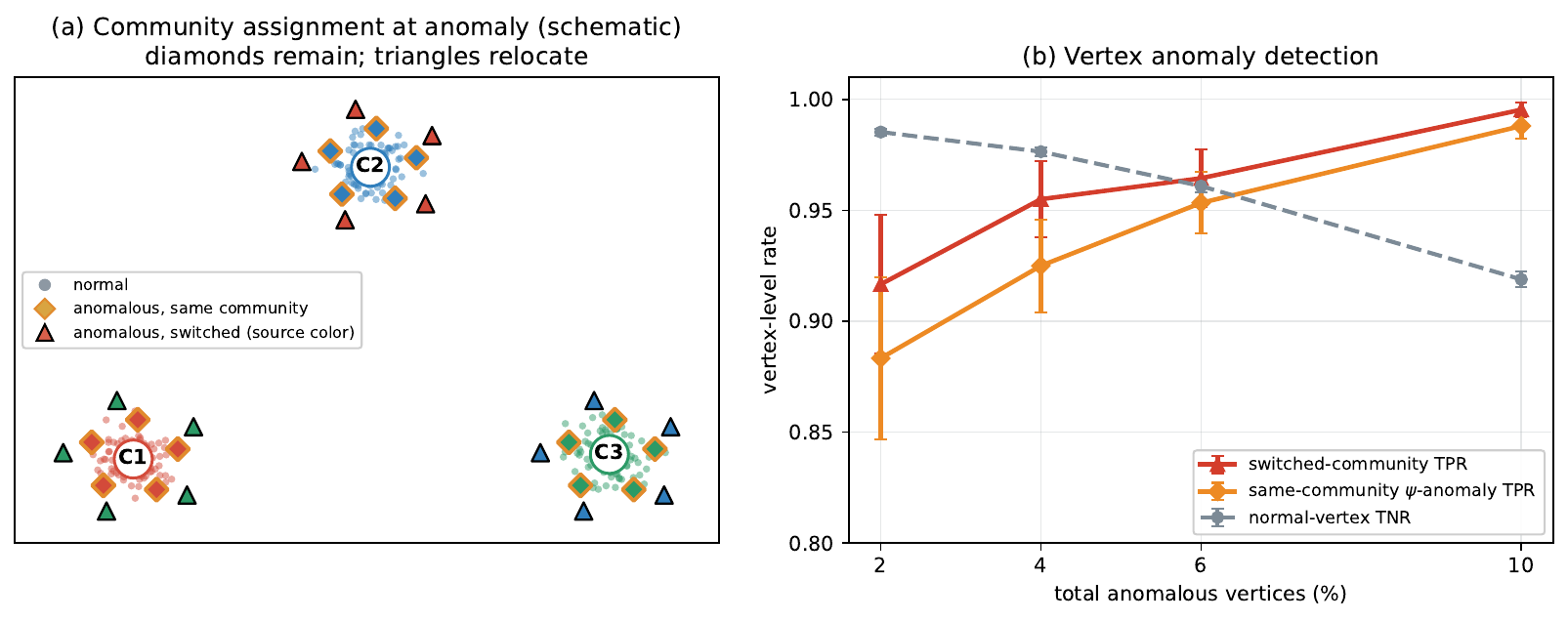}
\caption{Vertex anomalies with and without community reassignment. Left:
anomaly-time assignments when 10\% of the vertices are anomalous; diamonds
remain in their original communities and triangles move cyclically to the next
community. Right: detection rates for reassigned and same-community anomalies
and the true-negative rate for normal vertices.}
\label{fig:sim-vertex-mixed-community}
\end{figure}

\clearpage
\section*{S.6 Additional COVID-19 trade results}

	\subsection*{S.6.1 Country coverage}

	The January 2017--December 2020 window contains \(152\) reporters with
	at least one total-trade export or import record.  After removing the
	European Union and ASEAN aggregate reporters, \(150\) country or area
	reporters remain, of which \(104\) have at least one such record in every
	one of the \(48\) months.  Using ISO 3166-1 alpha-3 codes, the
	continuous \(104\)-reporter panel is AND, AGO, ATG, AZE, ARG, AUS, ARM,
	BRB, BEL, BOL, BIH, BWA, BRA, BLZ, BRN, BGR, MMR, BDI, BLR, KHM, CAN,
	CHL, CHN, COL, COD, HRV, CYP, BEN, DNK, DOM, ECU, SLV, FJI, FIN, FRA,
	PYF, GEO, GMB, PSE, DEU, GHA, GRC, GRD, GUY, HKG, HUN, ISL, IDN, IRL,
	ISR, ITA, JPN, KAZ, KOR, LAO, LSO, LVA, LTU, LUX, MAC, MDG, MWI, MYS,
	MUS, MEX, MNG, MDA, MSR, MOZ, NLD, NZL, NOR, PAK, PAN, PRY, PHL, POL,
	PRT, QAT, ROU, RUS, RWA, STP, SEN, SRB, SYC, IND, SGP, SVK, VNM, SVN,
	ESP, SWE, CHE, THA, TGO, TTO, TUR, MKD, EGY, GBR, USA, BFA, and ZMB.

	\subsection*{S.6.2 Sensitivity analysis}

	Table~\ref{tab:trade-robustness} examines robustness to the fixed dollar edge
	threshold, embedding dimension, and baseline window in the same \(104\)-economy
	panel.  The April-2020 dissolution is significant in every reported
	specification.  At US\(\$14\)--\(15\) million, the June-2020
	formation signal is also significant for each \(D\in\{6,8,10\}\).  Raising
	the threshold to US\(\$16\) million weakens that recovery signal, indicating
	that it depends more strongly than the April collapse on links near the
	selected cutoff.  Across baseline windows, April remains significant.  June
	is recovered by every window of \(18\) or more transitions but not
	by the shortest, 12-transition 2019-only baseline.  The primary 35-transition window uses all available pre-COVID transitions and also passes the BH-adjusted stationarity checks.

	\begin{table}[H]
		\centering
		\small
		\setlength{\tabcolsep}{3pt}
		\begin{tabular}{llcc}
			\hline
			\multicolumn{4}{l}{\emph{Panel A. Fixed edge threshold \(\times\) embedding dimension; baseline 2017.02--2019.12}}\\
			Threshold & \(D\) & Apr. 2020 \(\mathbf Z\) & Jun. 2020 \(\mathbf Y\)
			\\\hline
			US\(\$14\)m & 6  & \,0.002$^\ast$ & \,0.018$^\ast$ \\
			US\(\$14\)m & 8  & \,0.002$^\ast$ & \,0.005$^\ast$ \\
			US\(\$14\)m & 10 & \,0.002$^\ast$ & \,0.010$^\ast$ \\
			US\(\$15\)m & 6  & \,0.002$^\ast$ & \,0.020$^\ast$ \\
			US\(\$15\)m & 8  & \,0.002$^\ast$ & \,0.030$^\ast$ \\
			US\(\$15\)m & 10 & \,0.002$^\ast$ & \,0.034$^\ast$ \\
			US\(\$16\)m & 6  & \,0.002$^\ast$ & \,0.049$^\ast$ \\
			US\(\$16\)m & 8  & \,0.002$^\ast$ & \,0.062 \\
			US\(\$16\)m & 10 & \,0.002$^\ast$ & \,0.060 \\
			\hline
		\end{tabular}

		\vspace{1.2ex}
		\begin{tabular}{llcc}
			\hline
			\multicolumn{4}{l}{\emph{Panel B. Baseline window; US\(\$15\) million, \(D=8\)}}\\
			Baseline window & \# transitions & Apr. 2020 \(\mathbf Z\) & Jun. 2020 \(\mathbf Y\)
			\\
			\hline
			2019.01--2019.12 & 12 & \,0.015$^\ast$ & \,0.119 \\
			2018.01--2019.06 & 18 & \,0.006$^\ast$ & \,0.045$^\ast$ \\
			2017.04--2018.12 & 21 & \,0.005$^\ast$ & \,0.005$^\ast$ \\
			2017.04--2019.06 & 27 & \,0.003$^\ast$ & \,0.017$^\ast$ \\
			2017.02--2019.12 & 35 & \,0.002$^\ast$ & \,0.030$^\ast$ \\
			\hline
		\end{tabular}
		\caption{Specification sensitivity of the network-level
			COVID-trade detection.  Network-level \(p\)-values are shown for the
			April-2020 dissolution (\(\mathbf Z\)) and the June-2020
			formation (\(\mathbf Y\)) episode; \(^\ast\) marks values
			below \(0.05\).  These target-month \(p\)-values are pointwise and are not
			BH-adjusted.}
		\label{tab:trade-robustness}
	\end{table}
\end{document}